\documentclass[11pt, letterpaper]{article} 

\usepackage[margin=1in]{geometry}
\usepackage{amsmath}
\usepackage{comment}
\usepackage{amsthm}
\usepackage{amssymb}
\usepackage[ruled,vlined,linesnumbered]{algorithm2e}
\usepackage{enumitem}
\usepackage{graphicx}
\usepackage{subcaption}
\usepackage{xspace}
\usepackage{url}
\usepackage{framed}
\usepackage{authblk}
\usepackage{esvect}
\usepackage{mathscinet}
\usepackage[export]{adjustbox}
\usepackage{wrapfig}
\usepackage[font={small}]{caption}

\usepackage{pdfpages}
\usepackage{enumitem}

\usepackage{thmtools} 
\usepackage{thm-restate}
\allowdisplaybreaks

\newcommand{\mydef}{:=}

\newcommand{\poly}{\ensuremath{\mathcal{P}}\xspace}

\newcommand{\LL}{\ensuremath{\mathcal{L}}\xspace}
\newcommand{\RR}{\ensuremath{\mathbb R}\xspace}
\newcommand{\QQ}{\ensuremath{\mathbb Q}\xspace}

\newcommand{\ee}{\ensuremath{\varepsilon}\xspace}
\newcommand{\ER}{\ensuremath{\exists \mathbb{R}}\xspace}
\newcommand{\ETR}{\ensuremath{\textrm{ETR}}\xspace}
\newcommand{\NP}{\ensuremath{\textrm{NP}}\xspace}
\newcommand{\PSPACE}{\ensuremath{\textrm{PSPACE}}\xspace}
\newcommand{\etr}{\ensuremath{\textrm{ETR}}\xspace}
\newcommand{\etrinv}{\ensuremath{\textrm{ETR-INV}}\xspace}

\newcommand{\etrpar}[1]{\ensuremath{\textrm{ETR}^{#1,+,\cdot}}\xspace}
\newcommand{\etal}{et al.}

\newcommand{\nook}{nook}
\newcommand{\anook}{a nook}

\newcommand{\aemphnook}{a \emph{nook}}
\newcommand{\umbra}{umbra}
\newcommand{\anumbra}{an umbra}

\newcommand{\anemphumbra}{an \emph{umbra}}

\graphicspath{{./figures/},{../figures/}}

\newtheorem{theorem}{Theorem}
\newtheorem{corollary}[theorem]{Corollary}
\newtheorem{lemma}[theorem]{Lemma}

\newtheorem{observation}[theorem]{Observation}
\newtheorem{definition}[theorem]{Definition}

\usepackage{xcolor}
\newcounter{note}[section] 

\title{The Art Gallery Problem is $\exists \mathbb{R}$-complete\thanks{To appear at the 50th ACM Symposium on Theory of Computing (STOC 2018).}}
    
\author[1]{Mikkel Abrahamsen}
\author[1]{Anna Adamaszek}
\author[2]{Tillmann Miltzow}

\affil[1]{University of Copenhagen, Denmark.
\texttt{\{miab,anad\}@di.ku.dk}}
\affil[2]{ Universit\'e libre de Bruxelles (ULB), Brussels, Belgium. \texttt{t.miltzow@gmail.com}}

\date{}     
     
\begin{document}

\maketitle
\thispagestyle{empty}

     %


\begin{abstract}
We prove that the \emph{art gallery problem} is equivalent under polynomial time reductions to deciding whether a system of polynomial equations over the real numbers has a solution.
The art gallery problem is a classic problem in computational geometry, introduced in 1973 by Victor Klee. Given a simple polygon \poly and an integer $k$, the goal is to decide if there exists a set $G$ of $k$ \emph{guards} within \poly such that every point $p\in \poly$ is seen by at least one guard $g\in G$. Each guard corresponds to a point in the polygon \poly, and we say that a guard $g$ \emph{sees} a point $p$ if the line segment $pg$ is contained in \poly.

The art gallery problem has stimulated extensive research in 
geometry and in algorithms. However, the complexity status of 
the art gallery problem has not been resolved. It has long been 
known that the problem is NP-hard, but no one has been able to 
show that it lies in NP. Recently, the computational geometry 
community became more aware of the complexity class \ER,
which has been studied earlier by other communities.
 The class \ER consists of problems that can be reduced in polynomial 
 time to the problem of deciding whether a system of polynomial 
 equations with integer coefficients and any number of real variables 
 has a solution. It can be easily seen that $\NP\subseteq \ER$.
We prove that the art gallery problem is \ER-complete, implying 
that (1) any system of polynomial equations over the real numbers 
can be encoded as an instance of the art gallery problem, and (2) 
the art gallery problem is not in the complexity class \NP unless $\NP=\ER$.
As a corollary of our construction, we prove that for any real algebraic 
number $\alpha$, there is an instance of the art gallery problem 
where one of the coordinates of the guards equals $\alpha$ in any 
guard set of minimum cardinality.
That rules out many natural geometric approaches to the problem, as 
it shows that any approach based on constructing a finite set of 
candidate points for placing guards has to include points with 
coordinates being roots of polynomials with arbitrary degree.
As an illustration of our techniques, we show that for every
compact semi-algebraic set $S\subseteq [0,1]^2$, there exists a polygon with corners at
rational coordinates such that for every $p\in[0,1]^2$, there is a set of guards of minimum cardinality containing $p$ if and only if $p\in S$.

In the \ER-hardness proof for the art gallery problem, we introduce a new \ER-complete problem $\etrinv$. We believe that this problem is of independent interest, as it can be used to obtain \ER-hardness proofs for other problems.
\end{abstract}
     
\newpage

\tableofcontents

\newpage

\section{Introduction}\label{sec:intro}
 
\subsection{The art gallery problem}

Given a simple polygon \poly, we say that two points $p,q\in \poly$ \emph{see each other} if the line segment $pq$ is contained in \poly.
 A set of points $G\subseteq \poly$ is said to \emph{guard} the polygon \poly if every point $p\in \poly$ is seen by at least one guard $g\in G$.
 Such a set $G$ is called a \emph{guard set} of \poly, and the points of $G$ are called \emph{guards}. A guard set of \poly is \emph{optimal} if it is a minimum cardinality guard set of \poly.
 
 In the \emph{art gallery problem} we are given an integer $g$ and a polygon $\poly$ with corners at rational 
 coordinates, and the goal is to decide if \poly has a guard set of cardinality $g$.
 We consider
 a polygon as a Jordan curve consisting of finitely many line segments \emph{and} the region that it encloses.
The art gallery problem has been introduced in 1973 by Victor Klee, and it has stimulated extensive research in geometry and in algorithms.
However, the complexity status of the art gallery problem has stayed unresolved.
We are going to prove that the problem is \ER-complete, as defined below.

\subsection{The complexity class \ER}
The \emph{first order theory of the reals} is a set of all true
sentences involving real variables, universal and existential quantifiers, boolean and arithmetic operators, constants $0$ and $1$, parenthesis, equalities and inequalities, i.e., the alphabet is the set
$$\left\{X_1,X_2,\ldots, \forall, \exists, \land,\lor,\lnot, 0 ,1 ,+ ,- ,\cdot,\allowbreak\ (\ ,\ )\ ,=,<,\leq\right\}.$$
A formula is called a \emph{sentence} if it has no free variables, so that each variable present in the formula is bound by a quantifier.
Note that within such formulas, one can easily express integer constants (using binary expansion) and powers.
Each formula can be converted to a \emph{prenex form}, which means that it starts with all the quantifiers and is followed by a quantifier-free formula.
Such a transformation changes the length of the formula by at most a constant factor.

The \emph{existential theory of the reals} is a set of all true sentences of the first-order theory of the reals in prenex form with existential quantifiers only, i.e., sentences of the form 
$$(\exists X_1 \exists X_2 \ldots \exists X_n)\enspace \Phi(X_1,X_2,\ldots,X_n),$$ 
where $\Phi$ is a quantifier-free formula of the first-order theory of the reals with variables $X_1,\ldots,X_n$.
The problem \ETR is the problem of deciding whether a given existential formula of the above form is true.
 The complexity class \ER consists of  all problems that are reducible to \ETR in polynomial time. It is currently known that \[ \NP \subseteq \ER \subseteq \PSPACE.\]
 
It is not hard see that the problem \ETR is \NP-hard, yielding the first inclusion.
The containment \ER $\subseteq$ \PSPACE is highly non-trivial, and it has first been established by Canny~\cite{canny1988some}.
In order to compare the complexity classes $\NP$ and $\ER$, we suggest the reader to consider the following two problems.
The problem of deciding whether a given polynomial equation $Q(x_1,\ldots,x_n)=0$ with integer coefficients has a solution with all variables restricted to $\{0,1\}$ is easily seen to be \NP-complete.
On the other hand, if the variables are merely restricted to $\RR$, the problem is \ER-complete~\cite[Proposition 3.2]{matousek2014intersection}.

  \subsection{Our results and their implications}
  We prove that solving the art gallery problem is, up to a polynomial time reduction, as hard as deciding whether a system of polynomial equations and inequalities over the real numbers has a solution.
  
\begin{restatable}{theorem}{ThmMain}\label{thm:main}
    The art gallery problem is \ER-complete, even the restricted 
    variant where we are given a polygon with 
    corners at integer coordinates.
  \end{restatable} 

It is a classical result in Galois theory, and has thus been known since the 19th century, that there are polynomial equations of degree five with integer coefficients which have real solutions, but with no solutions expressible by radicals (i.e., solutions that can be expressed using integers, addition, subtraction, multiplication, division, raising to integer powers, and the extraction of $n$'th roots).
One such example is the equation $x^5-x+1=0$~\cite{wiki:quintic}.
It is a peculiar fact that using the reduction described in this paper, we are able to transform such an equation into an instance of the art gallery problem where no optimal guard set can be expressed by radicals. More generally, we can prove the following.

\begin{restatable}{theorem}{CorAlgebraicNumbers}
\label{cor:AlgebraicNumbers}
Given any real algebraic number $\alpha$, there exists a polygon $\poly$ with corners at rational coordinates such that in any optimal guard set of $\poly$ there is a guard with an $x$-coordinate equal $\alpha$.
\end{restatable}

This is a generalization of our work~\cite{abrahamsen2017irrational}, where we showed that irrational guards are sometimes needed in optimal guard sets.
 Our results justify the difficulty in constructing algorithms 
 for the art gallery problem, and explain the lack of 
 combinatorial algorithms for the problem (see the subsequent 
 summary of related work). In particular, Theorem~\ref{cor:AlgebraicNumbers} 
 rules out many algorithmic approaches to solving the art gallery problem.
A natural approach to finding a guard set for a given 
polygon \poly is to create a candidate set for the guards, 
and select a guard set as a subset of the candidate set.
For instance, a candidate set can consist of the corners of \poly.
The candidate set can then be expanded by considering all lines 
containing two candidates and adding all intersection points 
of these lines to the candidate set.
This process can be repeated any finite number of times, but 
only candidates with rational coordinates can be obtained that 
way, and the candidate set will thus not contain
 an optimal guard set in general.
Algorithms of this kind are discussed for instance by de 
Rezende et al.~\cite{engineering}.
One can get a more refined set of candidates by also 
considering certain quadratic curves~\cite{belleville1991computing}, 
or more complicated curves.
Our results imply that if the algebraic degree of the considered 
curves is bounded by a constant, such an approach cannot lead to 
an optimal solution in general, since the coordinates of the 
candidates will also have algebraic degree bounded by a constant.

A \emph{semi-algebraic set} is a set of the form $\{\mathbf x\in\RR^n\colon \Phi(\mathbf x)\}$, where $\Phi$ is a quantifier-free formula of the first-order theory of the reals with $n$ variables.
For many \ER-complete problems, there is a deep connection 
between their solution spaces and semi-algebraic sets. 
The most famous result of this kind is Mn{\"e}v's 
universality theorem~\cite{mnev1988universality}.
It yields 
that for every semi-algebraic set $S$
there exists a pseudoline arrangement $P$ such that the 
space $\mathcal{L}(P)$ (of all line arrangements 
topologically equivalent to $P$) is homotopy equivalent 
to $S$. We show a similar correspondence for the art 
gallery problem, see Theorem~\ref{thm:Correspondance} 
and Theorem~\ref{thm:StrongCorrespondance} in Section~\ref{sec:hardness}. 
Moreover, we can show the following result.

\begin{restatable}[Picasso Theorem]{theorem}{PicassoTheorem}
\label{thm:Picasso}
For any compact semi-algebraic set $S \subset [0,1]^2$,  
there is a polygon $\poly_S$ with corners at rational coordinates 
such that for any point  $p \in [0,1]^2$ we have $p \in S$ if 
and only if there exists an optimal guard set $G$ of $\poly_S$ with $p\in G$.  
\end{restatable}

The name of the last theorem stems from the following 
imaginative interpretation. We are given any black 
and white picture ($\approx$ semi-algebraic set), 
and we construct a special art gallery with this 
picture drawn at the floor.
The theorem says that we can guard the gallery optimally 
if and only if one of the guards stands on one 
of the black points of the picture.

 \subsection{Related work}
The art gallery problem has been extensively studied, with some books, surveys, and book chapters dedicated to it~\cite{de2000computational, devadoss2011discrete, matouvsek2002lectures, o1987art, o1998computational, 2004handbook, shermer1992recent, urrutia2000art}.
The research is stimulated by a large number of possible variants of the problem and related questions that can be studied.
The version of the art gallery problem considered in this paper is the one originally formulated by Victor Klee (see O'Rourke \cite{o1987art}). Other versions of the art gallery problem include restrictions on the positions of the guards, different definitions of visibility, restricted classes of polygons, restricting the part of the polygon that has to be guarded, etc.

The art gallery problem has been studied both from the combinatorial and from the algorithmic perspective.
Studies have been made on algorithms performing well in practice on real-world and simulated instances of the problem~\cite{DBLP:conf/compgeom/BorrmannRSFFKNST13, engineering}.
Another branch of research investigates approximation algorithms for the art gallery problem and its variants~\cite{BonnetM16, eidenbenz2001inapproximability, ApproXKirkpatrick15}.

The first exact algorithm for solving the art gallery problem was published in 2002 in the conference version of a paper by Efrat and Har-Peled~\cite{DBLP:journals/ipl/EfratH06}.
They attribute the result to Micha Sharir.
Before that time, the problem was not even known to be decidable.
The algorithm computes a formula 
in the first order theory of the reals corresponding to the art gallery instance, and uses standard algebraic 
methods, such as the techniques 
provided by Basu et al.~\cite{basu1996combinatorial}, to decide if the formula is true.
No algorithm is known that avoids the use of this powerful machinery.
The formula has both existential and universal quantifiers.

Lee and Lin~\cite{DBLP:journals/tit/LeeL86} proved, by constructing a reduction from $\textrm{3SAT}$, that the art gallery problem is \NP-hard when the guards are restricted to the corners of the polygon (note that this version, called the \emph{vertex-guard art gallery problem}, is obviously in \NP).
Lee and Lin's paper also contains a proof that the point-guard art gallery problem (which we consider in the present paper, i.e., where the guards can be anywhere in the polygon) is $\NP$-hard.
The latter result is due to Aggarwal~\cite{aggarwal1984art} according to O'Rourke~\cite{o1987art}.
Various papers showed other hardness results or conditional lower bounds for the art gallery problem and its variations~\cite{BonnetM16, broden2001guarding, eidenbenz2001inapproximability, TerrainGuardingNP, MonotonHard, laurentini1999guarding, DBLP:journals/tit/LeeL86, o1983some, DBLP:journals/mlq/SchuchardtH95, tomas2013guarding}.

A problem related to the art gallery problem is the \emph{terrain guarding problem}.
  Here, the area above an $x$-monotone polygonal curve $c$ has to be 
  guarded by a minimum number of guards restricted to $c$.
Friedrichs \etal~\cite{DiscretizeTerrain} showed recently that terrain guarding is $\NP$-complete.

The authors of the present paper~\cite{abrahamsen2017irrational} provided a simple 
instance of the art gallery problem with a unique optimal guard set consisting 
of three guards at points with irrational coordinates.
Any guard set using only points with rational coordinates requires at least four guards.
This could be an indication that the art gallery 
problem is actually more difficult than the related problems discussed above.
Friedrichs \etal~\cite{DiscretizeTerrain} stated that
``[\ldots] it is a long-standing open problem for the more
general Art Gallery Problem (AGP): For the AGP it is not known whether the coordinates
of an optimal guard cover can be represented with a polynomial number of bits''.
In the present paper we prove that under the assumption $\NP \neq \ER$ the art gallery problem is not in $\NP$, and such a representation does not exist.

A growing class of problems turn out to be equivalent (under polynomial time reductions) 
to deciding whether polynomial equations and inequalities over the reals have a solution.
These problems form the family of $\ER$-complete problems as it is currently known.
Although the name \ER was established not so long ago~\cite{matouvsek2002lectures, DBLP:journals/mst/SchaeferS17},
algorithms for deciding $\ETR$, which is the defining problem for \ER, 
have long been studied, see e.g.~the book of Basu et al.~\cite{basu2006algorithms}.
The class \ER includes problems like the stretchability of 
pseudoline arrangements~\cite{mnev1988universality,shor1991stretchability},
recognition of intersection graphs of various objects 
(e.g.~segments~\cite{matousek2014intersection}, unit disks~\cite{mcdiarmid2013integer}, 
and general convex sets~\cite{Schaefer2010}),
recognition of point visibility graphs~\cite{cardinal2017recognition},
the Steinitz problem for $4$-polytopes~\cite{richter1995realization},
deciding whether a graph with given edge lengths can be realized by 
a straight-line drawing~\cite{abel, schaefer2013realizability},
deciding whether a graph has a straight line drawing with a given number of edge crossings~\cite{bienstock1991some},
decision problems related to Nash equilibria~\cite{garg2015etr},
positive semidefinite matrix factorization~\cite{Shitov16a}, and nonnegative
matrix factorization~\cite{shitov2016universality}.
We refer the reader to the lecture notes by 
Matou{\v{s}}ek~\cite{matousek2014intersection} and 
surveys by Schaefer~\cite{Schaefer2010} and 
Cardinal~\cite{Cardinal:2015:CGC:2852040.2852053} for more 
information on the complexity class $\ER$.

\subsection{Overview of the paper and techniques}
In Section~\ref{sec:membership}, we show that the art gallery problem is in the complexity class \ER. For that we present a construction of an \ETR-formula $\Phi$ for any instance $(\poly,g)$ of the art gallery problem such that $\Phi$ has a solution if and only if \poly has a guard set of size $g$. The idea is to encode guards by pairs of variables and compute a set of witnesses (which depend on the positions of the guards) of polynomial size such that the polygon is guarded if and only if all witnesses are seen by the guards. 

The proof that the art gallery problem is \ER-hard is the main result of the paper, and it consists of two parts. The first part is of an algebraic nature, and in that we introduce a novel \ER-complete problem which we call \etrinv. A common way of making a reduction from \ETR to some other problem is to build gadgets corresponding to each of the equations $x=1$, $x+y = z$, and $x\cdot y = z$ for any variables $x,y,z$. Usually, the multiplication gadget is the most involved one. An instance of \etrinv is a conjunction of formulas of the form $x=1$, $x+y=z$, and $x\cdot y=1$, with the requirement that each variable must be in the interval $[1/2,2]$. In particular, the reduction from \etrinv requires building a gadget for \emph{inversion} (i.e., $x\cdot y=1$), which involves only two variables, instead of a more general gadget for multiplication involving three variables. The formal definition of \etrinv and the proof that it is \ER-complete is presented in Section~\ref{sec:etrV}. We believe that the problem \etrinv might be of independent interest, and that it will allow constructing \ER-hardness proofs for other problems, in particular those for which constructing a multiplication gadget was an obstacle that could not be overcome. \etrinv has already been used to prove  \ER-completeness of a geometric graph drawing problem with prescribed face areas~\cite{AreasKleist}, and to prove $\ER$-completeness of completing a partially (straight-line) drawn  graph~\cite{AnnaPreparation}.

In Section~\ref{sec:hardness}, we describe a polynomial time reduction from \etrinv to the art gallery problem, which shows that the art gallery problem is \ER-hard. This reduction constructs an art gallery instance $(\poly(\Phi), g(\Phi))$ from an \etrinv instance $\Phi$, such that $\poly(\Phi)$ has a guard set of size $g(\Phi)$ if and only if the formula $\Phi$ has a solution. We construct the polygon so that it contains $g(\Phi)$ \emph{guard segments} (which are horizontal line segments within $\poly$) and \emph{stationary guard positions} (points within \poly). By introducing \emph{pockets} we enforce that if \poly has a guard set of size $g(\Phi)$, then there must be exactly one guard at each guard segment and at each stationary guard position. Each guard segment represents a variable of $\Phi$ (with multiple segments representing the same variable) in the sense that the position of the guard on the segment specifies the value of the variable, the endpoints of a segment corresponding to the values $1/2$ and $2$.

We develop a technique for \emph{copying} guard segments, i.e., enforcing that guards at two segments correspond to the same variable. We do that by introducing \emph{critical segments} within the polygon, which can be seen by guards from two guard segments (but not from other guard segments). 
Then the requirement that a critical segment is seen introduces dependency between the guards at the corresponding segments. Different critical segments enforce different dependencies, and by enforcing that two guards together see two particular critical segments we ensure that the guards represent the same value. The stationary guards are placed to see the remaining areas of the polygon.

With this technique, we are able to copy two or three 
segments from an area containing guard segments 
corresponding to all variables into a \emph{gadget}, 
where we will enforce a dependency between the values 
of the variables represented by the two or three 
segments. This is done by constructing a \emph{corridor} 
containing two critical segments for each pair of 
copied segments. The construction is technically 
demanding, as it requires the critical segments 
not to be seen from any other segments.

The gadgets have features that enforce the variables $x,y,z$ represented by the guards to satisfy one of the conditions $x+y\geq z$, $x+y\leq z$, or $x\cdot y=1$. The conditions are enforced by a requirement that two or three guards can together see one or more regions in the gadget, where each region is a line segment or a quadrilateral.

In Section~\ref{sec:Picasso}, we prove the Picasso Theorem (Theorem~\ref{thm:Picasso}).

Finally, in Section~\ref{sec:concluding}, we conclude the paper by sketching how to avoid some kinds of degeneracies in the polygon we construct from a given instance of $\etrinv$ (collinear triples of corners and triples of edges with linear extensions intersecting each other at the same point), thus deducing that the art gallery problem is $\ER$-complete even under some general position assumptions.
We furthermore suggest a few open problems for future research.



\section{The art gallery problem is in $\ER$}\label{sec:membership}

In this section we will prove that the art gallery problem is in the complexity class $\ER$.
Our proof works also for a more general version of the art gallery problem where the input polygon can have polygonal holes. 



\begin{theorem}\label{thm:ermembership}
The art gallery problem is in the complexity class $\ER$.
\end{theorem}

\begin{proof}
Let $(\poly,g)$ be an instance of the art gallery problem where the polygon \poly has $n$ corners, each of which has rational coordinates represented by at most $B$ bits.
We show how to construct a quantifier-free formula $\Phi\mydef \Phi(\poly,g)$ of the first-order theory of the reals such that $\Phi$ is satisfiable if and only if $\poly$ has a guard set of cardinality $g$. The formula $\Phi$ has length $O(g^7n^7B^2)=O(n^{14}B^2)$ and can be computed in polynomial time. It has been our priority to define the formula $\Phi$ so that it is as simple as possible to describe. It might be possible to construct an equivalent but shorter formula.

The description of $\Phi$ is similar to the formula $\Psi$ that Micha Sharir described to Efrat and Har-Peled~\cite{DBLP:journals/ipl/EfratH06}:
\[\Psi \;\mydef \; \left[\exists x_1,y_1,\ldots, x_k,y_k\; \forall p_x,p_y :
\textrm{INSIDE-POLYGON}(p_x,p_y)\Longrightarrow \bigvee_{i=1}^k \ \textrm{SEES}(x_i,y_i,p_x,p_y)\right].\]
For each $i\in\{1,\ldots,k\}$, the variables $x_i,y_i$ represent the position of guard $g_i\mydef (x_i,y_i)$, and $p \mydef  (p_x,p_y)$ represents an arbitrary point.
The predicate $\textrm{INSIDE-POLYGON}(p_x,p_y)$ tests if the point $p$ is contained in the polygon \poly, and $\textrm{SEES}(x_i,y_i,p_x,p_y)$ checks if the guard $g_i$ can see the point $p$. Thus, the formula is satisfiable if and only if there is a guard set of cardinality $g$.
Note that although the implication ``$\Longrightarrow$'' is not allowed in the first order theory of the reals, we can always substitute ``$A\Longrightarrow B$'' by ``$\lnot A\lor B$''.

For the purpose of self-containment, we will briefly repeat the construction of the predicates $\textrm{INSIDE-POLYGON}(p_x,p_y)$ and $\textrm{SEES}(x_i,y_i,p_x,p_y)$. 
The elementary tool is evaluation of the sign of the determinant $\det(\vv u,\vv v)$ of two vectors $\vv u,\vv v$.
Recall that the sign of the expression $\det(\vv{u}, \vv{v})$ determines whether $\vv{v}$ points to the left of $\vv{u}$ (if $\det(\vv{u}, \vv{v}) > 0$), is parallel to $\vv{u}$ (if $\det(\vv{u}, \vv{v})=0$), or points to the right of $\vv{u}$ (if $\det(\vv{u}, \vv{v})< 0$).

We compute a triangulation $\mathcal{T}$ of the polygon \poly, e.g., using an algorithm from~\cite{de2000computational}, order the corners of each triangle of $\mathcal{T}$ in the counter-clockwise order, and orient each edge of the triangle accordingly.
A point is contained inside the polygon if and only if it is contained in one of the triangles of $\mathcal{T}$.
A point is contained in a triangle if and only if it is on one of the edges or to the left of each edge.
Thus the predicate $\textrm{INSIDE-POLYGON}(p_x,p_y)$ has length $O(nB)$.

A guard $g_i$ sees a point $p$ if and only if no two consecutive edges of \poly block the visibility. See Figure~\ref{fig:etrMembershipRegions} on why it is not sufficient to check each edge individually.
Given a guard $g_i$, a point $p$, and two consecutive edges $e_1,e_2$ of \poly, it can be checked by evaluating a constant number of determinants whether $e_1,e_2$ block the visibility between $g_i$ and $p$.
Thus $\textrm{SEES}(x_i,y_i,p_x,p_y)$ has length $O(nB)$ and consequently $\bigvee_{i=1}^k \ \textrm{SEES}(x_i,y_i,p_x,p_y)$ has length $O(knB)$.

Note that the formula $\Psi$ is not a formula in $\ETR$ because of the universal quantifier.
The main idea to get an equivalent formula with no universal quantifier is to find a polynomial number of points inside \poly which, if all seen by the guards, ensure that all of \poly is seen.
We denote such a set of points as a \emph{witness set}.

\begin{figure}[htbp]
\centering
\includegraphics[clip, trim=1cm 1cm 1cm 1cm]{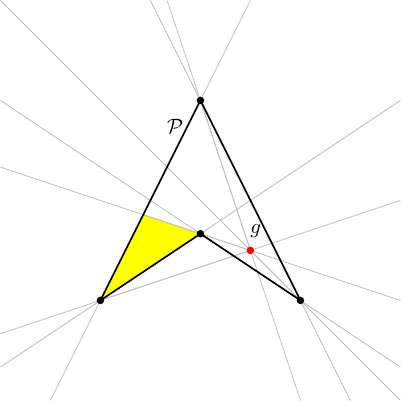}
\hspace{0.2\textwidth}
\includegraphics{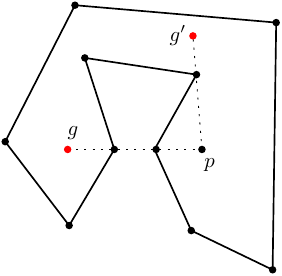}
\caption{Illustrations for the proof that the art gallery problem is in $\ER$. Left: A polygon \poly together with a guard $g$ and the resulting line arrangement \LL. The yellow region is not seen. Right: A line segment $gp$ going through the corners of \poly. This illustrates that it is not sufficient to check each edge individually for crossing. A corner of $\poly$ is on the line segment $g'p$, but the visibility between $p$ and $g'$ is not blocked in that case.}
\label{fig:etrMembershipRegions}
\end{figure} 
  
\vskip5pt
\noindent\textbf{Creating a witness set.}
We are now ready to describe the witness set that replaces the universal quantifier.
Let $\mathcal L\mydef\{\ell_1,\ldots,\ell_m\}$ be the set of lines containing either an edge of \poly, or a guard $g \in G$ and a corner $v \in P$.\footnote{Note that if a line is defined as passing through a guard $g \in G$ and a corner $v \in P$ such that $g$ and $v$ are coincident, the line is not well-defined.
Such lines are not considered in the partition into regions described below, but they are included in the set $\mathcal L$.
Later we will show how to ignore these lines in our formula.}
The well-defined lines in $\mathcal L$ divide the plane into \emph{regions}, which are connected components of $\RR^2\setminus \bigcup_{\ell\in \mathcal L}\ell$ (see Figure \ref{fig:etrMembershipRegions} for an example).

Let $\mathcal A$ be the set of these regions.
It is easy to verify that $\mathcal A$ has the following properties:
\begin{itemize}
\item
Each region in $\mathcal A$ is an open convex polygon.

\item
Each region in $\mathcal A$ is either contained in \poly or contained in the complement of \poly.

\item
The closure of the union of the
regions that are contained in \poly equals \poly.

\item
For each region $R\in\mathcal A$, each guard $g\in G$ either sees all points of $R$ or sees no point in $R$.
In particular, if a guard $g$ sees one point in $R$, it sees all of $R$ and its closure.
\end{itemize}
Thus it is sufficient to test that for each region $R\in \mathcal A$ which is in \poly, at least one point in $R$ is seen by a guard.
For three points $a,b,c$, define the \emph{centroid} of $a,b,c$ to be the point $C(a,b,c)\mydef (a+b+c)/3$.
If $a,b,c$ are three different corners of the same region $R\in\mathcal A$, then the centroid $C(a,b,c)$ must lie in the interior of $R$.
Note that each region has at least three corners and thus contains at least one such centroid.
Let $X$ be the set of all intersection points between two non-parallel well-defined lines in $\mathcal L$, i.e., $X$ consists of all corners of all regions in $\mathcal A$.
In the formula $\Phi$, we generate all points in $X$.
For any three points $a,b,c$ in $X$, we also generate the centroid $C(a,b,c)$.
If the centroid is in \poly, we check that it is seen by a guard.
Since there are $O(((kn)^2)^3)$ centroids of three points in $X$ and each is tested by a formula of size $O(knB)$, we get a formula of the aforementioned size.

\vskip5pt
\noindent\textbf{Constructing the formula $\Phi$.}
Each line $\ell_i$ is defined by a pair of points $\{ (p_i,q_i)$, $(p'_i,q'_i) \}$.
Let $\vv{\ell_i}\mydef (p'_i-p_i,q'_i-q_i)$ be a direction vector corresponding to the line.
A line $\ell$ is well-defined if and only if the corresponding vector $\vv{\ell}$ is non-zero.

The lines $\ell_i,\ell_j$ are well-defined and non-parallel if and only if $\det(\vv{\ell_i},\vv{\ell_j}) \neq 0$.
If two lines $\ell_i,\ell_j$ are well-defined and non-parallel, their intersection point $X^{ij}$ is well-defined and it has coordinates 
\begin{align*}
\left( \frac{(p_jq'_j - q_jp'_j)(p'_i - p_i)-(p_iq'_i - q_ip'_i)(p'_j - p_j)}{\det(\vv{\ell_i},\vv{\ell_j})},\frac{ (p_jq'_j - q_jp'_j)(q'_i - q_i)-(p_iq'_i - q_ip'_i)(q'_j - q_j)}{\det(\vv{\ell_i},\vv{\ell_j})} \right).
\end{align*}

For each pair $(i,j)\in\{1,\ldots,m\}^2$, we add the variables $x_{ij},y_{ij}$ to the formula $\Phi$ and we define $\textrm{INTERSECT}(i,j)$ to be the formula
\begin{align*}
\det(\vv{\ell_i},\vv{\ell_j}) \neq 0\; \Longrightarrow
\Big[& \det(\vv{\ell_i},\vv{\ell_j})\cdot x_{ij} = (p_jq'_j - q_jp'_j)(p'_i - p_i)-(p_iq'_i - q_ip'_i)(p'_j - p_j)\; \land\\
& \det(\vv{\ell_i},\vv{\ell_j})\cdot y_{ij} = (p_jq'_j - q_jp'_j)(q'_i - q_i)-(p_iq'_i - q_ip'_i)(q'_j - q_j)\Big].
\end{align*}
It follows that if the formula $\textrm{INTERSECT}(i,j)$ is true then either
\begin{itemize}
\item $\ell_i$ or $\ell_j$ is not well-defined or they are both well-defined, but parallel, or
\item $\ell_i$ and $\ell_j$ are well-defined and non-parallel and the variables $x_{ij}$ and $y_{ij}$ are the coordinates of the intersection point $X^{ij}$ of the lines.
\end{itemize}

Let $\Lambda\mydef \{\lambda_1,\ldots,\lambda_{m^6}\}=\{1,\ldots,m\}^6$ be all the tuples of six elements from the set $\{1,\ldots,m\}$.
Each tuple $\lambda\mydef (a,b,c,d,e,f)\in\Lambda$ corresponds to a centroid of the following three points: the intersection point of the lines $\ell_a,\ell_b$, the intersection point of the lines $\ell_c,\ell_d$, and the intersection point of the lines $\ell_e,\ell_f$. 
For each tuple $\lambda$, we proceed as follows. We define the formula $\textrm{CENTROID-DEFINED}(\lambda)$ to be
\[ \det(\vv{\ell_a},\vv{\ell_b})\neq 0\; \land\; \det(\vv{\ell_c},\vv{\ell_d})\neq 0\; \land\; \det(\vv{\ell_e},\vv{\ell_f})\neq 0. \]
We add the variables $u_{\lambda},v_\lambda$ to the formula $\Phi$, and define the formula $\textrm{CENTROID}(\lambda)$ as
\begin{align*}
& 3u_{\lambda} = x_{ab}+x_{cd}+x_{ef}\;\land\; 3v_{\lambda} = y_{ab}+y_{cd}+y_{ef}.
\end{align*}
It follows that if the formulas $\textrm{CENTROID-DEFINED}(\lambda)$ and $\textrm{CENTROID}(\lambda)$ are both true, then the lines in each of the pairs $(\ell_a,\ell_b),(\ell_c,\ell_d),(\ell_e,\ell_f)$ are well-defined and non-parallel, and the variables $u_\lambda$ and $v_\lambda$ are the coordinates of the centroid $C(X^{ab},X^{cd},X^{ef})$.

We are now ready to write up our existential formula as  
$$\exists x_1,y_1,\ldots, x_k,y_k \ \
\exists x_{11},y_{11},x_{12},y_{12},\ldots,x_{mm},y_{mm} \ \
\exists u_{\lambda_1},v_{\lambda_1},\ldots,u_{\lambda_{m^6}},v_{\lambda_{m^6}} : \Phi,$$
where
\begin{align*}
\Phi \mydef& \left[ \bigwedge_{(i,j)\in\{1,\ldots,m\}^2} \textrm{INTERSECT}(i,j)\; \right] \land\;\\
& \bigwedge_{\lambda\in\Lambda} \Bigg[\textrm{CENTROID-DEFINED}(\lambda)\;
\Longrightarrow\;\\
& \bigg[\textrm{CENTROID}(\lambda)\;\land\;
\Big[\textrm{INSIDE-POLYGON}(u_{\lambda},v_{\lambda})\Longrightarrow
\bigvee_{i=1}^k \ \textrm{SEES}(x_i,y_i,u_{\lambda},v_{\lambda})\Big]\bigg]\Bigg]. \qedhere
\end{align*}
\end{proof}


\section{The problem $\etrinv$}\label{sec:etrV}

To show that the art gallery problem is $\ER$-hard, we will provide a reduction from the problem $\etrinv$, which we introduce below.
In this section, we will show that $\etrinv$ is $\ER$-complete. 

\begin{restatable}[\etrinv]{definition}{EtrinvDef}
\label{def:etrinv}
In the problem $\etrinv$, we are given a set of real 
variables $\{x_1,\ldots,x_n\}$, and a set of equations of the form
$$
x=1,\quad x+y=z,\quad x\cdot y=1,
$$
for $x,y,z \in \{x_1, \ldots, x_n\}$.
The goal is to decide whether the system of equations has 
a solution when each variable is restricted to the range $[1/2,2]$.
\end{restatable}

If $\Phi(\mathbf x)$ is an instance of $\etrinv$ with variables $\mathbf x\mydef (x_1,\ldots,x_n)$, the space of solutions $S_\Phi\mydef\{\mathbf x \in [1/2,2]^n : \Phi(\mathbf x)\}$ consists of the vectors in $[1/2,2]^n$ that satisfy all the equations of $\Phi$.
In one sense, the range $[1/2,2]$ is the simplest possible:
The range of course needs to contain $1$, since we have equations of the form $x=1$.
In order to include just one more integer, namely $2$, we also need to include $1/2$ since we have equations of the form $x\cdot y=1$.

In order to show that $\etrinv$ is $\ER$-complete, we make use of the following problem.


%
\begin{definition}
In the problem $\etrpar{c}$, where $c \in \mathbb{R}$, we are given a set of real variables $\{x_1,\ldots,x_n\}$, and a set of equations of the form
$$
x=c,\quad x+y=z,\quad x\cdot y=z,
$$
for $x,y,z \in \{x_1, \ldots, x_n\}$. The goal is to decide whether the system of equations has a solution.

A modified version of the problem, where we additionally require that $x_1,\ldots,x_n\in[a,b]$ for some $a,b \in \RR$, is 
denoted by $\etrpar{c}_{[a,b]}$.
\end{definition}


We are now ready to prove that $\etrinv$ is $\ER$-complete.

\begin{restatable}{theorem}{ThmERComplete}\label{thm:etrinv}
The problem $\etrinv$ is $\ER$-complete.
\end{restatable}


\begin{proof}
To show that $\etrinv$ is $\ER$-hard, we will perform a series of polynomial time reductions, starting from $\etr$ and subsequently reducing it to the problems \etrpar{1}, $\etrpar{1/8}_{[-1/8,1/8]}$, and $\etrpar{1}_{[1/2,2]}$, and ending with $\etrinv$.

To simplify the notation, while considering a problem $\etrpar{c}$ or $\etrpar{c}_{[a,b]}$, we might substitute any variable in an equation by the constant $c$. For instance, $x+c=z$ is a shorthand for the equations $x+y=z$ and $y=c$, where $y$ is an additional variable.

\vskip5pt
\textbf{Reduction to $\etrpar{1}$.}
We will first argue that \etrpar{1} is \ER-hard. This seems to be folklore, but we did not find a formal statement. 
%
For the sake of self-containment and rigorousness, we present here a short proof based on the following lemma.
\begin{lemma}[Schaefer, \v{S}tefankovi\v{c} \cite{DBLP:journals/mst/SchaeferS17}]\label{lem:FEASABLE}
Let $\Phi(\mathbf x)$ be a quantifier-free formula of the first order theory of the reals, where $\mathbf x \mydef (x_1,x_2,\ldots,x_n)$ is a vector of variables.
We can construct in polynomial time a polynomial $F : \RR^{n+m} \to \RR$ of degree 4, for some $m = O(|\Phi|)$, so that
\[
\{\mathbf x  \in \RR^n  :  \Phi(\mathbf x) \}  = \{ \mathbf x \in \RR^n   : (\exists \mathbf y \in \RR^m) \; F(\mathbf x,\mathbf y) =0\}.\]
The coefficients of $F$ have bitlength $O(|\Phi|)$.
\end{lemma}

Thus it is \ER-hard to decide if a polynomial has a real root. We reduce this problem to \etrpar{1}.
Consider a polynomial equation $Q=0$.
First, we generate all variables corresponding to all the coefficients of $Q$, by using only the constant $1$, addition and multiplication. For example, a variable corresponding to $20$ can be obtained as follows: $V_1=1, V_{2^1}=1+1, V_{2^2}=V_{2^1}\cdot V_{2^1}, V_{2^4}=V_{2^2}\cdot V_{2^2}, V_{20}=V_{2^4}+V_{2^2}$.
We are now left with a polynomial $Q'$ consisting entirely of sums of products of variables, and we keep simplifying $Q'$ as described in the following.
Whenever there is an occurrence of a sum $x+y$ or a product $x\cdot y$ of two variables in $Q'$, we introduce a new variable $z$.
In the first case, we add the equation $x+y=z$ to $\Phi$ and substitute the term $x+y$ by $z$ in $Q'$.
In the latter case, we add the equation $x\cdot y=z$ to $\Phi$ and substitute $x\cdot y$ by $z$ in $Q'$.
We finish the construction when $Q'$ has been simplified to consist of a single variable, i.e., $Q'=x$, in which case we add the equation $x+V_1=V_1$ (corresponding to the equation $Q'=0$) to $\Phi$.
When the process finishes, $\Phi$ yields an instance of $\etrpar{1}$, and the solutions to $\Phi$ are in one-to-one correspondence with the solutions to the original polynomial equation $Q=0$.

\vskip5pt
\textbf{Reduction to $\etrpar{1/8}_{[-1/8,1/8]}$.}
We will now present a reduction from the problem $\etrpar{1}$ to $\etrpar{1/8}_{[-1/8,1/8]}$. 
%
%
We use the following result from algebraic geometry, which was stated by Schaefer and \v{S}tefankovi\v{c} \cite{DBLP:journals/mst/SchaeferS17} in a simplified form.
Given an instance $\Phi(\mathbf x)$ of $\ETR$
over the vector of variables $\mathbf x\mydef (x_1,\ldots,x_n)$, we define the semi-algebraic set $S_\Phi$ as the solution space
\[S_\Phi \mydef \{ \mathbf x\in \RR^n : \Phi(\mathbf x)\}.\]
The complexity $L$ of a semi-algebraic set $S_\Phi$ is defined as the number of symbols appearing in the formula $\Phi$ defining $S_\Phi$ (see~\cite{matousek2014intersection}).

\begin{corollary}[Schaefer and \v{S}tefankovi\v{c}~ \cite{DBLP:journals/mst/SchaeferS17}]\label{cor:ball}
Let $B$ be the set of points in $\RR^n$ at distance at most $2^{L^{8n}} = 2^{2^{8n\log L}}$ from the origin.
Every non-empty semi-algebraic set $S$ in $\RR^n$ of complexity at most $L \geq 4$
contains a point in $B$.
\end{corollary}

Let $\Phi$ be an instance of $\etrpar{1}$ with $n$ variables $x_1,\ldots,x_n$.
We construct an instance $\Phi'$ of $\etrpar{1/8}_{[-1/8,1/8]}$ such that $\Phi$ has a solution if and only
if $\Phi'$ has a solution.
Let us fix $k \mydef \lceil 8n\cdot \log L + 3\rceil$ and $\ee \mydef 2^{-2^{k}}$.
In $\Phi'$, we first define a variable $V_{\ee}$ satisfying
$V_\ee=\ee$, using $\Theta(k)$ new variables $V_{1/{2^{2^2}}}, V_{1/{2^{2^3}}}, \ldots, V_{1/{2^{2^{k}}}}$  and equations
\begin{align*}
V_{1/{2^{2^2}}}+V_{1/{2^{2^2}}} &= 1/8 ,\\
V_{1/{2^{2^2}}}\cdot V_{1/{2^{2^2}}} &= V_{1/{2^{2^3}}} ,\\
V_{1/{2^{2^3}}}\cdot V_{1/{2^{2^3}}} &= V_{1/{2^{2^4}}} ,\\
&\vdots \\
V_{1/{2^{2^{k-1}}}}\cdot V_{1/{2^{2^{k-1}}}} 
&= V_{\ee}. 
\end{align*}

In $\Phi'$, we use the variables $V_{\ee x_1},\ldots,V_{\ee x_n}$ instead of
$x_1,\ldots,x_n$.
An equation of $\Phi$ of the form $x=1$ is transformed to the equation
$V_{\ee x}=V_\ee$ in $\Phi'$.
An equation of $\Phi$ of the form $x+y=z$ is transformed to the equation
$V_{\ee x} + V_{\ee y} = V_{\ee z}$ of $\Phi'$.
An equation of $\Phi$ of the form $x\cdot y=z$ is transformed to the
following equations of $\Phi'$, where $V_{\ee^2 z}$ is a new variable satisfying
\begin{align*}
V_{\ee x} \cdot V_{\ee y} & = V_{\ee^2 z} ,\\
V_\ee\cdot V_{\ee z} & = V_{\ee^2 z}.
\end{align*}

Assume that $\Phi$ is true. Then there exists an assignment of
values to the variables $x_1,\ldots,x_n$ of $\Phi$ that satisfies all the equations
and where each variable $x_i$ satisfies
$|x_i|\in\left[0,2^{2^{8n\log L} }\right]$. Then 
the assignment $V_{\ee x_i}=\ee x_i$ and (when $V_{\ee^2 x_i}$ appears in
$\Phi'$) $V_{\ee^2 x_i} = \ee^2 x_i$ yields a solution to $\Phi'$ with all variables in the range $[-1/8,1/8]$.
On the other hand, if there is a solution to $\Phi'$, an analogous argument yields a corresponding solution to $\Phi$.
We have given a reduction from $\etrpar{1}$ to $\etrpar{1/8}_{[-1/8,1/8]}$.
The length of the formula increases by at most a polylogarithmic factor.

\vskip5pt
\textbf{Reduction to $\etrpar{1}_{[1/2,2]}$.}
We will now show a reduction from $\etrpar{1/8}_{[-1/8,1/8]}$ to $\etrpar{1}_{[1/2,2]}$. 
The reduction is similar to the one in~\cite{shor1991stretchability}.
We substitute each variable $x_i \in [-1/8,1/8]$ by $V_{x_i+7/8}$ which will be assumed to have a value of $x_i+7/8$. Instead of an equation $x=1/8$ we now have $V_{x+7/8}=1$. Using addition and the variable equal $1$, we can easily get the variables $V_{1/2}, V_{3/2}, V_{3/4}, V_{7/4}, V_{7/8}$  with corresponding values of $1/2, 3/2, 3/4, 7/4$, and $7/8$. Instead of each equation $x+y=z$ we now have equations:
\begin{align*}
 V_{x+7/8} + V_{y+7/8} & = V_{(z+7/8)+7/8} ,\\ 
 V_{z+7/8} + V_{7/8} & = V_{(z+7/8)+7/8}. 
 \end{align*}
 As the original variables $x,y,z$ have values in the interval $[-1/8,1/8]$, the added variables $V_{(z+7/8)+7/8}$ have a value in $[13/8,15/8]$.

Instead of each equation $x \cdot y=z$ we have the following set of equations
\begin{align*}
V_{x+7/8} + V_{y+7/8} &= V_{x+y+14/8} ,\ \ (V_{x+y+14/8} \in [12/8,2])\\
V_{x+y+7/8} + V_{7/8} &= V_{x+y+14/8} ,\ \ (V_{x+y+7/8} \in [5/8,9/8])\\
V_{x+7/8} + V_{1} &= V_{x+15/8} ,\ \ (V_{x+15/8} \in [14/8,2])\\
V_{x+1} + V_{7/8} &= V_{x+15/8} ,\ \ (V_{x+1} \in [7/8,9/8])\\
V_{y+7/8} + V_{1} &= V_{y+15/8} ,\ \ (V_{y+15/8} \in [14/8,2])\\
V_{y+1} + V_{7/8} &= V_{y+15/8} ,\ \ (V_{y+1} \in [7/8,9/8])\\
V_{x+1} \cdot V_{y+1} &= V_{xy+x+y+1} ,\ \ (V_{xy+x+y+1} \in [49/64,81/64])\\
V_{xy+x+y+1} + V_{1/2} &= V_{xy+x+y+3/2} ,\ \ (V_{xy+x+y+3/2} \in [81/64,113/64])\\
V_{xy+5/8} + V_{x+y+7/8} &= V_{xy+x+y+3/2} ,\ \ (V_{xy+5/8} \in [39/64,41/64])\\
V_{xy+5/8} + 1 &= V_{xy+13/8} ,\ \ (V_{xy+13/8} \in [103/64,105/64])\\
V_{z+7/8} + V_{3/4} &= V_{xy+13/8} .
\end{align*}


Each formula $\Phi$ of $\etrpar{1/8}_{[-1/8,1/8]}$ is transformed to a formula $\Phi'$ of $\etrpar{1}_{[1/2,2]}$, as explained above. If there is a solution for $\Phi$, there clearly is a solution for $\Phi'$, as all the newly introduced variables have a value within the intervals claimed above. If there is a solution for $\Phi'$, there is also a solution for $\Phi$, as the newly introduced variables $V_{2x+7/4}, V_{x+5/8} \in [1/2,2]$ ensure that $x \in [-1/8,1/8]$. The increase in the length of the formula is linear.

Note that the only place where we use multiplication is in the formula $V_{x+1} \cdot V_{y+1} = V_{xy+x+y+1}$, where $V_{x+1}, V_{y+1} \in [7/8,9/8]$. We will use this fact in the next step of the reduction.

\vskip5pt
\textbf{Reduction to $\etrinv$.}
We will now show that $\etrpar{1/8}_{[-1/8,1/8]}$ reduces to $\etrinv$. In the first step, we reduce a formula $\Phi$ of $\etrpar{1/8}_{[-1/8,1/8]}$ to a formula $\Phi'$ of $\etrpar{1}_{[1/2,2]}$, as described in the step above. In the sequel we have to express each equation $x\cdot y=z$ of $\Phi'$ using only the equations allowed in $\etrinv$. Note that, as explained above, multiplication is used only for variables $x, y \in [7/8,9/8]$.

Some of the steps in this reduction rely on techniques also used in the proof by Aho \etal~\cite[Section~8.2]{e_k1974design} that squaring and taking reciprocals is equivalent to multiplication.
In particular, note that for $x\notin \{0,1\}$, we have 
$\frac 1{x-1}-\frac 1x=\frac 1{x^2-x}$.
Therefore, a variable $V_{x^2}$ satisfying $V_{x^2}=x^2$ can be constructed 
from $x$ using only a sequence of additions and inversions. 
Similarly, as $(x+y)^2-x^2-y^2=2xy$, a variable $V_{xy}$ 
satisfying $V_{xy}=xy$ can be constructed from $x$ and $y$ 
using a sequence of additions and squarings. 
Our main contribution is to bound the ranges of variables at 
all stages of the reduction, thus ensuring that in 
the final construction all variables stay in the range $[1/2,2]$.

We first show how to define a new variable $V_{x^2}$ satisfying $V_{x^2}=x^2$, where $x \in [7/8,9/8]$.
\begin{align*}
x+V_{3/4} &= V_{x+3/4} ,\ \ (V_{x+3/4} \in [13/8,15/8])\\
V_{1/(x+3/4)}\cdot V_{x+3/4}&=1 ,\ \ (V_{1/(x+3/4)} \in [8/15,8/13])\\
V_{x-1/4}+1&= V_{x+3/4} ,\ \ (V_{x-1/4} \in [5/8,7/8])\\
V_{1/(x-1/4)}\cdot V_{x-1/4}&=1 ,\ \ (V_{1/(x-1/4)} \in [8/7,8/5])\\
V_{1/(x^2+x/2-3/16)} + V_{1/(x+3/4)} &=V_{1/(x-1/4)} ,\ \ (V_{1/(x^2+x/2-3/16)} \in [64/105,64/65])\\
V_{1/(x^2+x/2-3/16)}\cdot V_{x^2+x/2-3/16}&=1 ,\ \ (V_{x^2+x/2-3/16} \in [65/64,105/64])\\
x+V_{7/8} &= V_{x+7/8} ,\ \ (V_{x+7/8} \in [14/8,2])\\
V_{x+1/8}+V_{3/4} &= V_{x+7/8} ,\ \ (V_{x+1/8} \in [1,10/8])\\
V_{x/2+1/16}+V_{x/2+1/16} &= V_{x+1/8} ,\ \ (V_{x/2+1/16} \in [1/2,10/16])\\
V_{x^2-1/4} + V_{x/2+1/16} &= V_{x^2+x/2-3/16} ,\ \ (V_{x^2-1/4} \in [33/64,65/64])\\
V_{x^2-1/4} +V_{3/4} &= V_{x^2+1/2} ,\ \ (V_{x^2+1/2} \in [81/64,113/64])\\
V_{x^2}+V_{1/2}&= V_{x^2+1/2} .
\end{align*}
Note that the constructed variables are in the range $[\frac{1}{2},2]$.
In the following, as shorthand for the construction given above,
we allow to use equations of the form $x^2=y$, for a variable $x$ with a value in $[7/8,9/8]$.
We now describe how to express an equation $x \cdot y = z$, where $x,y \in [7/8,9/8]$.
\begin{align*}
x+V_{7/8} &=V_{x+7/8} ,\ \ (V_{x+7/8}\in [14/8,2])\\
V_{(x+7/8)/2}+V_{(x+7/8)/2}&=V_{x+7/8} ,\ \ (V_{(x+7/8)/2}\in [14/16,1])\\
y+V_{7/8} &=V_{y+7/8} ,\ \ (V_{y+7/8}\in [14/8,2])\\
V_{(y+7/8)/2}+V_{(y+7/8)/2}&=V_{y+7/8} ,\ \ (V_{(y+7/8)/2}\in [14/16,1])\\
V_{(x+7/8)/2}+V_{(y+7/8)/2}&=V_{(x+y)/2+7/8} ,\ \ (V_{(x+y)/2+7/8}\in [14/8,2])\\
V_{(x+y)/2}+V_{7/8}& =V_{(x+y)/2+7/8} ,\ \ (V_{(x+y)/2}\in [7/8,9/8])\\
V_{(x+y)/2}^2 & = V_{((x+y)/2)^2} ,\ \ (V_{((x+y)/2)^2}\in [49/64,81/64])\\
V_{((x+y)/2)^2}+V_{1/2}&=V_{((x+y)/2)^2+1/2} ,\ \ (V_{((x+y)/2)^2+1/2}\in [81/64,113/64])\\
x^2 &= V_{x^2} , \ \ (V_{x^2}\in [49/64,81/64])\\
y^2 &= V_{y^2} , \ \ (V_{y^2}\in [49/64,81/64])\\
V_{x^2}+V_{1/2} &= V_{x^2+1/2} , \ \ (V_{x^2+1/2}\in [81/64,113/64])\\
V_{x^2/2+1/4}+V_{x^2/2+1/4}&= V_{x^2+1/2} , \ \ (V_{x^2/2+1/4}\in [81/128,113/128])\\
V_{y^2}+V_{1/2} &= V_{y^2+1/2} , \ \ (V_{y^2+1/2}\in [81/64,113/64])\\
V_{y^2/2+1/4}+V_{y^2/2+1/4}&= V_{y^2+1/2} , \ \ (V_{y^2/2+1/4}\in [81/128,113/128])\\
V_{x^2/2+1/4}+V_{y^2/2+1/4} &= V_{(x^2+y^2)/2+1/2} , \ \ (V_{(x^2+y^2)/2+1/2}\in [81/64,113/64])\\
V_{(x^2+y^2)/2}+V_{1/2}&=V_{(x^2+y^2)/2+1/2} , \ \ (V_{(x^2+y^2)/2}\in [49/64,81/64])\\
V_{(x^2+y^2)/4+1/4}+V_{(x^2+y^2)/4+1/4}&=V_{(x^2+y^2)/2+1/2} , \ \ (V_{(x^2+y^2)/4+1/4}\in [81/128,113/128])\\
V_{(x^2+y^2)/4+1/4}+V_{xy/2+1/4}&=V_{((x+y)/2)^2+1/2} ,\ \ (V_{xy/2+1/4}\in [81/128,113/128])\\
V_{xy/2+1/4}+V_{xy/2+1/4}&=V_{xy+1/2} ,\ \ (V_{xy+1/2}\in [81/64,113/64])\\
z+V_{1/2}&=V_{xy+1/2}.
\end{align*}
The constructed variables are in a range $[1/2,2]$.

A formula $\Phi$ of $\etrpar{1/8}_{[-1/8,1/8]}$ has been first transformed into a formula $\Phi'$ of $\etrpar{1}_{[1/2,2]}$, and subsequently into a formula $\Phi''$ of $\etrinv$. If $\Phi$ is satisfiable, then both $\Phi'$ and $\Phi''$ are satisfiable. If $\Phi''$ is satisfiable, then both $\Phi'$ and $\Phi$ are satisfiable. We get that $\etrinv$ is $\ER$-hard.

As the conjunction of the equations of $\etrinv$, together with the inequalities describing the allowed range of the variables within $\etrinv$, is a quantifier-free formula of the first-order theory of the reals, $\etrinv$ is in $\ER$, which yields that $\etrinv$ is $\ER$-complete.
\end{proof}

%
%

%
Lemma~\ref{lem:FEASABLE} and Corollary~\ref{cor:ball} together with the reductions explained in this section imply the following lemma.

\begin{lemma}\label{lem:translate-ETRINV}
	Let $\Phi$ be an instance of $\ETR$ with variables $x_1,\ldots,x_n$. 
	Then there exists an instance $\Psi$ of \etrinv with variables $y_1,\ldots,y_m$, $m\geq n$, and constants $c_1,\ldots,c_n,d_1,\ldots,d_n\in\QQ$, such that
	\begin{itemize}
	\item
	there is a solution to $\Phi$ if and only if there is a solution to $\Psi$, and
	
	\item
	for any solution $(y_1,\ldots,y_m)$ to $\Psi$, there exists a solution $(x_1,\ldots,x_n)$ to $\Phi$ where $y_1=c_1x_1+d_1,\ldots,y_n=c_nx_n+d_n$.
	\end{itemize}
\end{lemma}

    
For the proof of Theorem~\ref{thm:StrongCorrespondance} and Theorem~\ref{thm:Picasso}, we will need a stronger statement.
%
%
Recall that Corollary~\ref{cor:ball} says that there is a large ball which intersects a given semi-algebraic set.
The following related result by Basu and Roy~\cite{BasuR10} says that if the semi-algebraic set is compact, the ball will in fact contain the set.

\begin{corollary}[Basu~and~Roy~\cite{BasuR10}]\label{cor:ball2}
For any compact semi-algebraic set $S \subseteq \RR^n$ 
with description complexity $L$ it 
holds that $S\subseteq B(0,2^{2^{O(L\log L)}})$.
\end{corollary}

From the corollary, we get the following alternative stronger version of Lemma~\ref{lem:translate-ETRINV} for compact semi-algebraic sets.

\begin{lemma}\label{lem:compact-translate-ETRINV}
	Let $\Phi$ be an instance of $\ETR$ with variables $x_1,\ldots,x_n$ and a compact set of solutions. 
	Then there exists an instance $\Psi$ of \etrinv with variables $y_1,\ldots,y_m$, $m\geq n$, and constants $c_1,\ldots,c_n,d_1,\ldots,d_n\in\QQ$, such that
	$(y_1,\ldots,y_m)$ is a solution to $\Psi$ if and only if there exists a solution $(x_1,\ldots,x_n)$ to $\Phi$ with $y_1=c_1x_1+d_1,\ldots,y_n=c_nx_n+d_n$.
\end{lemma}

The only change we need to make in the construction of the $\etrinv$ instance $\Psi$ is that instead of defining $k \mydef \lceil 8n\cdot \log L + 3\rceil$, as when we use Corollary~\ref{cor:ball}, we now define $k \mydef \lceil C\cdot L \log L + 3\rceil$, where $C$ is the constant hidden in $O(L\log L)$ in Corollary~\ref{cor:ball2}.

\section{Reduction from $\etrinv$ to the art gallery problem}\label{sec:hardness}

\subsection{Notation}

Given two different points $p,q$, the line containing $p$ and $q$ is denoted as $\overleftrightarrow{pq}$, the ray with the origin at $p$ and passing through $q$ is denoted as $\overrightarrow{pq}$, and the line segment from $p$ to $q$ is denoted as $pq$.
For a point $p$, we let $x(p)$ and $y(p)$ denote the $x$- and $y$-coordinate of $p$, respectively. Table~\ref{tab:Glossary} shows the definitions of some objects and distances frequently used in the description of the construction.

\begin{table}[p]
\centering
\begin{tabular}{|l|l|}
\hline
Name & Description/value   \\
\hline
$\Phi$ & instance of \etrinv that we reduce from \\
$X\mydef \{x_1,\ldots,x_{n}\}$ & set of variables of $\Phi$ \\
$\poly \mydef \poly(\Phi)$ & 
final polygon to be constructed from $\Phi$ \\
$g\mydef g(\Phi)$ & number of guards needed to guard $\poly$ if and only if $\Phi$ has a solution\\
$k$ & number of equations in  $\Phi$\\
$n$ & number of variables in $\Phi$ \\
$N$ & $4n+k$, upper bound on the number of gadgets at one side of $\poly$ \\
$C$ &   $200000$ \\
$\ell_b$ (base line)  & line that contains $4n$ guard segments at the bottom of \poly \\
$s_i \mydef a_ib_i$ & guard segment on $\ell_b$, $i\in \{1,\ldots,4n\}, \|a_ib_i\| = 3/2$ \\
$\poly_M$ (main part of $\poly$)& middle part of $\poly$, without corridors and gadgets \\
$\ell_r$, $\ell_l$ & vertical lines bounding $\poly_M$ \\
corridor  & connection between the main part of $\poly$ and a gadget \\
$c_0 d_0, c_1 d_1$ & corridor entrances, $d_0 \mydef c_0 + (0,\frac{3}{CN^2})$ and $d_1 \mydef c_1 + (0,\frac{1.5}{CN^2})$ \\
$m$ & $c_1+(\pm 1,-1)$, relative origin of a particular gadget \\
$r_i,r_j,r_l,r'_i$ & guard segments within gadgets, of length $\frac {1.5}{C N^2}$\\  
$a'_\sigma,b'_\sigma$, $\sigma\in\{i,j,l\}$ & left and right endpoint of $r_\sigma$\\
$\ell_c$& vertical line through $\frac{c_0+c_1}{2}$\\
$o$,$o'$ & intersections of rays $\overrightarrow{a_1 c_0}$ and $\overrightarrow{b_l' c_1}$ with the line $\ell_c$\\
 $\delta$, $\rho$ , $\ee$ & $\delta \mydef \frac{13.5}{CN^2}$, 
 $\rho \mydef\frac{\delta}{9} = \frac{1.5}{CN^2}$, 
 $\ee \mydef \frac{\rho}{12} = \frac{1}{8CN^2}$     \\
$V$ & $\frac{c_0 + c_1}{2} + (0,1) + [-38N\rho,+38N\rho]\times[-38N\rho,+38N\rho]$   \\
slab $S(q,v,r)$ & region of all points with distance at most $r$ to the line through $q$\\
&with direction $v$ \\
center of slab& line in the middle of a slab\\
$L$-slabs, $R$-slabs & uncertainty regions for visibility rays, see page~\pageref{def:slabs} \\
\hline
\end{tabular}
\caption{Parameters, variables, and certain distances that are frequently used are summarized in this table for easy access. Some descriptions are much simplified.}
\label{tab:Glossary}
\end{table}


\subsection{Overview of the construction}\label{sec:hardness-overview}

Let $\Phi$ be an instance of the problem \etrinv with $n$ variables $X\mydef\{x_1,\ldots,x_{n}\}$ and consisting of $k$ equations.
We show that there exists a polygon $\poly\mydef\poly(\Phi)$ with corners at rational
coordinates which can be computed in polynomial time such that $\Phi$ has a solution if and only if
$\poly$ can be guarded by some number $g\mydef g(\Phi)$ of guards. The number $g$ will follow from the construction.
A sketch of the polygon $\poly$ is shown in Figure~\ref{fig:very_big_picture}.

\begin{figure}[h]
\centering
\includegraphics[width=\textwidth]{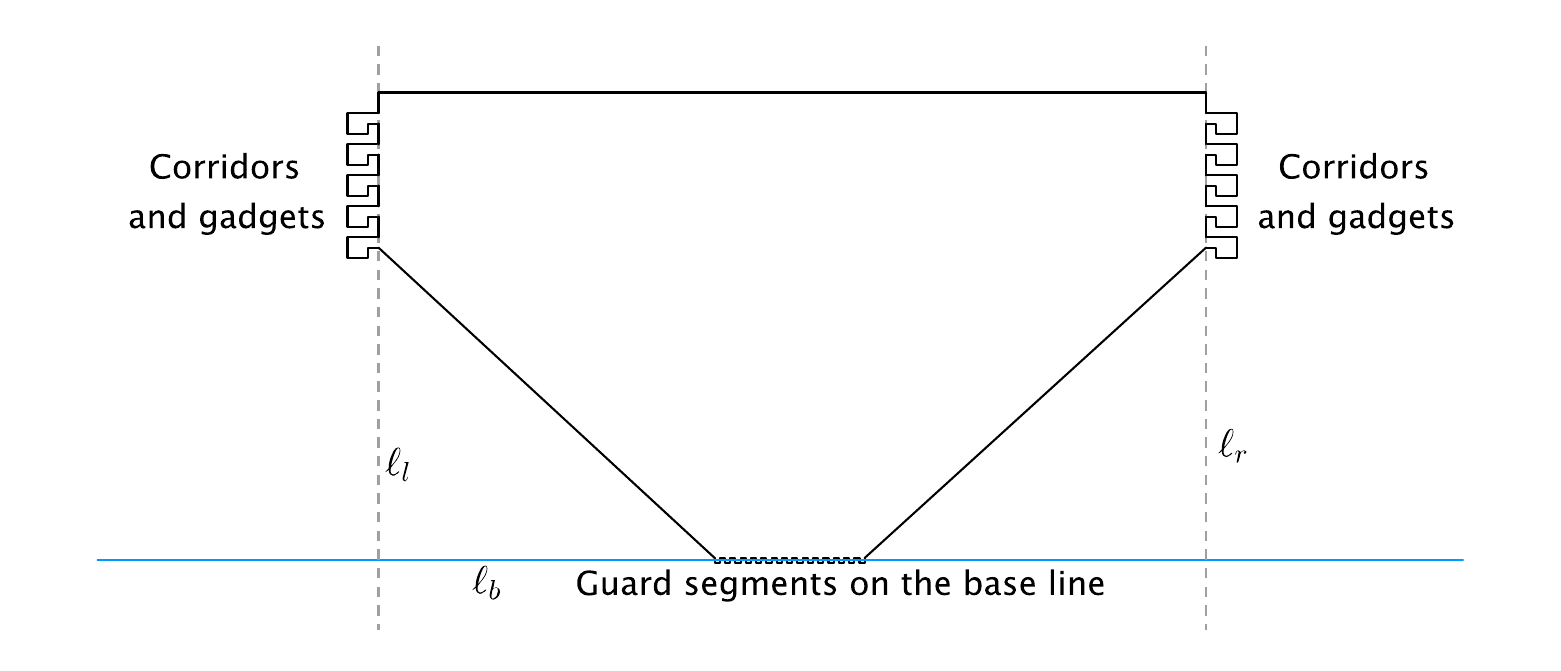}
\caption{A high-level sketch of the construction of the polygon $\poly$.}
\label{fig:very_big_picture}
\end{figure}

Each variable $x_i \in X$ is represented by a collection of \emph{guard segments}, which are horizontal line segments contained in the interior of $\poly$.
Consider one guard segment $s\mydef ab$, where $a$ is to the left of $b$, and assume that $s$ represents the variable $x_i$ and that there is exactly one guard $p$ placed on $s$. The guard segment $s$ can be oriented to the right or to the left.
The guard $p$ on $s$ specifies the value of the variable $x_i$ as
$\frac{1}{2}+\frac{3\|ap\|}{2\|ab\|}$ if $s$ is oriented to the right, and $\frac{1}{2}+\frac{3\|bp\|}{2\|ab\|}$ if $s$ is oriented to the left, i.e., the value is a linear map from $s$ to $[\frac 12,2]$.

Suppose that there is a solution to $\Phi$.
We will show that in that case any minimum guard set $G$ of $\poly$ has size $g(\Phi)$ and \emph{specifies} a solution to $\Phi$ in the sense that it satisfies the following two properties.
\begin{itemize}
\item Each variable $x_i \in X$ is specified \emph{consistently} by $G$, i.e., there is exactly one guard on each guard segment representing $x_i$, and all these guards specify the same value of $x_i$.
\item The guard set $G$ is \emph{feasible}, i.e., the values of $X$ thus specified is a solution to~$\Phi$.
\end{itemize}
Moreover, if there is no solution to $\Phi$, each guard set of $\poly$ consists of more than $g(\Phi)$ guards.

The polygon $\poly$ is constructed in the following way.
The bottom part of the polygon consists of a collection of \emph{pockets}, containing in total $4n$ collinear and equidistant guard segments.
We denote the horizontal line containing these guard segments as the \emph{base line} or $\ell_b$.
In order from left to right, we denote the guard segments as $s_1,\ldots,s_{4n}$.
The segments $s_1,\ldots,s_{n}$ are right-oriented segments representing the variables $x_1,\ldots,x_{n}$, as are the segments $s_{n+1},\ldots,s_{2n}$, and $s_{2n+1},\ldots,s_{3n}$.
The segments $s_{3n+1},\ldots,s_{4n}$ are left-oriented and they also represent the variables $x_1,\ldots,x_{n}$.
At the left and at the right side of $\poly$, there are some \emph{corridors} attached, each of which leads into a \emph{gadget}.
The \emph{entrances} to the corridors at the right side of $\poly$ are line segments contained in a vertical line $\ell_r$.
Likewise, the entrances to the corridors at the left side of $\poly$ are contained in a vertical line $\ell_l$.
The gadgets also contain guard segments, and they are used to impose dependencies between the guards in order to ensure that if there is a solution to $\Phi$, then any minimum guard set of $\poly$ consists of $g(\Phi)$ guards and specifies a solution to $\Phi$ in the sense defined above.
The corridors are used to copy the positions of guards on guard segments on the base line to guards on guard segments inside the gadgets.
Each gadget corresponds to a constraint of one of the types $x+y\geq z$, $x+y\leq z$, $x\cdot y=1$, $x+y\geq 5/2$, and $x+y\leq 5/2$.
The first three types of constraints are used to encode the dependencies between the variables in $X$ as specified by $\Phi$, whereas the latter two constraints are used to encode the dependencies between the right-oriented and left-oriented guard segments representing a single variable in~$X$.

The reason that we need three right-oriented guard segments $s_i,s_{i+n},s_{i+2n}$ representing each variable is that in the addition gadgets, we need to copy in guards to three right-oriented guard segments, and they are allowed to all represent the same variable.
In contrast to this, we need at most one left-oriented guard segment in each gadget.
Furthermore, since the right-oriented guard segments appear in three groups ($s_1,\ldots,s_{n}$, $s_{n+1},\ldots,s_{2n}$, and $s_{2n+1},\ldots,s_{3n}$), it is possible for each set $\{i,j,l\}\subseteq\{1,\ldots,n\}$ to choose three right-oriented guard segments $s_{i'},s_{j'},s_{l'}$ representing $x_i,x_j,x_l$, respectively, and appearing in any prescribed order on $\ell_b$.


\subsection{Creating a stationary guard position}

We denote some points of $\poly$ as \emph{stationary guard positions}. A guard placed at a stationary guard position is called a \emph{stationary guard}. We will often define a stationary guard position as the unique point $p \in \poly$ such that a guard placed at $p$ can see some two corners $q_1, q_2$ of the polygon~$\poly$.

We will later prove that for any guard set of size of at most $g(\Phi)$, there is a guard placed at each stationary guard position. 
For that, we will need the lemma stated below.
For an example of the application of the lemma, see Figure~\ref{fig:copyGuard}. The stationary guard position $g_2$ is the only point from which a guard can see both corners $q_1$ and $q_2$. Applying Lemma 
\ref{lem:stationary_guard} with $p\mydef g_2$, $W\mydef \{q_1,q_2\}$,  $A\mydef P$ and $M\mydef \{t_1,t_2\}$, we get that there must be a guard placed at $g_2$ in any guard set of size $3$.
The purpose of the area $A$ is so that we can restrict our arguments to a small area of the polygon.

\begin{lemma}\label{lem:stationary_guard}
Let $P$ be a polygon, $A \subseteq P$, and $M$ a set of points in $A$ such that no point in $M$ can be seen from a point in $P \setminus A$, and no two points in $M$ can be seen from the same point in $P$. Suppose that there is a point $p\in A$ and a set of points $W\subset A$ such that
\begin{enumerate}
\item no point in $W$ can be seen from a point in $P \setminus A$,
\item the only point in $P$ that sees all points in $W$ is $p$, and
\item no point in $P$ can see a point in $M$ and a point in $W$ simultaneously. 
\end{enumerate}
Then any guard set of $P$ has at least $|M|+1$ guards placed within $A$, and if a guard set with $|M|+1$ guards placed within $A$ exists, one of its guards is placed at $p$.
\end{lemma}

\begin{proof}
Let $q$ be a point in $W$.
Since no two points in the set $\{q\}\cup M$ can be seen from the same point in $P$, and no point from $P \setminus A$ can see a point in $\{q\}\cup M$, at least $|M|+1$ guards are needed within $A$.
Suppose that a guard set with exactly $|M|+1$ guards placed within $A$ exists.
There must be $|M|$ guards in $A$ such that each of them can see one point in $M$ and no point in $W$.
The last guard in $A$ has to be at the point $p$ in order to see all points in $W$.
\end{proof}

In the polygon $\poly$ we often use stationary guards for the purpose of seeing some region on one side of a line segment $\ell$, but no points on the other side of $\ell$.
Other guards have the responsibility to see the remaining area.
See Figure~\ref{fig:stationaryGuard} (left) for an explanation of how a stationary guard position can be constructed.

\begin{figure}[htbp]
\centering
\includegraphics{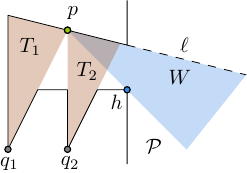}
\hspace{0.1\textwidth}
\includegraphics
{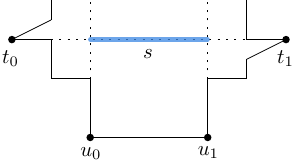}
\caption{Left: The construction of a stationary guard position $p$ that sees an area in $P$ below a line segment $\ell$.
The brown areas are the regions of points that see $q_1$ and $q_2$, and $p$ is the only point that sees both $q_1$ and $q_2$.
The point $p$ sees the points in the blue wedge, and the angle of the wedge can be adjusted by choosing the point $h$ accordingly.
Right: The construction of a guard segment $s$ (the blue segment).
In order to see the points $t_0,t_1$, a guard must be on the horizontal dotted segment.
Furthermore, in order to see $u_0,u_1$, the guard must be between the vertical dotted segments that contain the endpoints of $s$. Thus, a guard sees $t_0,t_1,u_0,u_1$ if and only if the guard is at $s$.}
\label{fig:stationaryGuard}
\end{figure}

\subsection{Creating a guard segment}\label{sec:guard-segment}

In the construction of $\poly$ we will denote some horizontal line segments of $\poly$ as \emph{guard segments}. We will later prove that for any guard set of size of at most $g(\Phi)$, there is exactly one guard placed on each guard segment.

We will always define a guard segment $s$ by providing a collection of four corners of $\poly$ such that a guard within $\poly$ can see all these four corners if and only if it is placed on the line segment~$s$. See Figure~\ref{fig:stationaryGuard} (right) for an example of such a construction. To show that there is a guard placed on a guard segment, we will use the following lemma.

\begin{lemma}\label{lem:guard_segment}
Let $P$ be a polygon, $A \subseteq P$, and $M$ a set of points in $A$ such that no point in $M$ can be seen from a point in $P \setminus A$, and no two points in $M$ can be seen from the same point in $P$. 
Suppose that there is a line segment $s$ in $A$, and points $t_0,t_1,u_0,u_1 \in A$ such that
\begin{enumerate}
\item no point in $\{t_0,t_1,u_0,u_1\}$ can be seen from a point in $P \setminus A$,
\item a guard in $P$ sees all of the points $t_0,t_1,u_0,u_1$ if and only if the guard is at $s$, and
\item no point in $P$ can see a point in $M$ and one of the points $t_0,t_1,u_0,u_1$.
\end{enumerate}
Then any guard set of $P$ has at least $|M|+1$ guards placed within $A$, and if a guard set with $|M|+1$ guards placed within $A$ exists, one of its guards is placed on the line segment $s$.
\end{lemma}

\begin{proof}
Similar to the proof of Lemma~\ref{lem:stationary_guard}.
\end{proof}

Consider once more the example pictured in Figure~\ref{fig:copyGuard}, where we want to guard the polygon with only three guards. We define two guard segments, $a_0b_0$ and $a_1b_1$. The first one is defined by the corners $t_0,t_1,u_0,u_1$, and the second one by the corners $t_2,t_3,u_2,u_3$. Applying Lemma~\ref{lem:guard_segment} we get that there must be a guard placed on $a_0b_0$ (we set $A\mydef P$ and $M\mydef \{t_2,q_2\}$) and at $a_1b_1$ (we set $A\mydef P$ and $M\mydef \{t_1,q_2\}$) in any guard set of size $3$.

As already explained in Section~\ref{sec:hardness-overview}, guards placed on the guard segments will be used to encode the values of the variables of $\Phi$.


\subsection{Imposing inequalities by \nook s and \umbra s}\label{sec:nooks and umbras}

In this section we introduce \nook s and \umbra s, which are our basic tools used to impose dependency between guards placed on two different guard segments. 
For the following definitions, see Figure~\ref{fig:nook_and_umbra}.

\begin{figure}[h]
\centering
\includegraphics[clip, trim= 1.1cm 0.5cm 1.1cm 0.5cm,scale=0.65]{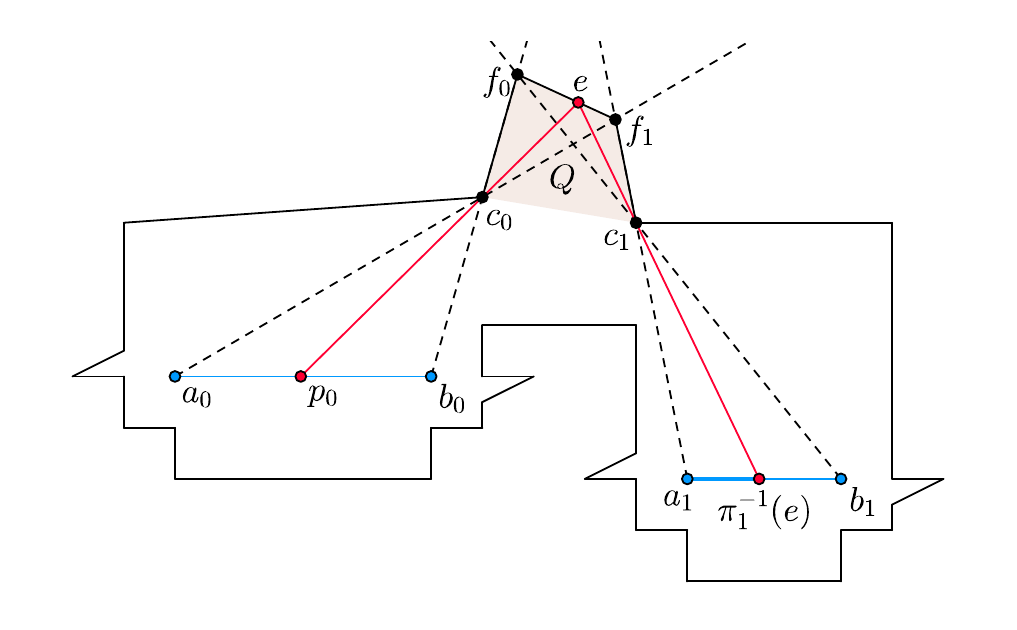}

\hspace{0.3cm}

\includegraphics[clip, trim= 1.1cm 0.5cm 1.1cm 0.5cm,scale=0.65]{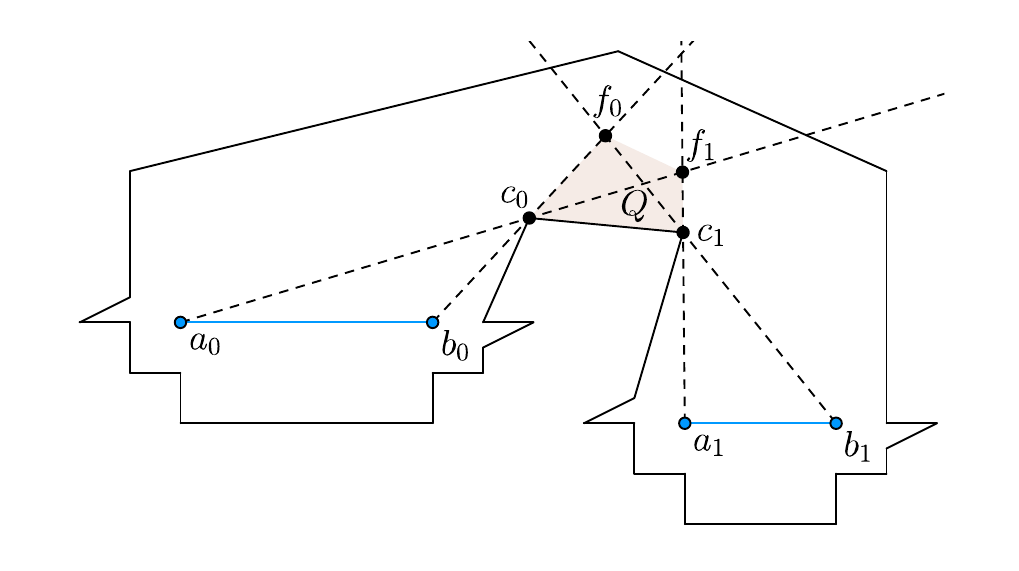}
\caption{The brown area $Q$ representing \anook\ (top), and \anumbra\ (bottom).
In the left figure, note that if a guard $p_1$ placed on the segment $a_1b_1$ has to see the whole line segment $f_0f_1$ together with $p_0$, then $p_1$ must be on or to the left of the point $\pi^{-1}_1(e)$, where $e\mydef \pi_0(p_0)$.}
\label{fig:nook_and_umbra}
\end{figure}


\begin{definition}[nook and umbra]\label{def:NookUmbra}
Let $\poly$ be a polygon with guard segments $r_0\mydef a_0b_0$ and $r_1\mydef a_1b_1$, where $r_0$ is to the left of $r_1$. Let $c_0,c_1$ be two corners of $\poly$, such that $c_0$ is to the left of $c_1$.
Suppose that the rays $\overrightarrow{b_0c_0}$ and $\overrightarrow{b_1c_1}$ 
intersect at a point $f_0$, the lines $\overrightarrow{a_0c_0}$ and $\overrightarrow{a_1c_1}$ intersect at a point $f_1$, and that $Q\mydef c_0c_1f_1f_0$ is a convex quadrilateral contained in $\poly$.
For each $i\in\{0,1\}$ define the function $\pi_i\colon r_i\longrightarrow f_0f_1$ such that $\pi_i(p)$ is the intersection of the ray $\overrightarrow{pc_i}$ with the line segment $f_0f_1$, and suppose that $\pi_i$ is bijective.

We say that $Q$ is \aemphnook\ of the pair of guard segments $r_0,r_1$ if for each $i\in\{0,1\}$ and every $p\in r_i$, a guard at $p$ can see all of the segment $\pi_i(p)f_{1-i}$ but nothing else of $f_0f_1$.
We say that $Q$ is \anemphumbra\footnote{Our choice of the term ``\umbra'' was inspired by its meaning in astronomy: ``the complete or perfect shadow of an opaque body, as a planet, where the direct light from the source of illumination is completely cut off''~\cite{dictionary}.} of the segments $r_0,r_1$ if for each $i\in\{0,1\}$ and every $p\in r_i$, a guard at $p$ can see all of the segment $\pi_i(p)f_i$ but nothing else of $f_0f_1$.
The functions $\pi_0,\pi_1$ are called \emph{projections} of the \nook\ or the \umbra.
\end{definition}

\begin{figure}[htbp]
\centering
\includegraphics[clip, trim= 1cm 1cm 1cm 1cm,width=0.8\textwidth]
{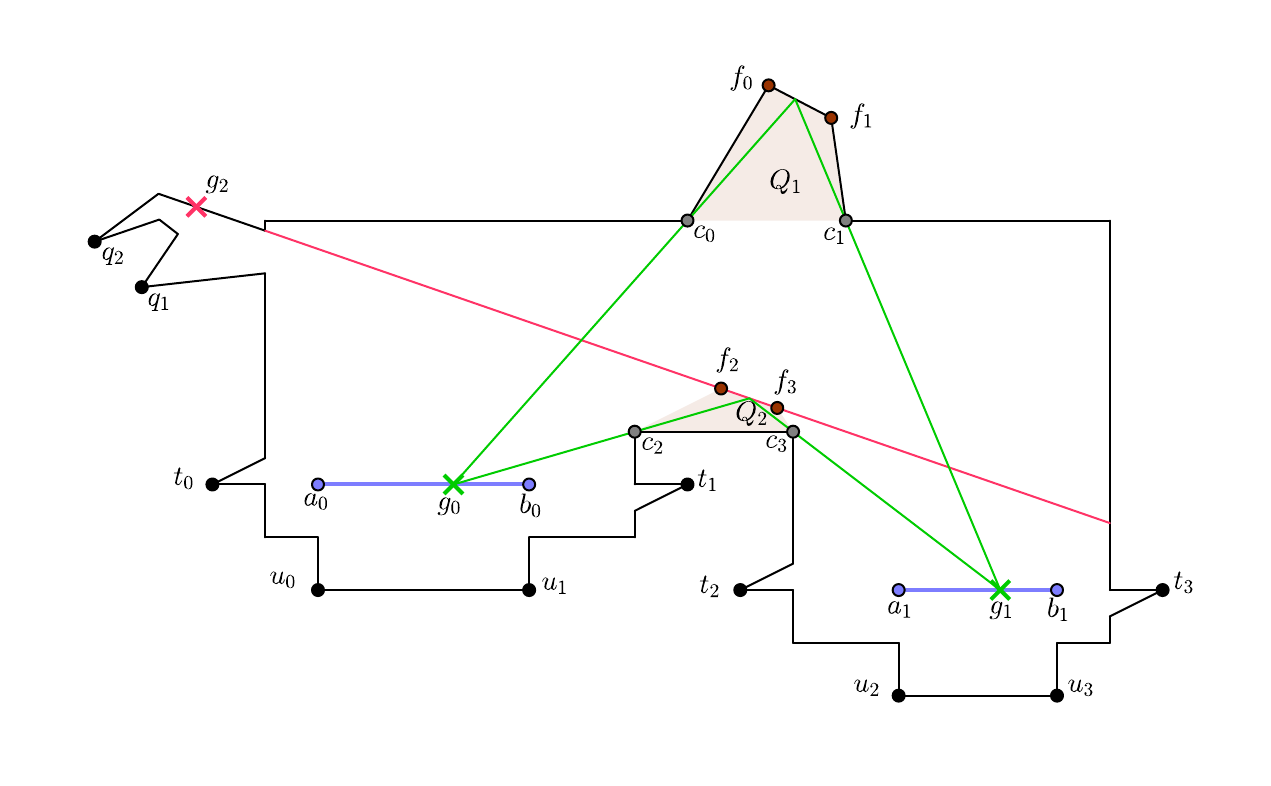}
\caption{$Q_1$ is a copy-\nook\ of the segments $r_0\mydef a_0b_0$ and $r_1\mydef a_1b_1$ with a critical segment $f_0f_1$, and $Q_2$ is a copy-\umbra\ for the same pair with a critical segment $f_2f_3$. Lemmas~\ref{lem:stationary_guard} and \ref{lem:guard_segment} imply that this polygon cannot be guarded by fewer than $3$ guards, and any guard set with $3$ guards must contain a guard $g_0$ on $r_0$,  a guard $g_1$ on $r_1$, and a stationary guard at the point $g_2$. The guards $g_0$ and $g_1$ must specify the same value on $r_0$ and $r_1$, respectively.}
\label{fig:copyGuard}
\end{figure}

We will construct \nook s and \umbra s for pairs of guard segments where we want to enforce dependency between the values of the corresponding variables.
When making use of \anumbra, we will also create a stationary guard position from which a guard sees the whole quadrilateral $Q$, but nothing on the other side of the line segment $f_0f_1$.
In this way we can enforce that the guards on $r_0$ and $r_1$ together see all of $f_0f_1$, since they need to see an open region on the other side of, and bounded by, $f_0f_1$.
For the case of a nook, the segment $f_0f_1$ will always be on the polygon boundary, and then there will be no stationary guard needed.
See Figure~\ref{fig:copyGuard} for an example of a construction of both a nook and an umbra for a pair of guard segments.


\begin{definition}[critical segment and shadow corners]\label{def:criticalSegment}
Consider \anook\ or \anumbra\ $Q\mydef c_0c_1f_1f_0$ of a pair of guard segments $r_0,r_1$.
The line segment $f_0f_1$ is called the \emph{critical segment} of $Q$, and the corners $c_0,c_1$ are called the \emph{shadow corners} of $Q$.
\end{definition}

Consider \anook\ or \anumbra\ of a pair of guard segments $r_0,r_1$.
Let $p_0,p_1$ be the guards placed on the guard segments $r_0$ and $r_1$, respectively, and assume that $p_0$ and $p_1$ together see all of the critical segment $f_0f_1$.
Let $e\mydef \pi_0(p_0)$. The condition that $p_0,p_1$ together see all of $f_0f_1$ enforces dependency between the position of the guard $p_{1}$ and the point $\pi^{-1}_{1}(e)$.
%
If $Q$ is \anook, $p_{1}$ must be in the closed wedge $W$ between the rays $\overrightarrow{ec_{0}}$ and $\overrightarrow{ec_{1}}$.
%
%
If $Q$ is \anumbra, $p_{1}$ must be in $\overline{\poly\setminus W}$ (the closure of the complement of $W$), i.e., either on or to the right of $\pi^{-1}_{1}(e)$.
This observation will allow us to impose an inequality on the $x$-coordinates of $p_0,p_1$, and thus on the variables corresponding to the guard segments $r_0,r_1$.


\subsection{Copying one guard segment}


\begin{definition}
Let $Q$ be \anook\ or \anumbra\ of a pair of guard 
segments $r_0\mydef a_0b_0$ and $r_1\mydef a_1b_1$ 
with the same orientation, such that
the shadow corners $c_0$ and $c_1$ have the same $y$-coordinate.
We then call $Q$ a \emph{copy-\nook} or a \emph{copy-\umbra}, respectively.
\end{definition}

We can show the following result.

\begin{lemma}\label{lem:copy-proof}
Let $Q$ be a {copy-\nook} or a {copy-\umbra} for a pair of guard segments $r_0\mydef a_0b_0$ and $r_1\mydef a_1b_1$. Then
for every point $e\in f_0f_1$ we have $\frac{\|a_0\pi^{-1}_0(e)\|}{\|a_0b_0\|}=\frac{\|a_{1}\pi^{-1}_{1}(e)\|}{\|a_{1}b_{1}\|}$, i.e., the points $\pi^{-1}_0(e)$ and $\pi^{-1}_1(e)$ on the corresponding guard segments $r_0$ and $r_1$ represent the same value.
\end{lemma}
\begin{proof}
See~Figure \ref{fig:copy_proof}.
\begin{figure}[htbp]
\centering
\includegraphics[clip, trim=1cm 0.9cm 1cm 0.9cm,scale=1.0]{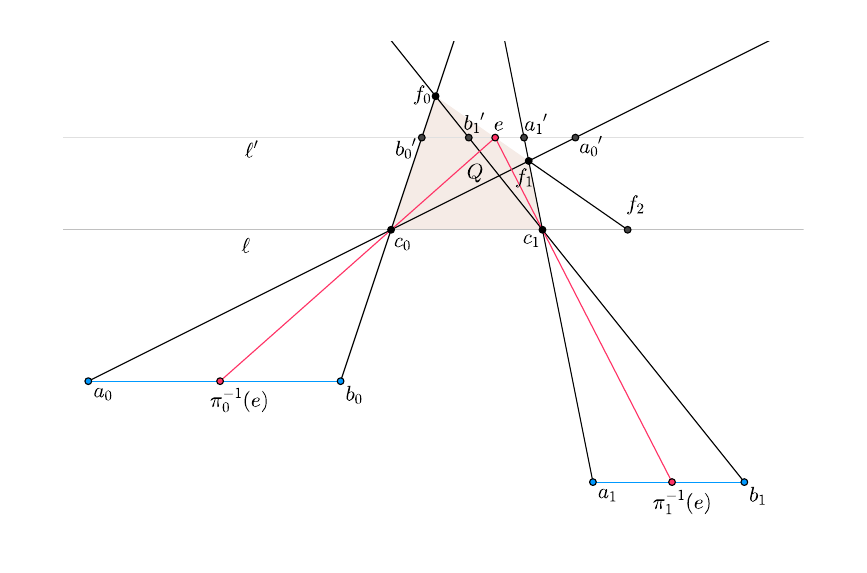}
\caption{A {copy-\nook} or a {copy-\umbra} $Q$ for a pair of guard segments $r_0\mydef a_0b_0$ and $r_1\mydef a_1b_1$. The points $\pi^{-1}_0(e)$ and $\pi^{-1}_1(e)$ represent the same value.}
\label{fig:copy_proof}
\end{figure}
Let $\ell\mydef \overleftrightarrow{c_0c_1}$ be the horizontal line containing the line segment $c_0c_1$, and $\ell'$ a horizontal line passing through $e$. Let $f_2$ be an intersection point of the line $\overleftrightarrow{f_0 f_1}$ with the line $\ell$.
Let $a_0', a_1', b_0', b_1'$ be the intersection points of the rays $\overrightarrow{a_0 c_0}, \overrightarrow{a_1 c_1}, \overrightarrow{b_0 c_0}, \overrightarrow{b_1 c_1}$, respectively, with the line $\ell'$.

We obtain
$$\frac{\|a_0\pi^{-1}_0(e)\|}{\|\pi^{-1}_0(e) b_0\|} \cdot \frac{\|\pi^{-1}_1(e) b_1\|}{\|a_1\pi^{-1}_1(e)\|} = 
\frac{\|e a'_0\|}{\|b'_0 e\|} \cdot \frac{\| b'_1 e\|}{\|e a'_1\|} =
\frac{\|e a'_0\|}{\| e a'_1\|} \cdot \frac{\|b'_1 e\|}{\| b'_0 e \|} =
\frac{\| c_0 f_2\|}{\| c_1 f_2 \|} \cdot \frac{\| c_1 f_2 \|}{\| c_0 f_2 \|} = 1.
$$
The first equality holds as the following pairs of triangles are similar: $a_0\pi^{-1}_0(e)c_0$ and $a'_0ec_0$, $\pi^{-1}_0(e) b_0 c_0$ and $e b'_0 c_0$,
$a_1\pi^{-1}_1(e)c_1$ and $a'_1ec_1$, $\pi^{-1}_1(e) b_1 c_1$ and $e b'_1 c_1$. The third equality holds as the following pairs of triangles are similar:
$e a'_0 f_1$ and $f_2 c_0 f_1$, $e a'_1 f_1$ and $f_2 c_1 f_1$, $b'_1ef_0$ and $c_1f_2f_0$, and $b'_0 e f_0$ and $c_0 f_2 f_0$.
\end{proof}

The following lemma is a direct consequence of Lemma~\ref{lem:copy-proof}.
See Figure~\ref{fig:copyGuard} for an example of how a construction as described in the lemma can be made.

\begin{lemma}\label{lem:copy-lemma}
Let $r_0, r_1$ be a pair of guard segments oriented in the same way
for which there is both a copy-\nook\ and a copy-\umbra.
Suppose that there is exactly one guard $p_0$ placed on $r_0$ and one guard $p_1$ placed on $r_1$, and that the guards $p_0$ and $p_1$ together see both critical segments.
Then the guards $p_0$ and $p_1$ specify the same value.
\end{lemma}

\begin{definition}
Let $r_0, r_1$ be a pair of guard segments for which there is both a copy-\nook\ and a copy-\umbra. We say that $r_1$ is a \emph{copy} of $r_0$. 
If there is only a copy-\nook\ \emph{or} a copy-\umbra\ of the pair $r_0,r_1$, we say that $r_1$ is a \emph{weak copy} of $r_0$. 
\end{definition}

It will follow from the construction of the polygon $\poly$ that if there is a solution to $\Phi$, then for any optimal guard set of $\poly$ and any pair of guard segments $r_0,r_1$  
such that $r_1$ is a copy of $r_0$, there is exactly one guard on each segment $r_0,r_1$, and the guards together see the whole critical segment of both the copy-\nook\ and the copy-\umbra.


\subsection{Overall design of the polygon $\poly$}\label{sec:poly}

Recall the high-level sketch of the polygon $\poly$ in Figure \ref{fig:very_big_picture}. The bottom part of the polygon consists of \emph{pockets} containing $4n$ guard segments $s_1,\ldots,s_{4n}$.
The guard segments are placed on the base line $\ell_b$, each segment having a width of $3/2$ and contained within a pocket of width $13.5$. 
Therefore the horizontal space used for the $4n$ guard segments on the base line $\ell_b$ is $54n$.
The wall of $\poly$ forming the $4n$ pockets is denoted the \emph{bottom wall}.
The detailed description of the pockets and of the bottom wall is presented in Section~\ref{sec:base line}.

Let $N\mydef 4n+k$. 
At the right side of $\poly$ and at the left side of $\poly$ there will be at most $N$ \emph{corridors} attached, each of which leads into a \emph{gadget}.
The \emph{entrances} to the corridors are contained in the vertical lines $\ell_r$ and $\ell_l$.
The corridors are described in Section~\ref{sec:copy}, and they are placed equidistantly, with a vertical distance of $3$ between the entrances of two consecutive corridors along the lines $\ell_r$ and $\ell_l$.
The gadgets are described in Sections~\ref{sec:additionGadget}--\ref{sec:inversionGadget}.
The total vertical space occupied by the corridors and the gadgets at each side of $\poly$ is at most $3N$.

\begin{figure}[htbp]
\centering
\includegraphics[clip, trim=1cm 0.9cm 1cm 0.9cm,width=\textwidth]{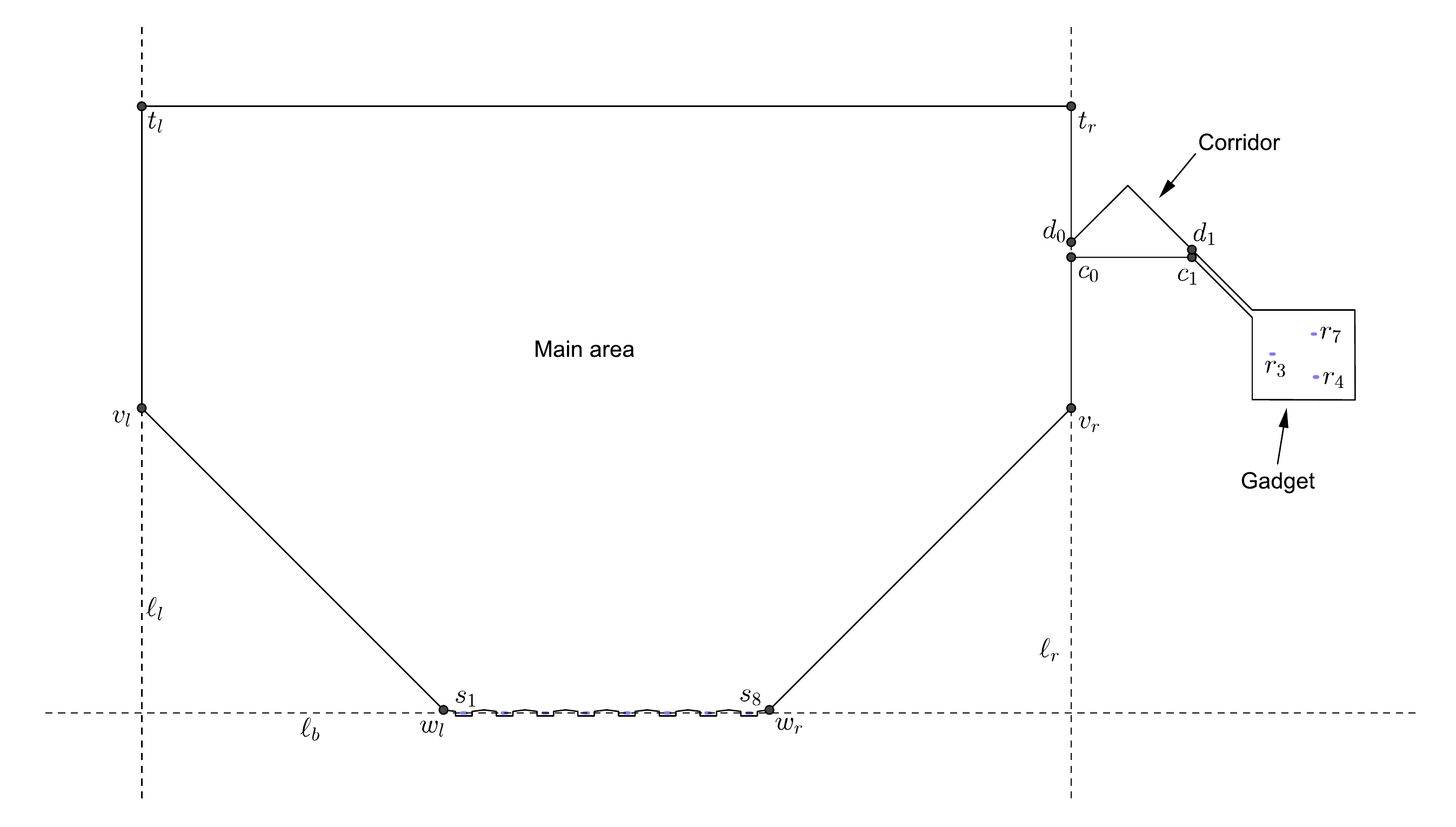}
\caption{A sketch of the construction of $\poly$ with $8$ guard segments $s_1,\ldots,s_8$ on the base line, and one gadget at the right side of $\poly$. The gadget contains the guard segments $r_2,r_3,r_6$, which are copies of the segments $s_2,s_3,s_6$, respectively.
The gadget is the part of the polygon to the right of the line segment $c_1d_1$. The proportions in the drawing are not correct.}
\label{fig:bigPicture}
\end{figure}

Consider the sketch of the polygon $\poly$ in Figure \ref{fig:bigPicture}.
Define a constant $C\mydef 200000$. 
Let $w_l$ and $w_r$ denote the left and right endpoint of the bottom wall, respectively.
The horizontal distance from $w_l$ to the line $\ell_l$ is $CN^2-54n+6$, as is the horizontal distance from $w_r$ to $\ell_r$. The horizontal distance from the left endpoint $a_1$ of the leftmost segment to $\ell_r$, as well as the horizontal distance from the right endpoint $b_{4n}$ of the rightmost segment to $\ell_l$, is $CN^2$. 
The vertical distance from $\ell_b$ up to the entrance of the first corridor is $C N^2$. 
The boundary of $\poly$ contains an edge connecting $w_l$ to the point $v_l\mydef w_l+(-(CN^2-54n+6),CN^2-1)$ 
on $\ell_l$, and an edge connecting $w_r$ to the point $v_r\mydef w_r+(CN^2-54n+6,CN^2-1)$ on $\ell_r$.
Let $t_l\mydef v_l+(0,3N)$ and $t_r\mydef v_r+(0,3N)$. 
The \emph{main area} $\poly_M$ of $\poly$ is the area bounded by the bottom wall of $\poly$ (to be defined in Section~\ref{sec:base line}) and a polygonal curve defined by the points $w_lv_lt_lt_rv_rw_r$.
The entrances to the corridors are on the segment $v_lt_l$ in the left side and on the segment $v_rt_r$ in the right side.
The set $\poly\setminus\poly_M$ outside of the main area consists of corridors and gadgets.

The reason why we need the distances from the guard segments on the base line $\ell_b$ to the gadgets to be so large is that we want all the rays from the guard segments on the base line through the corridor entrances on $\ell_r$ ($\ell_l$) to have nearly the same slopes.
That will allow us to describe a general method for copying guard segments from the base line into the gadgets.

Each gadget corresponds to a constraint involving either two or three variables, where each variable corresponds to a guard segment on the base line. Gadgets are connected with the main area $\poly_M$ via corridors. A corridor does not contain any guard segments, and its aim is enforcing consistency between (two or three) pairs of guard segments, where one segment from each pair is in $\poly_M$ and the other one is in the gadget. Each corridor has two vertical \emph{entrances}, the entrance $c_0d_0$ of height $\frac{3}{CN^2}$ connecting it with $\poly_M$, and the entrance $c_1d_1$ of height $\frac{1.5}{CN^2}$ connecting it with a gadget. The bottom wall of a corridor is a horizontal line segment $c_0c_1$ of length $2$. The shape of the upper wall is more complicated, and it depends on the indices of the guard segments involved in the corresponding constraint, and on the height at which the corridor is placed with respect to the base line $\ell_b$.

A gadget can be thought of as a room which is connected with the main area $\poly_M$ of $\poly$ via a corridor, i.e., attached to the corners $c_1$ and $d_1$ of the corridor. There are five different kinds of gadgets, each corresponding to a different kind of inequality or equation, and, unlike for the case of corridors, all gadgets of the same type are identical. Each gadget contains one or two guard segments for each variable present in the corresponding formula. All guard segments within a gadget are of length $\frac {1.5}{C N^2}$, and are placed very close to the \emph{middle point} of the gadget, defined as $m\mydef c_1+(1,-1)$ for gadgets at the right side of $\poly$, and $m\mydef c_1+(-1,-1)$ for gadgets at the left side of $\poly$.


\subsection{Construction of the bottom wall}\label{sec:base line}

In this section we present the construction of the bottom wall of $\poly$.
We first describe the overall construction, as shown in Figure~\ref{fig:guardSegments}, and later we introduce small features corresponding to each equation of the type $x_i=1$ in $\Phi$.

The bottom wall forms $4n$ \emph{pockets}, each pocket containing one guard segment on the base line $\ell_b$.
Each pocket has a width of $13.5$.
Each guard segment has a length of $3/2$, and the distance between two consecutive segments is $12$.

\begin{figure}[h]
\centering
\includegraphics[trim = 50mm 10mm 50mm 10mm,clip,width=\textwidth]{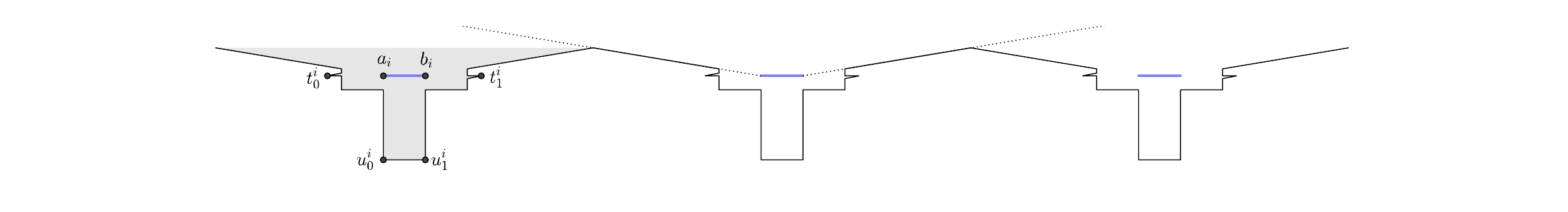}
\caption{The construction of three consecutive guard segments (blue) on the base line. A pocket corresponding to a single guard segment $s_i\mydef a_ib_i$ is marked in grey.}
\label{fig:guardSegments}
\end{figure}

Let $s_1\mydef a_1b_1,\ldots,s_{4n}\mydef a_{4n}b_{4n}$ be the guard segments in order from left to right.
A pocket for a guard segment $s_i\mydef a_ib_i$ is constructed as shown in Figure~\ref{fig:guardSegments} (the grey area in the figure). The left endpoint of the pocket is at the point $a_i+(-6,1)$, 
and the right endpoint of the pocket is at the point $a_i+(7.5,1)$.
The guard segment is defined by the following points, which are corners of the pocket: $t^i_0\mydef a_i+(-2,0)$, $t^i_1\mydef a_i+(3.5,0)$, $u^i_0\mydef a_i+(0,-5)$, $u^i_1\mydef a_i+(1.5,-5)$. The vertical edges of the pocket are contained in lines $x=x(a_i)-1.5$,  $x=x(a_i)$, $x=x(a_i)+1.5$, and $x=x(a_i)+3$.
The horizontal edges of the pocket are contained in lines $y=0, y=-0.5$, and $y=-5$. The remaining edges are constructed so that the points $t^i_0, t^i_1$ can be seen only from within the pocket, and that any point on $s_i$ sees the whole pocket.

Consider an equation of the form $x_i=1$ in $\Phi$.
There are four guard segments representing $x_i$, i.e., the guard segments $s_i, s_{i+n}, s_{i+2n}$, and $s_{i+3n}$, where the first three are right-oriented and the last one is left-oriented.
We add two spikes in the construction of the leftmost of these guard segments, i.e., the segment $s_i$, as shown in Figure~\ref{fig:guardSegmentsEqualOne}.
The dashed lines in the figure intersect at the point $g_i \in s_i$, where $g_i \mydef  a_i+(1/2,0)$. 
The spike containing $q^i_1$ enforces the guard to be at the point $g_i$ or to the right of it, while the spike containing $q^i_2$ enforces the guard to be at $g_i$ or to the left of it.
Also, the points $q^i_1$ and $q^i_2$ are chosen so that they cannot be seen by any points from within the corridors or gadgets.
The guard segment is thus reduced to a stationary guard position $g_i$ corresponding to the value $x_i=1$. 

\begin{figure}[h]
\centering
\includegraphics[clip, trim=1.5cm 0.5cm 1.5cm 0.5cm,width=0.6 \textwidth]{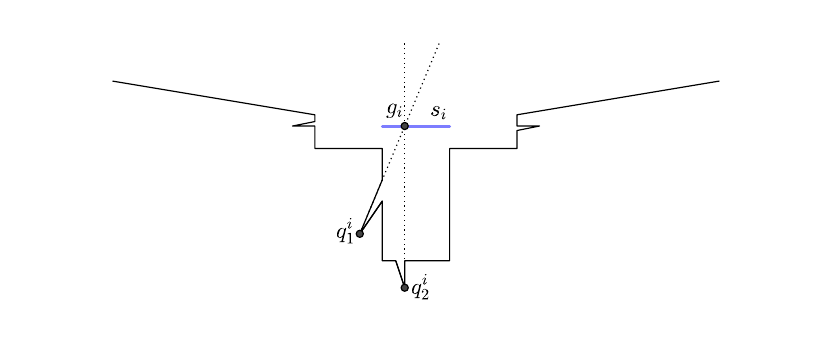}
\caption{The spikes with corners at $q^i_1$ and $q^i_2$ enforce the guard from the guard segment $s_i$ to be at the point $g_i$ corresponding to the value of $1$.}
\label{fig:guardSegmentsEqualOne}
\end{figure}

Note that we only need to add such spikes to the pocket containing $s_i$, since the construction described in Section~\ref{sec:leftRight}
will enforce the guards on all the segments $s_i, s_{i+n}, s_{i+2n}$, and $s_{i+3n}$ to specify the value of $x_i$ consistently.
We have now specified all the details of the main area $\poly_M$ of $\poly$ (recall the definition of $\poly_M$ from Section~\ref{sec:poly}). The following lemma holds.

\begin{lemma}\label{lem:rational-bottom-wall}
There is a constant $\zeta$ such that for any instance $\Phi$ of \etrinv, we can construct in polynomial time the bottom wall corresponding to $\Phi$ such that every corner has rational coordinates, with the numerator bounded from above by $\zeta N$ and denominator bounded from above by $\zeta$.
\end{lemma}

We can prove the following lemma.

\begin{lemma}\label{lemma:mainAreaGuards}
Any guard set $G$ of the polygon $\poly$ satisfies the following properties.
\begin{itemize}
\item $G$ has at least $4n$ guards placed in $\poly_M$.
\item If $G$ has exactly $4n$ guards placed in $\poly_M$, then it has one guard within each of the $4n$ guard segments $s_1, \ldots ,s_{4n}$.
\item If $G$ has exactly $4n$ guards placed in the main area of $\poly$, then for each variable $x_i \in X$ such that there is an equation $x_i=1$ in $\Phi$, the guard at the segment $s_i$ is at the position corresponding to the value $1$.
\end{itemize}
Moreover, let $G'$ be any set of points such that (i) $G'$ has a point within each of the $4n$ guard segments of $\poly_M$, and (ii) for each variable $x_i \in X$ such that there is an equation $x_i=1$ in $\Phi$, there is a guard at $s_i$ at the position corresponding to the value of $1$.
Then $G'$ is a guard set of $\poly_M$.
\end{lemma}
\begin{proof}
Recall that each point $t^i_1$, for $i \in \{1,\ldots,4n\}$, can be seen only from within the pocket corresponding to the guard segment $s_i$. Therefore any guard set requires at least one guard within each of these pockets, i.e., it contains at least $4n$ guards in $\poly_M$.

Now, consider any guard segment $s_i$, for $i \in \{1,\ldots,4n\}$. Let $M\mydef \{t^j_1: j \in \{1,2,\ldots ,4n\}, j \neq i \}$. Note also that none of the points $t^i_0,t^i_1,u^i_0,u^i_1$, and also no point from $M$ can be seen by a guard placed within a corridor or a gadget of $\poly$, i.e., outside of the main area of $\poly$. By Lemma~\ref{lem:guard_segment}, by taking $t^i_0,t^i_1,u^i_0,u^i_1$ as $t_0,t_1,u_0,u_1$ and $A\mydef \poly_M$, 
we obtain that if a guard set $G$ has exactly $4n$ guards placed in $\poly_M$, then it has a guard within $s_i$. 

For the third property, consider any variable $x_i \in X$ such that there is an equation $x_i=1$ in $\Phi$. As none of the points $q^i_1$ and $q^i_2$ can be seen from guards within the corridors or gadgets, or from the guards within the other guard segments on the base line, both of them must be seen by the only guard $g_i$ placed on $s_i$.
Therefore, $g_i$ must be placed at the position corresponding to the value $1$.

Finally, consider any set $G'$ of points such that
(i) $G'$ has a point within each of the $4n$ guard segments of $\poly_M$, and (ii) for each variable $x_i \in X$ such that there is an equation $x_i=1$ in $\Phi$ there is a guard at $s_i$ at the position corresponding to the value of $1$.
We now show that $G'$ is a guard set of $\poly_M$. From the construction of the pockets, a guard within a pocket containing $s_i$ can see the whole pocket
(in particular, the guards at positions corresponding to the value of $1$ can see all of the added spikes).
The guard at $s_i$ also sees all of $\poly_M$ which is above the pocket, i.e., all points of $\poly_M$ with $x$-coordinates in $[x(a_i)-6,x(a_i)+7.5]$.
The part of $\poly_M$ to the left of the leftmost pocket can be seen by the guard on the leftmost guard segment, and the part of $\poly_M$ to the right of the rightmost pocket can be seen by the guard on the rightmost guard segment.
\end{proof}


\subsection{Corridor for copying guard segments into a gadget}\label{sec:copy}

In this section, we describe the construction of a corridor, the purpose of which is to copy variables from the base line into a gadget.
Inside each gadget there are three (or two) guard segments $r_i, r_j , r_l$ (or $r_i, r_j$) corresponding to three (or two) pairwise different guard segments from the base line $s_i,s_j,s_l$ (or $s_i,s_j$).
We require that for each $\sigma \in \{i,j,l\}$ the guard segments $s_\sigma, r_\sigma$ have the same orientation. 
For the corridors attached at the right side of $\poly$ we assume $i < j < l$, and for the corridors attached at the left side we assume $i > j > l$. We describe here how to construct a corridor that ensures that the segments $r_i, r_j, r_l$ are copies of the segments $s_i, s_j, s_l$, respectively. This construction requires that the guard segments within the gadget satisfy the conditions of some technical lemmas (see Lemma~\ref{lemma:r-slabs} and \ref{lemma:rev-r-slabs}).

Note that this construction can be generalized for copying an arbitrary subset of guard segments, but since we only need
to copy two or three segments, we explain the construction in the setting of three segments. 
The construction for two segments is analogous but simpler. We first describe how to copy into a gadget at the right side of the polygon $\poly$; copying into gadgets at the left side of $\poly$ can be done in a symmetric way and is described shortly in Section~\ref{subsec:corridor-symmetric}.

As described briefly in Section~\ref{sec:poly}, the lower wall of the corridor of the gadget is a horizontal edge $c_0c_1$ of length $2$, where $c_0$ is on the line $\ell_r$ and $c_1$ is to the right of $c_0$. The upper wall of the corridor is more complicated, and it will be described later. It has the left endpoint at $d_0\mydef c_0+(0,\frac 3{C N^2})$, and the right endpoint $d_1\mydef c_1+(0,\frac {1.5}{C N^2})$.
The vertical line segments $c_0d_0$ and $c_1d_1$ are called the \emph{entrances} of the corridor.

\subsubsection{Idea behind the corridor construction}

\begin{figure}[htbp]
	\centering
	\includegraphics{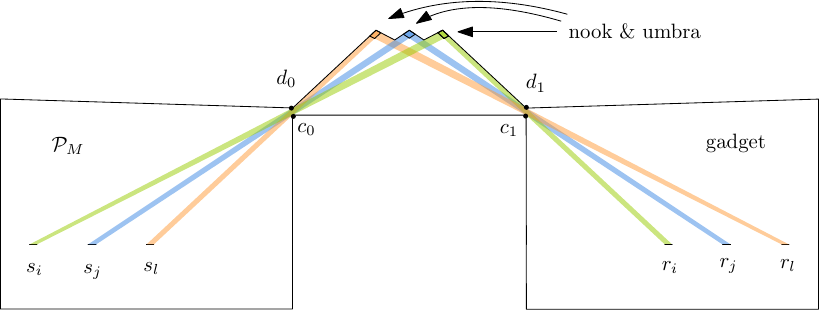}
	\caption{In this figure, we display a simplified corridor construction. The corners $c_0, c_1$ serve as shadow corners for three copy-umbras simultaneously for the pairs $(s_i,r_i),(s_j,r_j),(s_l,r_l)$. Each of these pairs also have a small copy-nook in the top of corridor.
	The entrances $c_0d_0$ and $c_1d_1$ to the corridor are sufficiently small so that the critical segments of the nook and umbra of each pair of segments $s_{\sigma},r_{\sigma}$ (contained in the small boxes at the top of the figure) are not seen by other guard segments.}
	\label{fig:SimplifiedCopying}
\end{figure}

To ensure that the segments $r_i, r_j, r_l$ are copies of the segments $s_i, s_j, s_l$, we need to construct within the corridor copy-\nook s and  copy-\umbra s for the pairs of corresponding segments, see Figure~\ref{fig:SimplifiedCopying} for a simplified illustration.
The corners $c_0,c_1$ of the corridor act as shadow corners in three overlapping copy-\umbra s for the pairs $(s_{i},r_i)$, $(s_{j},r_j)$, and $(s_{l},r_l)$, respectively.
We construct the chain of $\poly$ from $d_0$ to $d_1$ bounding the corridor from above so that it creates three copy-\nook s for the same pairs.
To enforce that for any guard set of size $g(\Phi)$, for each $\sigma\in\{i,j,l\}$ the guard segments $s_{\sigma}$ and $r_{\sigma}$ specify the same value, we have to ensure that no guards on guard segments other than $s_{\sigma}$ and $r_{\sigma}$
can see the critical segments of the copy-\umbra\ and the copy-\nook\ of the pair $s_{\sigma}, r_{\sigma}$.
For each $\sigma\in\{i,j,l\}$ we also introduce a stationary guard position, so that guards placed at these positions together see all the copy-\umbra s, but nothing on the other sides of the critical segments of the copy-\umbra s. We also need to ensure that the guards placed on the stationary guard positions cannot see the critical segments of the copy-\umbra s and the copy-\nook s of other pairs.
We then obtain (see Lemma~\ref{lemma:copy:works}) that for any guard set with one guard at each guard segment, and with no guards placed outside of the guard segments and stationary guard positions, the segments $s_{\sigma}, r_{\sigma}$ specify the value of $x_\sigma$ consistently.

Our construction will ensure that for any $\sigma\in\{i,j,l\}$, only the guard segment $s_\sigma$ from the base line, and only the guard segment $r_\sigma$ from within the gadget can see the critical segments of the corresponding copy-\nook\ and copy-\umbra.
In particular, we will ensure that the vertical edge of $\poly$ directly above the entrance $c_0d_0$ blocks visibility from all guard segments $s_{\sigma'}$ for $\sigma' \in \{1,\ldots, \sigma -1\}$, whereas the vertical edge of $\poly$ directly below $c_0d_0$ blocks visibility from all guard segments $s_{\sigma'}$ for $\sigma' \in \{\sigma +1,\ldots,4n\}$. An analogous property will be ensured for the gadget guard segments.

The main idea to achieve the above property is to make the entrances $c_i d_i$ of the corridor sufficiently small.
However, we cannot place the point $d_0$ arbitrarily close to $c_0$, since both endpoints $a_\sigma$ and $b_\sigma$ of the segment $s_\sigma$ have to see the left endpoint of the critical segment of the copy-\umbra, and the right endpoint of the critical segment of the copy-\nook\ for the pair $s_\sigma,r_\sigma$ (the points $f_0$ and $f_1$, respectively, in the context of Section~\ref{sec:nooks and umbras}). By placing the corridor sufficiently far away from the segments on the base line, we obtain that the visibility lines from the guard segment endpoints through the points $c_0,d_0$ are all almost parallel and can be described by a simple pattern.
The same holds for the pair of points $c_1,d_1$ and the endpoints of the guard segments $r_\sigma$, $\sigma\in\{i,j,l\}$.
The pattern enables us to construct the corridor with the desired properties.

In the following, we introduce objects that make it possible to describe the upper corridor wall and prove that the construction works as intended.

\subsubsection{Describing visibilities by a grid of slabs}
In a small area around the point $\frac{c_0+c_1}{2}+(0,1)$, every ray from an endpoint of a base line guard segment through one of the points $c_0,d_0$ intersects every ray from an endpoint of a gadget guard segment through one of the points $c_1,d_1$. 
These rays intersect at angles close to $\pi/2$, and they form an arrangement consisting of quadrilaterals, creating a nearly-regular pattern.
However, the arrangement of rays is not completely regular. 
We therefore introduce a collection of thin slabs, where each slab contains one of the rays in a small neighbourhood around $\frac{c_0+c_1}{2}+(0,1)$, and such that the slabs form an orthogonal grid with axis $(1,1)$ and $(-1,1)$.
Thus, the slabs are introduced in order to handle the ``uncertainty'' and irregularity of the rays.

Given a point $q$ and a vector $v$, the \emph{slab} $S(q,v,r)$ consists of all points at a distance of at most $r$ from the line through $q$ parallel to $v$.
The \emph{center} of the slab $S(q,v,r)$ is the line through $q$ parallel to $v$.

Let $r_i\mydef a'_ib'_i, r_j\mydef a'_jb'_j, r_l\mydef a'_lb'_l$.
Let $\ell_c$ be a vertical line passing through the middle of the segment $c_0c_1$.
Recall that here we describe the construction of a corridor to be attached at the right side of $\poly$.
Let $o$ be the intersection point of the ray $\overrightarrow{a_1 c_0}$ with $\ell_c$, and $o'$ the intersection point of the ray $\overrightarrow{b'_l c_1}$ with $\ell_c$. All the points $o, o', \frac{c_0+c_1}{2}+(0,1)$ lie on the vertical line $\ell_c$.

\label{def:slabs}
Let us define vectors $\alpha\mydef (1,1)$, $\beta\mydef (-1,1)$, and introduce a grid of slabs parallel to $\alpha$ and $\beta$.
Let us fix \[ \delta\mydef \frac{13.5}{C N^2}, \ \rho\mydef \frac{\delta}{9}=\frac{1.5}{C N^2}\text{, and }\varepsilon\mydef \frac{\rho}{12} = \frac{1}{8 C N^2}.\]
For each $\sigma\in\{1,\ldots,4n\}$ and $\gamma \in \{0,1,2,3\}$ we define a slab
$$L^{\gamma}_\sigma \mydef S(o+(0,(\sigma-1) \delta + \gamma \rho),\alpha,\varepsilon),$$
which we denote as an \emph{$L$-slab}.
Let $\tau(i)\mydef 2$, $\tau(j)\mydef 1$, and $\tau(l)\mydef 0$. For each $\sigma\in\{i,j,l\}$ and $\gamma \in \{0,1,2,3\}$ we define a slab 
$$R^{\gamma}_\sigma \mydef S(o'+(0,\tau(\sigma) \delta + \gamma \rho),\beta,\varepsilon),$$
which we denote as an \emph{$R$-slab}.

In the case of gadgets with just two guard segments $r_i,r_j$, we define the point $o'$ as the intersection point of the ray $\overrightarrow{b'_jc_1}$ with $\ell_c$, and we define $\tau(i)\mydef 1$ and $\tau(j)\mydef 0$.
Then, the $R$-slabs $R_\sigma^\gamma$ are defined as above for $\sigma\in\{i,j\}$.

See Figure~\ref{fig:DefineSlabs} for an illustration of the area where the $L$-slabs intersect the $R$-slabs.

\begin{figure}[htbp]
	\centering
	\includegraphics{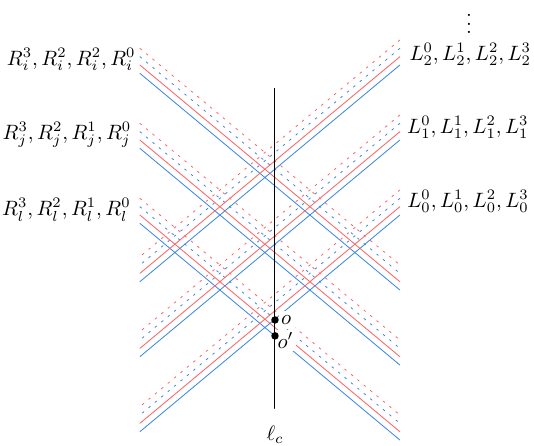}
	\caption{The $L$-slabs have slope $1$ and the $R$-slabs have slope $-1$. For each guard segment we get $4$ equidistant slabs. The width of each slab is $2\ee$. The distance between two slabs from the same group is $\rho$ and the distance between two groups is $\delta$. All intersections are contained in the region denoted by $V$.}
	\label{fig:DefineSlabs}
\end{figure}
Let $V$ be the square $\frac{c_0+c_1}{2}+(0,1)+[-38 N\rho,38 N\rho]\times [-38 N\rho,38 N\rho]$. 
We now prove that the area where the $L$-slabs and the $R$-slabs intersect is contained in $V$. 
Let us denote by $\mathcal{R}$ all the rays with endpoint at one of the guard segments in the main area, going through $c_0$ and $d_0$ and all the rays from the endpoints of gadget guard segments through the points $c_1,d_1$. 
We also ensure that all rays in $\mathcal{R}$ are inside a predefined slab within the area $V$.

In sections specific to the particular gadgets, we will prove the following lemma.
We state the lemma for gadgets with three guard segments $r_i,r_j,r_l$, but it has a natural analogue for gadgets with just two guard segments $r_i,r_j$.

\begin{lemma}\label{lemma:r-slabs}
For any gadget to be attached at the right side of the polygon $\poly$ and containing the guard segments $r_i\mydef a'_ib'_i, r_j\mydef a'_jb'_j, r_l\mydef a'_lb'_l$ the following holds, where $c_1$ is the bottom-right endpoint of the corridor corresponding to the gadget.
\begin{enumerate}
\item\label{rslabs:1} The intersection of any $R$-slab with the line $\ell_c$ is contained in $V$.
\item\label{rslabs:2} For each $\sigma\in\{i,j,l\}$, it holds that $\overrightarrow{b'_\sigma c_1}\cap V\subset R^{0}_\sigma,\ \overrightarrow{a'_\sigma c_1}\cap V\subset R^{1}_\sigma,\  \overrightarrow{b'_\sigma d_1}\cap V\subset R^{2}_\sigma$, and~$\overrightarrow{a'_\sigma d_1}\cap V\subset R^{3}_\sigma$.
\item\label{rslabs:3} There are no stationary guard positions or guard segments different from $r_i, r_j, r_l$ within the gadget from which any point of the corridor can be seen. \end{enumerate}
\end{lemma}

Assuming that the above lemma holds, we will prove the following.

\begin{lemma}\label{lemma:slabs}
For any corridor to be attached at the right side of the polygon $\poly$, the following properties are satisfied.
\begin{enumerate}
\item\label{lslabs:1}
The intersection of any $L$-slab with any $R$-slab is contained in $V$.
\item\label{lslabs:2}
For each $\sigma\in\{1,\ldots,4n\}$, it holds that $\overrightarrow{a_\sigma c_0}\cap V\subset L^{0}_\sigma, \overrightarrow{b_\sigma c_0}\cap V\subset L^{1}_\sigma, \overrightarrow{a_\sigma d_0}\cap V\subset L^{2}_\sigma$, and $\overrightarrow{b_\sigma d_0}\cap V\subset L^{3}_\sigma$.
\end{enumerate}
\end{lemma}

\begin{proof}
We will first show that the intersection of any $L$-slab with the line $\ell_c$ is contained in $V$. Consider 
Figure~\ref{fig:slabsProof1}. Recall that $o$ is the intersection point of the ray $\overrightarrow{a_1 c_0}$ with $\ell_c$. The horizontal distance between $\ell_c$ and $c_0$ is $1$. Let $u$ be the intersection point of the lines $\ell_b$ and $\ell_r$. From the polygon description in Section~\ref{sec:poly} we know that $\|a_1 u\| = CN^2$, and $\|c_0 u \| \in [CN^2, CN^2+3N]$.
The distance between the point $\frac{c_0 + c_1}{2}$ and the point $o$ is $\frac{\|c_0 u \|}{\|a_1 u\|} \in [1,1+\frac{3N}{CN^2}] = [1,1+2N\rho]$.
From the definition of the $L$-slabs, the intersection of the $L$-slabs with the vertical line $\ell_c$ is contained in the line segment $(o-(0,2\varepsilon), o+(0,4N\delta)) \subseteq (\frac{c_0 + c_1}{2}+(0,1-2\varepsilon), \frac{c_0 + c_1}{2}+(0,1+2N\rho+4N\delta))$, which is contained in $V$ as $\delta \mydef 9\rho$.

By Property~\ref{rslabs:1} of Lemma~\ref{lemma:r-slabs}, any point in the intersection of an $L$-slab and an $R$-slab must have a $y$-coordinate within the range of $V$. As the angles of the slabs are exactly $\pi/4$ and $3\pi/4$, we get that also the $x$-coordinates of the intersection must be within the range of $V$, see also to the right of Figure~\ref{fig:slabsProof1}. That gives us Property~\ref{lslabs:1}.


\begin{figure}[h]
\centering
\includegraphics[clip, trim = 2cm 0cm 0cm 0cm,scale=1.3]{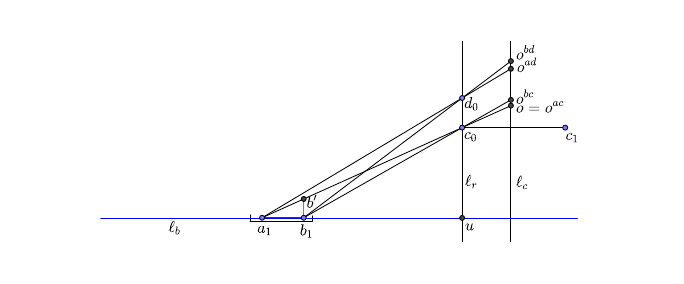}
\hskip-10pt
\includegraphics[clip, trim = 0cm -0.7cm 0cm 0cm,scale=1.3]{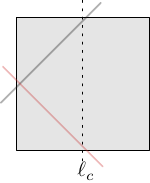}
\caption{Left: Structure of rays with origins $a_1$ and $b_1$, containing points $c_0$ and $d_0$. Right: Even in the case that a left ray intersects $\ell_c$ at the very top of $V$ and a right ray intersects at the very bottom of $V$, they still have to intersect within $V$.}
\label{fig:slabsProof1}
\end{figure}

For Property~\ref{lslabs:2}, let us first consider $\sigma=0$. Let us define $o^{ac}, o^{ad}, o^{bc}, o^{bd}$ as the intersection points of the rays $\overrightarrow{a_1 c_0}, \overrightarrow{a_1 d_0}, \overrightarrow{b_1 c_0}, \overrightarrow{b_1 d_0}$ with the line $\ell_c$ (see Figure \ref{fig:slabsProof1}). We have $o = o^{ac}$. The points $o^{ad}, o^{bc}, o^{bd}$ lie above $o$, and we will now estimate the distance between each of them and $o$. We do that as follows.

First, consider the distance $\|o o^{ad}\|$. From below, we have a trivial bound
\begin{align*}
\|o o^{ad}\|\geq \|c_0d_0\|=\frac{3}{C N^2} = 2\rho.
\end{align*}

From the similarity of triangles $a_1 c_0 d_0$ and $a_1 o o^{ad}$, and as the distance between the line $\ell_c$ and the line $\ell_r$ equals $1$, we obtain the following upper bound for $\|o o^{ad}\|$
\begin{align*}
\|o o^{ad}\|=\|c_0d_0\|\cdot\frac{\|a_1u\|+1}{\|a_1u\|}=
\frac 3{CN^2}\cdot \frac{CN^2+1}{CN^2}=
\frac 3{CN^2}\cdot \left(1+\frac{1}{CN^2}\right)\leq
2\rho+\frac{\varepsilon}{6}.
\end{align*}

Let $b'$ be a vertical projection of $b_1$ on the ray $\overrightarrow{a_1 c_0}$. From similarity of triangles $c_0 o o^{bc}$ and $c_0 b' b_1$, and triangles $a_1 b_1 b'$ and $a_1 u c_0$, we get the following equality

$$\|o o^{bc}\| = \|b_1 b'\| \cdot \frac{1}{\|b_1 u\|}=
\|c_0 u\| \cdot \frac{\|a_1 b_1\|}{\|a_1 u\|} \cdot \frac{1}{\|b_1 u\|} = 3/2 \cdot \frac{\|c_0 u\|}{\|a_1 u\| \|b_1 u\|}.$$

We instantly get
$$\|o o^{bc}\| \ge 3/2 \cdot \frac{C N^2}{(C N^2)^2} =\rho.$$
For an upper bound, we compute
$$\frac{\|a_1 u\| \|b_1 u\|}{\|c_0 u\|} \ge \frac{C N^2(CN^2-3/2)}{C N^2 + 3N}
\ge \frac{(C N^2 + 3N)(C N^2 - 4 N)}{C N^2 + 3N} = C N^2 - 4 N \ge \frac{C N^2}{1+1/72},$$
where the last inequality follows since $C \ge 73 \cdot 4 = 292$. That gives us 
$$\|o o^{bc}\| \le 3/2 \cdot \frac{1+1/72}{C N^2} = \frac{3/2}{C N^2}+\frac{1/48}{C N^2} = \rho+\frac{\varepsilon}{6}.$$

In the same way as for $\|o o^{bc}\|$ we obtain the following bounds
$$\rho \le \|o^{ad} o^{bd}\| \le \rho + \varepsilon/6.$$
As $\|o o^{bd}\| = \|o o^{ad}\| + \|o^{ad} o^{bd}\|$, we instantly get 
$$3 \rho \le \|o o^{bd}\| \le 3 \rho + \varepsilon/3.$$

Summarizing this part, we have the following bounds: $$\rho \le \|o o^{bc}\| \le \rho + \varepsilon/3,\ 2 \rho \le \|o o^{ad}\| \le 2 \rho + \varepsilon/3,\ 3 \rho \le \|o o^{bd}\| \le 3 \rho + \varepsilon/3.$$
Therefore, the intersection points of the rays $\overrightarrow{a_1 c_0}, \overrightarrow{b_1 c_0}, \overrightarrow{a_1 d_0}, \overrightarrow{b_1 d_0}$ are contained in the required slabs, at a vertical distance of at most $\varepsilon/3$ from the centers of the slabs.

\vskip10pt
Now, consider any $\sigma\in\{1,\ldots,4n\}$, $\tau \in \{a,b\}$ and $\eta \in \{c,d\}$. 
Let us denote by $\mathcal{R}$ all the rays with endpoint at one of the guard segments in the main area, going through $c_0$ and $d_0$ and all the rays from the endpoints of gadget guard segments through the points $c_1,d_1$.

Let $\widehat o$ be the intersection point of the ray $\overrightarrow{\tau_\sigma \eta_0}$ with the line $\ell_c$ (see Figure \ref{fig:slabsProof}).
\begin{figure}[h]
\centering
\includegraphics[width=\textwidth]{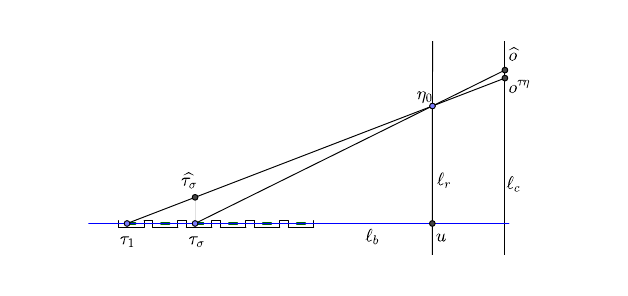}
\caption{Bounding $\|o \widehat o\|$.}
\label{fig:slabsProof}
\end{figure}
We first bound the distance $\|o^{\tau\eta} \widehat o\|$.
Recall that $u$ is the intersection point of $\ell_b$ and $\ell_r$.
Let $\widehat{\tau_\sigma}$ be the point on the ray $\overrightarrow{\tau_1 \eta_0}$ vertically above $\tau_\sigma$.

As the triangles $\eta_0 o^{\tau\eta} \widehat o$ and $\eta_0 \widehat{\tau_\sigma} \tau_\sigma$ are similar, the triangles $\tau_1 \tau_\sigma \widehat{\tau_\sigma}$ and $\tau_1 u \eta_0$ are similar, and the distance between the lines $\ell_c$ and $\ell_r$ is $1$, we get the following equality
$$\|o^{\tau\eta} \widehat o\| =
\| \tau_\sigma \widehat{\tau_\sigma}\| \cdot \frac 1{\|\tau_\sigma u\|}
 =\|u \eta_0\|\cdot \frac {\| \tau_1 \tau_\sigma\|} {\| \tau_1 u\|}\cdot \frac 1{\| \tau_\sigma u\|}
 =13.5 \sigma \cdot \frac{\|u \eta_0\|}{\| \tau_1 u\| \| \tau_\sigma u\|}.$$

We first bound $\|o^{\tau\eta} \widehat o\|$ from above. As $C \ge 1297\cdot58=75226$, we get
\begin{align*}
\frac{\| \tau_1 u\| \| \tau_\sigma u\|}{\|u \eta_0\|} & \ge \frac{(C N^2 - 3/2)(C N^2 - 54 N)}{C N^2 + 3N}
\ge \frac{(C N^2 + 3N)(C N^2 - 58 N)}{C N^2 + 3N} \\
& = C N^2 - 58 N \ge \frac{C N^2}{1+\frac{1}{1296N}},
\end{align*}
and, as $\sigma \le 4 N$, we can bound
$$\| o^{\tau\eta} \widehat o\| \le 13.5 \sigma \cdot \frac{1+\frac{1}{1296N}}{C N^2}
\le \sigma \cdot \frac{13.5}{C N^2}+\frac{1/24}{C N^2}=\sigma \delta + \frac{\varepsilon}{3}.$$

To bound $\|o^{\tau\eta} \widehat o\|$ from below, we compute 
$$
\| o^{\tau\eta} \widehat o\|= 
13.5 \sigma \cdot \frac{\|u \eta_0\|}{\| \tau_1 u\| \| \tau_\sigma u\|}\ge
13.5 \sigma \cdot \frac{CN^2}{(CN^2)^2}=
\sigma \delta.
$$

Therefore, as $\| o \widehat o\| = \| o^{\gamma\eta} \widehat o\| + \| o o^{\gamma\eta}\|$,
the intersection points of the rays $\overrightarrow{a_\sigma c_0}, \overrightarrow{b_\sigma c_0}, \overrightarrow{a_\sigma d_0}, \overrightarrow{b_\sigma d_0}$ are contained in the required slabs, at a vertical distance of at most $2 \varepsilon/3$ from the centers of the slabs.

\vskip10pt
We now need to verify that the rays stay in their respective slabs within the range of the $x$-coordinates of $V$.
We therefore bound the slope of a ray $\overrightarrow{\tau_\sigma \eta_0}$ for any $\sigma\in\{1,\ldots,4n\}$, $\tau \in \{a,b\}$ and $\eta \in \{c,d\}$. 
A bound from below is
$$\frac{\|u\eta_0\|}{\|\tau_\sigma u\|}\ge \frac{CN^2}{CN^2}\ge 1.$$

To bound the slope of $\overrightarrow{\tau_\sigma \eta_0}$ from above, we compute (as $C \ge 57\cdot1368+54=78030$)
$$
\frac{\|u \eta_0\|}{\| \tau_\sigma u\|}\leq
\frac{CN^2+3N}{CN^2-54N}\leq
1+\frac{57 N}{CN^2-54N}\leq
1+\frac{1}{1368 N}.
$$

The center of each $L$-slab has the slope equal to $1$.
As the vertical distance between $\widehat o$ and the center of the slab is at most $2\varepsilon/3$, and 
$\frac{1}{1368N} \cdot 38N \rho = \frac{\rho}{36} = \frac{\varepsilon}{3}$, we get that 
$\overrightarrow{\tau_\sigma \eta_0}\cap V$ is contained in the corresponding slab.
\end{proof}

\begin{figure}[htbp]
\centering
\includegraphics[clip, trim=1cm 0.6cm 1cm 0.9cm,width=0.95\textwidth]
{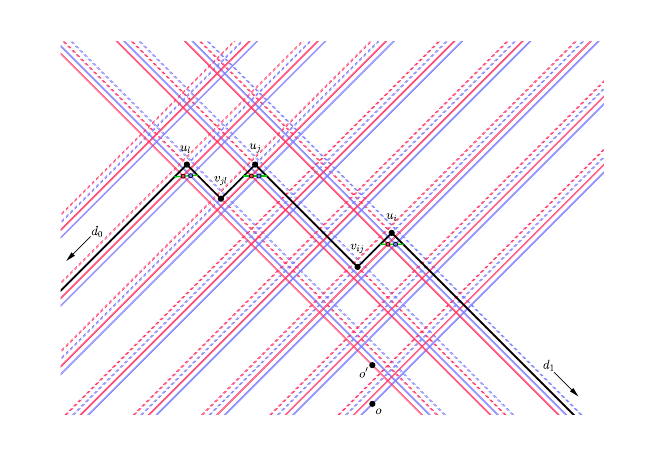}
\caption{The grid of $L$- and $R$-slabs (bounded by blue and red lines) and an approximate shape $\Lambda$ (black) of the upper wall of the corridor copying the guard segments $s_i,s_j,s_l$ to the guard segments $r_i,r_j,r_l$ in the gadget.
Here, there are $8$ guard segments $s_1,\ldots,s_8$ on the base line, and we have $i=3$, $j=6$, and $l=7$.
The blue lines bound the slabs corresponding to the rays originating at the left endpoints of the guard segments (i.e., slabs $L^{0}_\sigma, L^{2}_\sigma, R^{1}_\sigma, R^{3}_\sigma$), and the red lines bound the slabs corresponding to the rays originating at the right endpoints. The full lines bound the slabs corresponding to the rays passing through $c_0$ or $c_1$ (i.e., slabs $L^{0}_\sigma, L^{1}_\sigma, R^{0}_\sigma, R^{1}_\sigma$), and the dashed lines bound the slabs corresponding to the rays passing through $d_0$ or $d_1$.
The blue points and the red points are the intersections of the rays
$\protect\overrightarrow{a_\sigma c_0}\cap\protect\overrightarrow{a'_\sigma c_1}$
and $\protect\overrightarrow{b_\sigma c_0}\cap\protect\overrightarrow{b'_\sigma c_1}$, respectively,
for $\sigma\in\{i,j,l\}$. The green segments each contains a critical segment (the part between the blue and the red point) of a copy-\umbra\ for $s_\sigma$ and $r_\sigma$ with shadow corners $c_0,c_1$.
}
\label{fig:lattice}
\end{figure}

\subsubsection{Specification of the corridor}
We are now ready to describe the exact construction of the corridor.
As mentioned before, the bottom wall is simply the line segment $c_0c_1$.
We first describe the approximate shape of the upper wall, defined by a polygonal curve $\Lambda$ connecting the points $d_1$ and $d_0$. Later we will present how to modify $\Lambda$ into a final polygonal curve $\Lambda'$, which is exactly the upper wall of the corridor.

Note that in the corridor construction here we assume that $i < j < l$. In particular, the $L$-slabs $L^\gamma_l$ are above the $L$-slabs $L^\gamma_j$, which are above $L^\gamma_i$. For the $R$-slabs it is the opposite, i.e., the $R$-slabs $R^\gamma_l$ are below the $R$-slabs $R^\gamma_j$, which are below $R^\gamma_i$.

Figure~\ref{fig:lattice} shows the grid of slabs and a sketch of the curve $\Lambda$ approximating the upper wall (excluding most of the leftmost and rightmost edge of $\Lambda$, with endpoints at $d_0$ and $d_1$, respectively, since they are too long to be pictured together with the middle segments).
For $\sigma\in\{i,j,l\}$, let $u_\sigma$ be the intersection point of the rays $\overrightarrow{a_\sigma d_0}$ and $\overrightarrow{b'_\sigma d_1}$.
Let $v_{ij}$ be the intersection point of the rays $\overrightarrow{a_id_0}$ and $\overrightarrow{b'_jd_1}$, and $v_{jl}$ the intersection point of the rays $\overrightarrow{a_jd_0}$ and $\overrightarrow{b'_ld_1}$.
The curve $\Lambda$ is then a path defined by the points $d_1u_iv_{ij}u_jv_{jl}u_ld_0$.
By Lemma~\ref{lemma:slabs}, the set $\Lambda \cap V$ is contained in the union of the $L$-slabs and the $R$-slabs, as shown in Figure~\ref{fig:lattice}. Due to the relative position of the slabs $L^\gamma_l, L^\gamma_j, L^\gamma_i$ and $R^\gamma_l, R^\gamma_j, R^\gamma_i$ as discussed above, the curve $\Lambda$ is $x$-monotone, and the point $v_{ij}$ (resp.~$v_{jl}$) has smaller $y$-coordinate than the neighbouring points $u_i,u_j$ (resp.~$u_j,u_l$), i.e., the curve $\Lambda$ always has a zig-zag shape resembling the one from Figure~\ref{fig:lattice}.

We will now show how to modify $\Lambda$ by adding to the curve some features.
The first modification is in order to construct copy-\nook s $Q_i^n,Q_j^n,Q_l^n$ for each of the pairs $(s_i,r_i)$, $(s_j,r_j)$, and $(s_l,r_l)$, respectively.
Note that the area above $c_0c_1$ and below $\Lambda$ already contains a copy-\umbra\ $Q_\sigma^u$ for each pair $s_\sigma, r_\sigma$ for $\sigma \in \{i,j,l\}$ with shadow corners $c_0$ and $c_1$ (as $Q_\sigma^u$ is contained in the triangular area bounded above $c_0c_1$ and below the rays $\overrightarrow{b_\sigma c_0}, \overrightarrow{a'_\sigma c_1}$, which, due to Lemmas \ref{lemma:r-slabs} and \ref{lemma:slabs}, is below $\Lambda$).
The second reason why we need to modify $\Lambda$ is in order to create stationary guard positions $p_i,p_j,p_l$ that see the areas of the copy-umbras, but nothing above their critical segments.
In the following, we explain how to modify the fragment of $\Lambda$ consisting of the leftmost two edges, i.e., the path $v_{jl} u_l d_0$. The construction is presented in Figure~\ref{fig:lattice:closeup}.
We then perform similar modifications for the fragments of $\Lambda$  consisting of the paths $v_{ij}u_jv_{jl}$ and $d_1u_iv_{ij}$.

\begin{figure}[t]
\centering
\includegraphics[clip, trim = 0.8cm 0.6cm 0.8cm 0.6cm,width=0.8\textwidth]{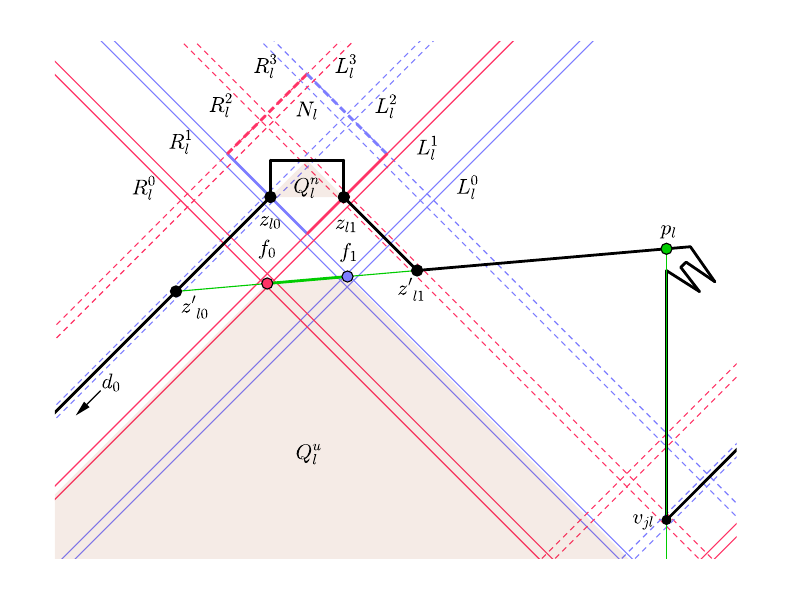}
\caption{A closeup of the upper wall of the corridor from Figure~\ref{fig:lattice} with features for copying the leftmost pair $s_l,r_l$ of guard segments.
The curve $\Lambda'$ is drawn in black.
The boundary of the square $N_l$ is drawn with thick line segments.
The points $z_{l0}, z_{l1}$ are the shadow corners of the copy-\nook\ $Q_l^n$ (brown area) of $s_l,r_l$.
The critical segment of $Q_l^n$ is on the topmost black segment and it is so short that it appears as if it was just a point.
The green line segment $f_0f_1$ is the critical segment of a copy-\umbra\ $Q_l^u$ of $s_l,r_l$ with shadow corners $c_0,c_1$.
The point $p_l$ is a stationary guard position, from which a guard can see the area below the segment $z_{l0}z_{l1}$ containing $f_0f_1$. Furthermore, $p_l$ sees the area to the left of the vertical ray emitting downwards from $p_l$.}
\label{fig:lattice:closeup}
\end{figure}

\begin{figure}[t]
\centering
\includegraphics[clip, trim = 1.1cm 0cm 1.1cm .3cm,width=0.6\textwidth]{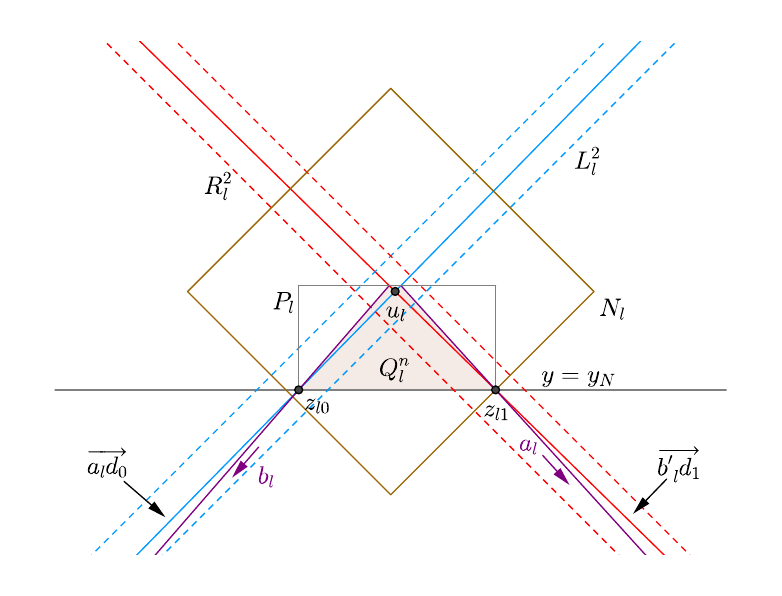}
\caption{Construction of the \nook\ $Q_l^n$.}
\label{fig:lattice:closeup:zoom}
\end{figure}

First, we show how to construct a copy-\nook\ $Q_l^n$ of $s_l$ and $r_l$ with shadow corners at $\Lambda$. The curve $\Lambda$ will then be modified so that $Q_l^n$ is contained within the corridor.
Let $N_l$ be the square consisting of points which are above the slabs $L^1_l$ and $R^1_l$, but not above $L^3_l$ or $R^3_l$. The square $N_l$ is approximately the area which is seen both from the right endpoint $b_l$ of $s_l$, and the left endpoint $a'_l$ of $r_l$.
Note that $N_l$ contains the point $u_l$ (as $u_l \in L^2_l \cap R^2_l$, and thus $u_l$ is above $L^1_l$ and $R^1_l$ and below $L^3_l$ and $R^3_l$).
The copy-\nook\ $Q_l^n$ for the pair $s_l,r_l$ will be created inside $N_l$ (see Figure \ref{fig:lattice:closeup:zoom}).
Consider the two intersection points of the boundary of $N_l$ with the line segments $v_{jl} u_l$ and $u_l d_0$.
Let $y_N$ be the larger of the $y$-coordinates of these two intersection points.
The shadow corners of the nook $Q_l^n$ are chosen as intersection points of the horizontal line $y=y_N$ with the line segments $v_{jl}u_l$ and $u_ld_0$, and they are denoted by $z_{l1}$ and $z_{l0}$, respectively.
In this way we ensure that both shadow corners are visible from any point within the segments $s_l$ and $r_l$, and that they define a copy-\nook\ $Q_l^n$ for the pair of segments $s_l,r_l$. Note that the entire nook $Q_l^n$ is contained in $N_l$, since by Lemmas~\ref{lemma:r-slabs} and \ref{lemma:slabs} no point on $s_l$ or $r_l$ can see a point above the slabs $L^3_l$ or $R^3_l$. We now modify the curve $\Lambda$ as follows. Let $P_l$ be a quadrilateral with two corners at $z_{l0}$ and $z_{l1}$, and such that it contains vertical edges incident to $z_{l0}$ and $z_{l1}$, and an edge containing the critical segment for $Q_l^n$. We modify $\Lambda$ so that between the points $z_{l0}$ and $z_{l1}$, it consists of the vertical edges and the topmost edge of $P_l$.

Now, consider the copy-\umbra\ $Q_l^u$ for the pair of segments $s_l, r_l$ with shadow corners $c_0$ and $c_1$.
Let $f_0\mydef \overrightarrow{b_l c_0}\cap\overrightarrow{b'_l c_1}$ and $f_1\mydef \overrightarrow{a_l c_0}\cap\overrightarrow{a'_l c_1}$. Note that the points $f_0, f_1$ correspond to the red and blue points in Figures \ref{fig:lattice} and \ref{fig:lattice:closeup}.
The segment $f_0f_1$ is the critical segment for $Q_l^u$.
By Lemmas~\ref{lemma:r-slabs} and \ref{lemma:slabs}, $f_0 \in L_l^{1}\cap R_l^{0}$ and $f_1 \in L_l^{0}\cap R_l^{1}$, and the squares $L_l^{1}\cap R_l^{0}, L_l^{0}\cap R_l^{1}$ have a sidelength of $2\varepsilon$. Therefore the slope of the line $\overleftrightarrow{f_0 f_1}$ is in the interval $\left[ - \frac{2 \sqrt{2}\varepsilon}{\sqrt{2} \rho},\frac{2 \sqrt{2}\varepsilon}{\sqrt{2} \rho}\right] = \left[ -1/6, 1/6\right]$, and this line intersects both line segments $v_{jl}u_l$ and $u_ld_0$. 
Let $z'_{l0}$ and $z'_{l1}$ be the intersection points of the line $\overleftrightarrow{f_0f_1}$ with the line segments $u_l d_0$ and $u_l v_{jl}$, respectively.
(We similarly define points $z'_{i0},z'_{i1}$ on $d_1u_iv_{ij}$ as the intersection points with the line containing the critical segment of the umbra $Q_i^u$, and $z'_{j0},z'_{j1}$ on $v_{ij}u_jv_{jl}$ as the intersection points with the line containing the critical segment of the umbra $Q_j^u$.) 
We introduce a stationary guard position $p_l$ by creating a pocket which will require modifying the curve $\Lambda$ again. 
The pocket is extruding to the right from $v_{jl}u_l$, following the line $\overleftrightarrow{f_0f_1}$, as pictured in Figure \ref{fig:lattice:closeup}.
Likewise, it is extruding vertically up from $v_{jl}$.
The pocket contains a stationary guard position $p_l$ on the line $\overleftrightarrow{f_0f_1}$.
Clearly, a guard placed at $p_l$ sees nothing above the line segment $f_0f_1$.
Note that it sees the part of $Q_l^u$ to the left of the vertical line through $v_{jl}$.

For the middle two edges $v_{ij}u_jv_{jl}$ of $\Lambda$, we place the stationary guard position $p_j$ vertically above $v_{ij}$ so that it sees an area below the critical segment of the umbra $Q_j^n$ and to the left of the vertical line through $v_{ij}$.
For the rightmost edges $d_1u_iv_{jl}$, we place the stationary guard position $p_i$ vertically above $d_1$ so that it sees an area below the critical segment of the umbra $Q_i^u$ to the left of the vertical line through $d_1$.
Let $\Lambda'$ be the wall obtained by doing the modifications to $\Lambda$ described here, and let $C$ be the corridor, that is, the area bounded by the lower wall $c_0c_1$, the upper wall $\Lambda'$ between $d_0$ and $d_1$, and by the vertical entrance segments $c_0d_0$ and $c_1d_1$. 
See Figure~\ref{fig:corridor:global} for a picture of the complete corridor.

\begin{figure}[t]
\centering
\includegraphics[clip, trim = 1cm 0.5cm 1cm 0.5cm,width=0.8\textwidth]{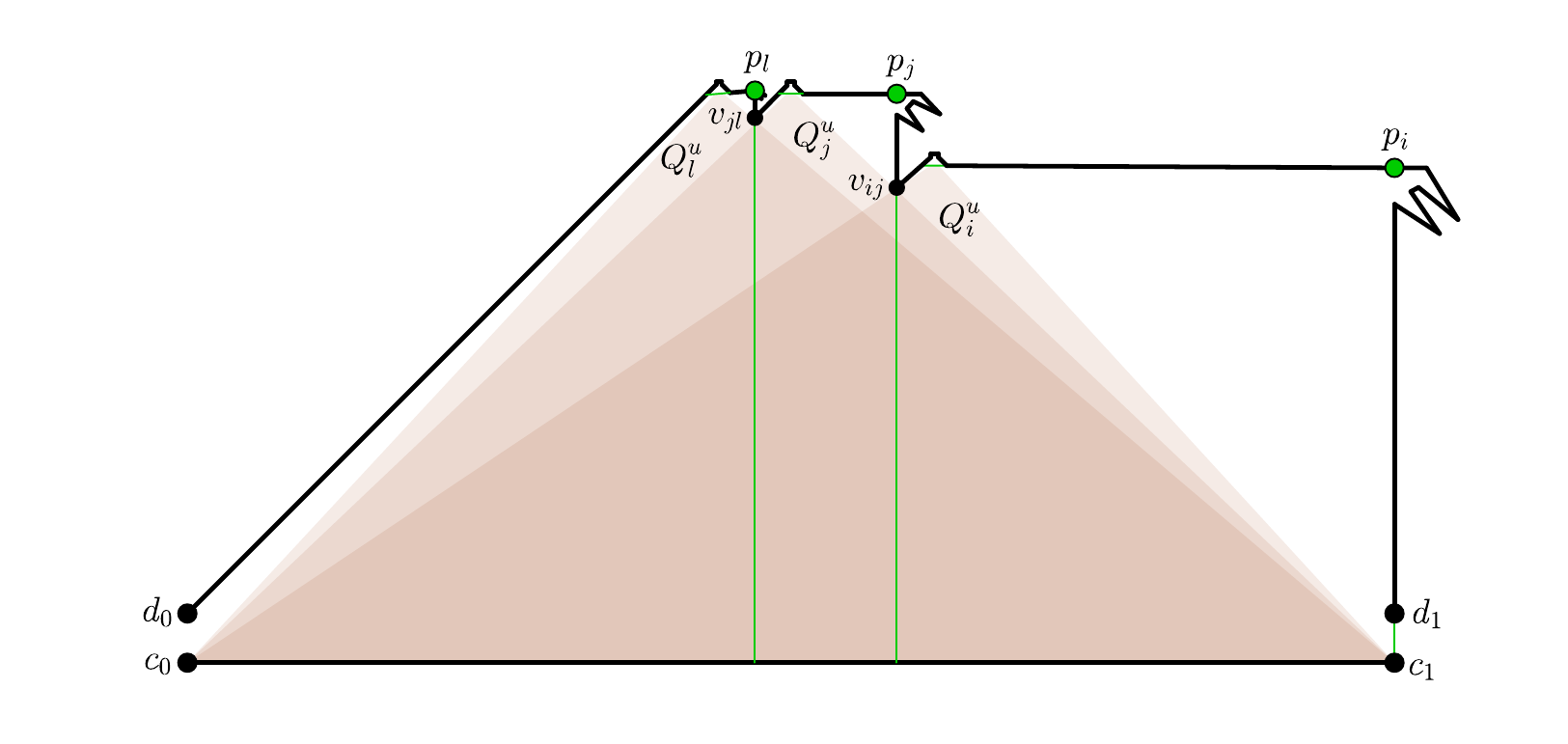}
\caption{The complete construction of the corridor.}
\label{fig:corridor:global}
\end{figure}


\begin{lemma}\label{lemma:corridor:stationary}
The stationary guard positions $p_i,p_j,p_l$ have the following three properties.
\begin{itemize}
\item
The three stationary guard positions $p_i,p_j,p_l$ together see all of the corridor except the points above the segments $z'_{i0}z'_{i1}$, $z'_{j0}z'_{j1}$, $z'_{l0}z'_{l1}$.

\item
None of the guards can see anything to the right of the right entrance $c_1d_1$.

\item
None of the stationary guard positions $p_i,p_j,p_l$ for the pairs $(s_i,r_i),(s_j,r_j),(s_l,r_l)$, respectively, can see any point on the critical segment of the \nook\ or \umbra\ of one of the other pairs.
\end{itemize}
\end{lemma}

\begin{proof}
See Figure~\ref{fig:corridor:global}.
For the first claim, note that the vertical lines through $v_{jl}$ and $v_{ij}$ divide the corridor into three parts.
It is now clear that all points in the leftmost part below $z'_{l0}z'_{l1}$ are seen by $p_l$, all points in the middle part below $z'_{j0}z'_{j1}$ are seen by $p_j$, and all points in the rightmost part below $z'_{i0}z'_{i1}$ are seen by $p_i$.

For the second part, observe that the point $p_i$ cannot see any point to the right of the vertical line through $d_1$, and the visibility of $p_j$ and $p_l$ is bounded by vertical lines more to the left.

For the last part, we note that the curve $\Lambda'$ passes through points $v_{jl}$ and $v_{ij}$, blocking visibility between stationary guard positions and critical segments corresponding to different pairs.
\end{proof}

The following lemma states that the corridor construction indeed has the properties needed for the guard segments to be copied into the gadget in the right way, and that it has a sufficiently small size for the corridors to be placed equidistantly on $\ell_r$ (or $\ell_l$ -- see Section~\ref{subsec:corridor-symmetric} for a discussion on gadgets at the left side of $\poly$) with distance $3$ without corridors or gadgets overlapping.

\begin{lemma}\label{lemma:copy:works}
Suppose that in each of the pairs $(s_i,r_i)$, $(s_j,r_j)$, $(s_l,r_l)$ of guard segments corresponding to a corridor $C$, the two segments have the same orientation.
Then $C$ satisfies the following properties.
\begin{enumerate}
\item\label{lemma:copy:works:3} 
In any guard set $G$ of $\poly$ there are at least $3$ guards placed within the corridor $C$, and if there are exactly $3$ then they are placed at the stationary guard positions $p_i,p_j,p_l$. (The number is $2$ instead of $3$ if we construct $C$ to copy only two segments.)

\item\label{lemma:copy:works:4}
Let $G$ be any set of points with exactly one guard on each guard segment and each stationary guard position, and with no guards outside of stationary guard positions and guard segments. If all of $C$ is seen by $G$, then for each of the pairs 
$(s_i,r_i)$, $(s_j,r_j)$, $(s_l,r_l)$ the two guards on the segments specify the same value.

\item\label{lemma:copy:works:5} For any set of points $G$ which satisfies the properties: (i) there is a guard at each point $p_i,p_j,p_l$ and at each guard segment $s_i,s_j,s_l$ and $r_i,r_j,r_j$, and (ii) the values specified by the pairs of segments $s_i,s_j,s_l$ and $r_i,r_j,r_j$ are consistent, $G$ sees all of $C$.

\item\label{lemma:copy:works:6} No guard at a stationary guard position or a guard segment outside the gadget can see any point inside the gadget below the line $\overleftrightarrow{d_0 c_1}$.

\item\label{lemma:copy:works:7} The vertical distance from $c_0c_1$ to the topmost point of the corridor is at most $1.4$.
\end{enumerate}
\end{lemma}

\begin{proof}
For Property~\ref{lemma:copy:works:3}, note that the points defining the stationary guards within $C$ can be seen only from within $C$.
We can now use Lemma~\ref{lem:stationary_guard}, setting $A$ as the corridor area and choosing the points defining the stationary guards to construct the set $M$, to prove the desired property. 

For Property~\ref{lemma:copy:works:4}, consider the set $G$ as described, and let $\sigma \in \{i,j,l\}$. 
The stationary guard positions $p_i,p_j,p_l$ cannot see any points above the line containing the critical segments of the \umbra\ $Q_\sigma^u$.
Lemma~\ref{lemma:corridor:stationary} and Property \ref{rslabs:3} of Lemma~\ref{lemma:r-slabs} give us that guards at the guard segments $s_1,s_2,\ldots,s_{4n},r_i,r_j,r_l$ must see the critical lines of the nook and umbra $Q_\sigma^n,Q_\sigma^u$.
We will now show that among all these segments, only guards placed on the segments $s_\sigma,r_\sigma$ are able to do so.
From Lemma~\ref{lemma:slabs} we get that for any $\sigma'>\sigma$, no guard on the guard segment $s_{\sigma'}$ can see a point in the square $V$ below $L_{\sigma'}^{0}$. As $L_{\sigma'}^{0}$ is above $L_{\sigma}^{3}$, it is also above the critical segments of the \nook\ $Q_\sigma^n$ and the \umbra\ $Q_\sigma^u$ of the pair $s_\sigma,r_\sigma$.
Likewise, for any $\sigma'<\sigma$, no guard on the guard segment $s_{\sigma'}$ can see a point in $V$ above $L_{\sigma'}^{3}$. As $L_{\sigma'}^{3}$ is below $L_\sigma^{0}$, it is also below the critical segments of the pair $s_\sigma,r_\sigma$.
A similar argument shows that among the guard segments $r_i,r_j,r_l$, only guards on $r_\sigma$ can see any points on the critical segments for the pair $s_\sigma,r_\sigma$. Therefore the two guards on $s_\sigma,r_\sigma$ must together see both critical segments of that pair, and by Lemma~\ref{lem:copy-lemma} the guards must specify the same value.


For Property~\ref{lemma:copy:works:5}, consider a set $G$ satisfying (i) and (ii). By Lemma~\ref{lemma:corridor:stationary}, the stationary guards can see all of $C$ except for the points which are above the line segments $z'_{i0}z'_{i1},z'_{j0}z'_{j1},z'_{l0}z'_{l1}$ containing the critical segments of the \umbra s. Consider any $\sigma \in \{i,j,l\}$, and the guard segments $s_\sigma,r_\sigma$. As these guards can see the complete critical segment for the umbra $Q_\sigma^u$, they can see the whole area contained above the critical segment and below $\Lambda$. As they can see the complete critical interval for the nook $Q_\sigma^n$, they can see all of the polygon $P_\sigma$. Therefore, they can see the whole area above the critical interval and below $\Lambda'$. 

Property~\ref{lemma:copy:works:6} is a clear consequence of Lemma~\ref{lemma:corridor:stationary}.

For Property~\ref{lemma:copy:works:7}, note that the top wall of the corridor can only extend beyond the square $V$ due to a part of the wall $\Lambda'$ creating a stationary guard position.
Recall that the slope of the critical segments of the umbras $Q_i^u,Q_j^u,Q_l^u$ is in the range $\left[ -1/6, 1/6\right]$.
Since $\|c_0c_1\|=2$, it follows that any point on $\Lambda'$ is at height at most $1+38 N \rho + 2\cdot 1/6<1.4$ above $c_0c_1$.

\end{proof}

\begin{lemma}\label{lem:rational-corridor}
Assume that the endpoints of guard segments corresponding to a corridor $C$ are at rational points, with the nominators and the denominators upper-bounded by $(\zeta CN^2)^{O(1)}$.
Then we can construct the corridor $C$ in such a way that each corner of $C$ has rational coordinates, with the nominator and the denominator upper-bounded by $(\zeta CN^2)^{O(1)}$. The corridor construction can be done in polynomial time.
\end{lemma}
\begin{proof}
Note that the entrances to the corridor are also at rational points, with nominators and denominators upper-bounded by $(\zeta CN^2)^{O(1)}$. Therefore, each of the lines defining the polygonal curve $\Lambda$ is defined by two rational points with this property. The same holds for the lines bounding the $L$-slabs and the $R$-slabs.

Consider the construction of a copy-nook $Q_\sigma^n$ within the corridor. The corners of $Q_\sigma^n$ are then at points which are defined as intersection of two lines, where each line is defined by two rational points with nominators and denominators upper-bounded by $(\zeta CN^2)^{O(1)}$. Therefore, the corners of the nook, and therefore also the corners of the quadrilateral $P_\sigma$ are also of this form.

The stationary guard positions $p_i,p_j,p_l$ are the intersection points of two lines, again each line defined by two points with the above property. Therefore, the nominators and denominators of the stationary guard positions are also upper-bounded by $(\zeta CN^2)^{O(1)}$. The corners of the pockets corresponding to $p_i,p_j,p_l$ can be chosen with much freedom, and therefore they can also be at points satisfying the lemma statement.
\end{proof}

\subsubsection{Corridor for a gadget at the left side of $\poly$}\label{subsec:corridor-symmetric}

For a gadget attached at the left side of the polygon $\poly$, the construction of the corridor is analogous. Now the points $c_0, d_0$ lie on the line $\ell_l$ instead of $\ell_r$, and the points $c_1, d_1$ are to the left of $c_0, d_0$. 
As we want the points $o$ ($o'$) to correspond to the lowest intersection point of a ray from an endpoint of a guard segment in the base line (in the gadget, respectively) containing the point $c_0$ ($c_1$, respectively) with $\ell_c$, we redefine these points in the following way.
The point $o$ is the intersection point of the ray $\overrightarrow{b_{4n} c_0}$ with $\ell_c$, and $o'$ is the intersection point of the ray $\overrightarrow{a'_l c_1}$ with $\ell_c$. As we want the slabs $L^{\gamma}_\sigma$ to contain fragments of rays from the endpoints of the segment $s_\sigma$, we redefine
$$L^{\gamma}_\sigma \mydef S(o+(0,(4n-\sigma) \delta + \gamma \rho),\beta,\varepsilon).$$
Similarly, for $\sigma\in\{i,j,l\}$, we redefine
$$R^{\gamma}_\sigma \mydef S(o'+(0,\tau(\sigma) \delta + \gamma \rho),\alpha,\varepsilon),$$
where $\tau\colon\{i,j,l\}\longrightarrow\{0,1,2\}$ is as defined for gadgets at the right side of $\poly$.

As now the left endpoints of the gadget guard segments are further away from the line $\ell_c$ than the right endpoints, each gadget attached to the left side of $\poly$ has to satisfy the following (instead of Lemma~\ref{lemma:r-slabs}), which will be proven in sections specific to the particular gadgets.

\begin{lemma}\label{lemma:rev-r-slabs}
For any gadget to be attached to the left side of the polygon  $\poly$ and containing the guard segments $r_i\mydef a'_ib'_i, r_j\mydef a'_jb'_j, r_l\mydef a'_lb'_l$ the following holds, where $c_1$ is the bottom-left endpoint of the corridor corresponding to the gadget.
\begin{enumerate}
\item\label{rev-rslabs:1} The intersection of any $R$-slab with the line $\ell_c$ is contained in $V$.
\item\label{rev-rslabs:2} For each $\sigma\in\{i,j,l\}$, it holds that $\overrightarrow{a'_\sigma c_1}\cap V\subset R^{0}_\sigma, \overrightarrow{b'_\sigma c_1}\cap V\subset R^{1}_\sigma, \overrightarrow{a'_\sigma d_1}\cap V\subset R^{2}_\sigma$, and $\overrightarrow{b'_\sigma d_1}\cap V\subset R^{3}_\sigma$.
\item\label{rev-rslabs:3} There are no stationary guard positions or guard segments different from $r_i, r_j, r_l$ within the gadget, from which any point of the corridor can be seen.
\end{enumerate}
\end{lemma}

For the same reason, instead of Lemma~\ref{lemma:slabs} we get the following lemma.

\begin{lemma}\label{lemma:rev-slabs}
Within any corridor to be attached to the left side of the polygon $\poly$ the following properties are satisfied.
\begin{enumerate}
\item\label{rev-lslabs:1}
The intersection of any $L$-slab with any $R$-slab is contained in $V$.
\item\label{rev-lslabs:2}
For each $\sigma\in\{1,\ldots,4n\}$, it holds that $\overrightarrow{b_\sigma c_0}\cap V\subset L^{0}_\sigma, \overrightarrow{a_\sigma c_0}\cap V\subset L^{1}_\sigma, \overrightarrow{b_\sigma d_0}\cap V\subset L^{2}_\sigma$, and $\overrightarrow{a_\sigma d_0}\cap V\subset L^{3}_\sigma$.
\end{enumerate}
\end{lemma}

Due to symmetry of $\poly_M$, the proof of Lemma~\ref{lemma:rev-slabs} is the same as the proof of Lemma~\ref{lemma:slabs}. The corridor construction and the proof of Lemma~\ref{lemma:copy:works} is then the same as for the corridor attached at the right side of $\poly$. The only difference is that in order to ensure that the $L$-slabs $L^\gamma_l$ are above the $L$-slabs $L^\gamma_j$, which are above $L^\gamma_i$ (which is required to get a meaningful zig-zag shape of the upper wall of the corridor) we now have to assume that $i > j > l$ (as the definition of the $L$-slabs has changed).


\subsubsection{Convenient tool to verify gadgets}

The following technical lemma turns out to be useful when proving that a gadget attached to the right side of $\poly$ satisfies Properties~\ref{rslabs:1} and \ref{rslabs:2} of Lemma~\ref{lemma:r-slabs}.

\begin{lemma}\label{lem:ugly}
Let $G$ be any gadget to be attached at the right side of the polygon $\poly$ such that the guard segments $r_i,r_j,r_l$ have a length of $\frac{3/2}{CN^2}$ and are contained in the box $m+[-\Delta, \Delta] \times [-\Delta, \Delta]$, where $m\mydef c_1+(1,-1)$ and $\Delta\mydef\frac{50}{CN^2}$, which is the axis-parallel square centered at $m$ with side length $2\Delta$.
Suppose furthermore that the left endpoint of $r_j$ is on the line through $a'_i+(\delta,0)$ parallel to the vector $(-1,1)$ and the left endpoint of $r_l$ is on the line through $a'_i+(2\delta,0)$ parallel to the same vector, where $a'_i$ is the left endpoint of $r_i$.
Then properties~\ref{rslabs:1} and \ref{rslabs:2} of Lemma~\ref{lemma:r-slabs} both hold for $G$.
\end{lemma}

\begin{proof}
Assume in this proof without loss of generality that $c_1\mydef (0,0)$.
Then $m=(1,-1)$.
Define $(\omega_{xi},\omega_{yi})\mydef CN^2\cdot(b'_i-m)$,
$(\omega_{xj},\omega_{yj})\mydef CN^2\cdot(b'_j-m)$, and 
$(\omega_{xl},\omega_{yl})\mydef CN^2\cdot(b'_l-m)$.
It follows from the conditions in the lemma that each of these values is in $[-50,50]$.

Define $o'_\sigma$ for each $\sigma\in\{i,j,l\}$ to be the intersection point of the ray $\overrightarrow{b'_\sigma c_1}$ with $\ell_c$.
Recall that the point $o'$ is defined as $o'_l$.
We first verify that Property~\ref{rslabs:1} holds, that is, each point $o'_\sigma$ is in the rectangle $V$.
We thus need to ensure that $y(o'_\sigma)\in[1-38N\rho,1+38N\rho]$.
The horizontal distance between $\ell_c$ and $c_1$ is $1$.
Therefore, the vertical distance between $c_0c_1$ and the point $o'_\sigma$ is $y(o'_\sigma)=\frac{-y(b'_\sigma)}{x(b'_\sigma)}=\frac{CN^2-\omega_{y\sigma}}{CN^2+\omega_{x\sigma}}$.
We have $y(o'_\sigma)\geq \frac{CN^2-50}{CN^2+50}$.
Note that the lower edge of the square $V$ has $y$-coordinate $1-38N\rho=1-\frac{57}{CN}$.
Now, the inequality $\frac{CN^2-50}{CN^2+50}\geq 1-\frac{57}{CN}$ is equivalent to $57CN^3+2850N\geq 100CN^2$, which is true for all $N\geq 2$.

Similarly, we have that $y(o'_\sigma)\leq \frac{CN^2+50}{CN^2-50}$.
The upper edge of $V$ has $y$-coordinate $1+38N\rho=1+\frac{57}{CN}$.
The inequality $\frac{CN^2+50}{CN^2-50}\leq 1+\frac{57}{CN}$ is equivalent to $100CN^2\leq 57CN^3-2850N$, which is also true for all $N\geq 2$.
Hence, Property~\ref{rslabs:1} holds.

Fix any $\sigma \in \{i,j,l\}$.
Recall that the vertical distance between the centers of two consecutive rays from $R^0_\sigma, R^1_\sigma, R^2_\sigma, R^3_\sigma$ is $\rho$.
We now show that the following four properties imply that Property~\ref{rslabs:2} holds.
Afterwards, we will prove that the four properties are satisfied.
\begin{enumerate}[label=\textbf{\alph*}]
\item\label{ugly_cond_0} The vertical distance from $o'_\sigma$ to the center of the slab $R^0_\sigma$ is at most $\varepsilon/4$,

\item\label{ugly_cond_1} the distance $d^{0}_\sigma$ between the intersection points of $\overrightarrow{b'_\sigma c_1}$ and $\overrightarrow{b'_\sigma d_1}$ with $\ell_c$ is in the interval $[2 \rho - \frac{\varepsilon}{4}, 2 \rho + \frac{\varepsilon}{4}]$,

\item\label{ugly_cond_2} for any symbol $\mu \in \{c,d\}$ the distance $d^{\mu}_\sigma$ between the intersection points of  $\overrightarrow{a'_\sigma \mu_1}$ and $\overrightarrow{b'_\sigma \mu_1}$ with $\ell_c$ is in $\left[ \rho - \frac{\varepsilon}{4}, \rho + \frac{\varepsilon}{4} \right]$, and

\item\label{ugly_cond_3} the absolute value of the slope of any ray with origin at an endpoint of $r_\sigma$ and passing through a point in $c_1d_1$ is in $\left[1-\frac{1}{38N\rho} \cdot \frac{\varepsilon}{4},1+\frac{1}{38N\rho} \cdot \frac{\varepsilon}{4}\right]$.
\end{enumerate}
The first three properties yield that all rays $\overrightarrow{b'_\sigma c_1}, \overrightarrow{a'_\sigma c_1}, \overrightarrow{b'_\sigma d_1}$, and $\overrightarrow{a'_\sigma d_1}$ intersect the line $\ell_c$ within their corresponding slabs at the vertical distance of at most $\frac{3 \varepsilon}{4}$ from the center of the slab.
The last property yields that the rays are contained in the corresponding slabs throughout the square $V$.

We now prove Property~\ref{ugly_cond_0}.
For $\sigma=l$, the distance is $0$ by definition.
We thus have to bound the distances $\|o' o'_i\|$ and $\|o' o'_j\|$.

Note that the conditions in the lemma gives that $y(b'_j)=y(b'_l)+x(b'_l)-x(b'_j)-\delta$ and $y(b'_i)=y(b'_l)+x(b'_l)-x(b'_i)-2\delta$.
We have 
  \begin{align*}
 \|o' o'_j\| & = y(o'_j)-y(o')
 = \frac{-y(b'_j)}{x(b'_j)} - \frac{-y(b'_l)}{x(b'_l)}
 = \frac{-y(b'_l)-x(b'_l)+x(b'_j)+\delta}{x(b'_j)} + \frac{y(b'_l)}{x(b'_l)}\\
& = \frac{-y(b'_l)-x(b'_l)+\delta}{x(b'_j)} + \frac{y(b'_l)+x(b'_l)}{x(b'_l)}\\
& = \frac{2CN^2-\omega_{yl}-\omega_{xl}+CN^2\delta}{CN^2+\omega_{xj}}+\frac{-2CN^2+\omega_{yl}+\omega_{xl}}{CN^2+\omega_{xl}} \\
& \in \left[\frac{CN^2\delta}{CN^2+50},\frac{CN^2\delta}{CN^2-50}\right]
\subset \left[\frac{CN^2\delta}{CN^2+50N^2},\frac{CN^2\delta}{CN^2-50N^2}\right] \\
& = \left[\frac{4000}{4001}\cdot \delta,\frac{4000}{3999}\cdot \delta\right]
\subset \left[\delta-\varepsilon/4,\delta+\varepsilon/4\right].
 \end{align*}
 A similar argument gives $\|o'o'_i\|\in [2\delta-\varepsilon/4,2\delta+\varepsilon/4]$.

We will prove Property~\ref{ugly_cond_1} as follows.
We have $|x(b'_\sigma)-x(m)| \le \Delta \le \frac{1}{1500}$ (as $C > 75000$), and 
\begin{align*}
d^{0}_\sigma & =\| c_1d_1\| \cdot \frac{2+x(b'_\sigma)-x(m)}{1+x(b'_\sigma)-x(m)}\\
& \in \left[ \rho \cdot \frac{2-\frac{1}{1500}}{1+\frac{1}{1500}}, \rho \cdot \frac{2+\frac{1}{1500}}{1-\frac{1}{1500}} \right] =
\left[\rho\cdot \left(2-\frac{2}{1501}\right),\rho\cdot \left(2+\frac{3}{1499}\right)\right] \subset \left[2 \rho - \frac{\varepsilon}{4}, 2 \rho + \frac{\varepsilon}{4}\right].
\end{align*}

For Property~\ref{ugly_cond_2}, denote $H$ as the distance between the point $a'_\sigma$ and its vertical projection on the ray $\overrightarrow{b'_\sigma \mu_1}$.
We have
$\frac{d^{\mu}_\sigma}{H} = \frac{1}{1+x(a'_\sigma)-x(m)} =
\frac{1}{1+x(b'_\sigma)-x(m)-\rho}$,
$\frac{H}{\rho} = \frac{1+y(m)-y(b'_\sigma)+\|\mu_1 c_1\|}{1+x(b'_\sigma)-x(m)}$ and $\rho\leq\frac 1{50000}$ (as $C\geq 75000$),
and therefore 
\begin{align*}
d^{\mu}_\sigma & = {\rho} \cdot \frac{1+y(m)-y(b'_\sigma)+\|\mu_1 c_1\|}{1+x(b'_\sigma)-x(m)} \cdot \frac{1}{1+x(b'_\sigma)-x(m)-\rho} \\
& \in
\left[
\rho \cdot \frac{1-\frac{1}{1500}}{1+\frac{1}{1500}}\cdot
\frac 1{1+\frac{1}{1500}},
\rho \cdot \frac{1+\frac{1}{1500}+\frac{1}{50000}}{1-\frac{1}{1500}}\cdot
\frac 1{1-\frac{1}{1500}-\frac 1{50000}} \right]\\
&=\left[\rho\cdot \left(1-\frac{4501}{2253001}\right),
\rho\cdot \left(1+\frac{458897}{224695603}\right)\right] \\
& \subset
\left[ \delta - \frac{\varepsilon}{4}, \delta + \frac{\varepsilon}{4} \right] .
\end{align*}

For Property~\ref{ugly_cond_3}, we have to bound the slope of the rays.
Note that $\frac{1}{38N\rho} \cdot \frac{\varepsilon}{4}=\frac 1{1824N}$.
Since $C\geq 50\cdot 1824=91200$, we get that the slope is at least
$$\frac{y(c_1)-y(b'_\sigma)}{x(b'_\sigma)-x(c_1)} = \frac{1+y(m)-y(b'_\sigma)}{1+x(b'_\sigma)-x(m)}\geq
\frac{1-\Delta}{1+\Delta}=1-\frac{50}{CN^2+50}\geq
1-\frac{50}{CN}\geq 1-\frac 1{1824N}.$$

On the other hand, since $C\geq 101.5\cdot 1824+50=185186$, we get that the slope is at most
$$\frac{y(d_1)-y(b'_\sigma)}{x(b'_\sigma)-x(d_1)} \leq 
\frac{1+\Delta+\rho}{1-\Delta} =
1+\frac{101.5}{CN^2-50}\leq
1+\frac{101.5}{(C-50)N}\leq 1+\frac 1{1824 N}.$$

Since all four properties hold, so does Property~\ref{rslabs:2}.
\end{proof}

For gadgets attached to the left side of $\poly$, we have the following symmetric version.

\begin{lemma}\label{lem:rev-ugly}
Let $G$ be any gadget to be attached at the left side of the polygon $\poly$ such that the guard segments $r_l,r_j,r_i$ have a length of $\frac{3/2}{CN^2}$ and are contained in the box $m+[-\Delta, \Delta] \times [-\Delta, \Delta]$, where $m\mydef c_1+(-1,-1)$ and $\Delta\mydef\frac{50}{CN^2}$, which is the axis-parallel square centered at $m$ with side length $2\Delta$.
Suppose furthermore that the left endpoint of $r_j$ is on the line through $a'_i+(-\delta,0)$ parallel to the vector $(1,1)$ and the left endpoint of $r_l$ is on the line through $a'_i+(-2\delta,0)$ parallel to the same vector, where $a'_i$ is the left endpoint of $r_i$.
Then properties~\ref{rev-rslabs:1} and \ref{rev-rslabs:2} of Lemma~\ref{lemma:rev-r-slabs} both hold for $G$.
\end{lemma}

\subsection{The $\ge$-addition gadget}\label{sec:additionGadget}

In this section we present the construction of the $\geq$-addition gadget which represents an inequality $x_i+x_j\ge x_l$, where $i,j,l\in\{1,\ldots,n\}$.
In Section~\ref{sec:additionGadget2} we will show how to modify the construction to obtain the $\leq$-addition gadget for the inequality $x_i+x_j \le x_l$. For any equation of the form $x_i+x_j = x_l$ in $\Phi$ we will then add both gadgets to our polygon $\poly$.

\subsubsection{Idea behind the gadget construction}\label{subsec:addition-idea}
We first describe the general idea behind the construction of a gadget imposing an inequality $x'_i+x_j\geq x_l$ for three variables $x'_i,x_j,x_l$.
See Figure~\ref{fig:additionPrinciple} for a sketch of the construction.
Let $w,v,h>0$ be rational values such that $w>v+3/2$. 
Let $r'_i,r_j,r_l$ be right-oriented guard segments\footnote{We use $r'_i$ instead of $r_i$ here (and $x'_i$ instead of $x_i$), as the guard segment $r'_i$ specifying the value $x'_i$ will only be a weak copy of the segment from the base line with a value $x_i$, i.e., it will hold that $x'_i\leq x_i$. More details will be provided later.}
 of length $3/2$ such that $r_j$ has its left endpoint at the point $(-w,0)$, $r'_i$ has its right endpoint at $(w,0)$, and $r_l$ has its left endpoint at $(-2,-h)$. 
Let $g'_i\mydef (w-2+x'_i,0)$, $g_j\mydef (-w-1/2+x_j,0)$, and $g_l\mydef (-5/2+x_l,-h)$ be three guards on $r'_i,r_j,r_l$, respectively, representing the values $x'_i,x_j,x_l \in [1/2,2]$.

\begin{figure}[h]
\centering
\includegraphics[clip, trim = 0cm 1cm 0cm 0cm,scale=0.7]{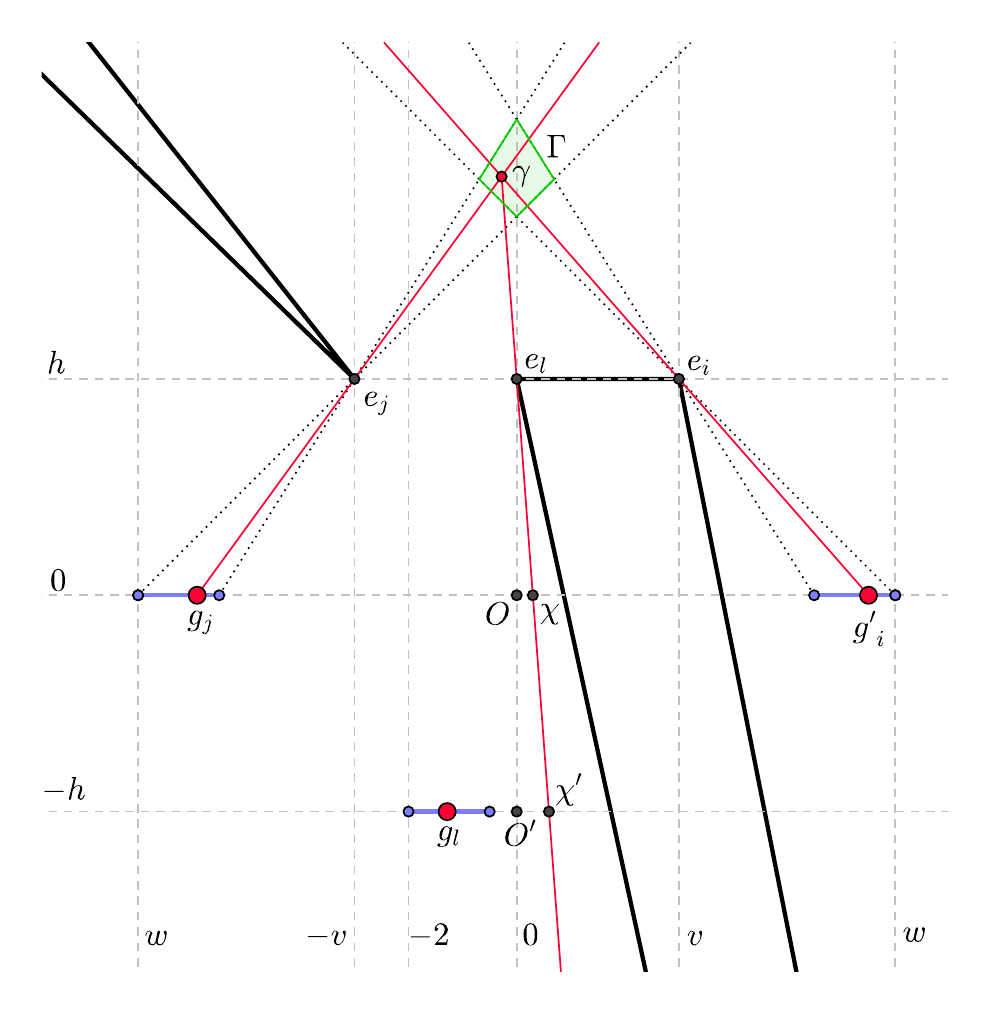}
\caption{The thick black segments are edges of the polygon.
In order to see the quadrilateral $\Gamma$ together with $g'_i$ and $g_j$, the guard $g_l$ must be on or to the left of the point $\chi'$.}
\label{fig:additionPrinciple}
\end{figure}

Let $e_i\mydef (v,h),e_j\mydef (-v,h),e_l\mydef (0,h)$.
Let $\Gamma$ be a collection of points $\gamma$ such that the ray $\overrightarrow{\gamma e_i}$ intersects $r'_i$, and the ray $\overrightarrow{\gamma e_j}$ intersects $r_j$.
Then $\Gamma$ is a quadrilateral, bounded by the following rays: the rays with origin at the endpoints of $r'_i$ and containing $e_i$, and the rays with origin at the endpoints of $r_j$ and containing $e_j$.
Suppose that
\begin{itemize}
\item
for every point $g'_i$ on $r'_i$ and $\gamma$ in $\Gamma$, the points $\gamma$ and $g'_i$ can see each other if and only if $\gamma$ is on or to the right of the line $\overleftrightarrow{g'_i e_i}$,

\item
for every point $g_j$ on $r_j$ and $\gamma$ in $\Gamma$, the points $\gamma$ and $g_j$ can see each other if and only if $\gamma$ is on or to the right of the line $\overleftrightarrow{g_j e_j}$,

\item
for every point $g_l$ on $r_l$ and $\gamma$ in $\Gamma$, the points $\gamma$ and $g_l$ can see each other if and only if $\gamma$ is on or to the left of the line $\overleftrightarrow{g_l e_l}$.
\end{itemize}
These properties are enforced by polygon edges in Figure~\ref{fig:additionPrinciple}.
Under these assumptions, the following lemma holds.

\begin{lemma}\label{lem:inequality}
The guards $g'_i, g_j, g_l$ can together see the whole quadrilateral $\Gamma$ if and only if $x'_i + x_j \ge x_l$.
\end{lemma}

\begin{proof}
Let $\gamma \in \Gamma$ be the intersection point of the rays $\overrightarrow{g'_ie_i}$ and $\overrightarrow{g_je_j}$.

Suppose that the guards $g'_i,g_j,g_l$ together see the whole quadrilateral $\Gamma$.
Since $g'_i$ cannot see the area to the left of the line $\overleftrightarrow{\gamma g'_i}$, and $g_j$ cannot see the area to the left of the line $\overleftrightarrow{\gamma g_j}$, there are points arbitrarily close to $\gamma$ which are not seen by any of the guards $g'_i, g_j$. Therefore, $g_l$ has to see $\gamma$.

Consider the rays $\overrightarrow{\gamma e_i}$, $\overrightarrow{\gamma e_j}$, and $\overrightarrow{\gamma e_l}$. Let $\chi$ be the intersection point of the ray $\overrightarrow{\gamma e_l}$ with a horizontal line $y=0$, and $\chi'$ the intersection point of the ray $\overrightarrow{\gamma e_l}$ with a horizontal line $y=-h$. Note that the guard $g_l$ can see the point $\gamma$ if and only if $g_l$ is coincident with $\chi'$ or to the left of $\chi'$.

From the similarity of triangles $g_j g'_i \gamma$ and $e_j e_i \gamma$ we get that $\frac{y(\gamma)}{y(\gamma)-h}=\frac{\|g'_i-g_j\|}{2v}=\frac{2w+x'_i-x_j-3/2}{2v}$. From the similarity of triangles $g_j \chi \gamma$ and $e_j e_\ell \gamma$ we get that $\frac{y(\gamma)}{y(\gamma)-h}=\frac{\|\chi - g_j\|}{v}$, and therefore $\|\chi - g_j\|= w+x'_i/2-x_j/2-3/4$, and $\chi=(x'_i/2+x_j/2-5/4,0)$.
Let $O\mydef (0,0)$ and $O'\mydef (0,-h)$.
From the similarity of triangles $O \chi e_l$ and $O' \chi' e_l$ we get that $\chi'=(x'_i+x_j-5/2,0)$. The condition that the guard $g_l$ is coincident with $\chi'$ or to the left of $\chi'$ is equivalent to $-5/2+x_l \le x'_i+x_j-5/2$, i.e., $x'_i + x_j \ge x_l$.

On the other hand, if $x'_i+x_j \ge x_l$ then the guard $g_l$ is coincident with $\chi'$ or to the left of $\chi'$, and therefore $g_l$ can see $\gamma$. Then the guards $g'_i, g_j, g_l$ can together see the whole $\Gamma$.
\end{proof}

\subsubsection{Fragment of the gadget for testing the inequality}\label{subsec:ineq-testing-gadget}
We now present the construction of a polygon $\poly_{\textrm{ineq}}$ containing three guard segments $r'_i, r_j, r_l$ and corners $e_i,e_j,e_l$ with coordinates as described above, where we set $w\mydef 26$, $v\mydef 10$, and $h\mydef 10.5$, enforcing an inequality on the values corresponding to the guard segments. The main part of the polygon $\poly_{\textrm{ineq}}$ is pictured in Figure~\ref{fig:additionGadget2Principle}.
The three blue line segments correspond to the guard segments, and the three green dots to stationary guard positions. 
The stationary guard positions $g_t,g_m,g_b$ have been chosen as follows. The point $g_t$ is at the ray with origin at the right endpoint of $r_j$ and containing $e_j$, $g_b$ is at the ray with origin at the right endpoint of $r_j$ and containing $e_l$, and $g_m$ is at the ray with origin at the right endpoint of $r'_i$ and containing $e_i$.
For each guard segment and each stationary guard position we 
introduce pockets of the polygon, such that the points defining each guard segment and stationary guard position cannot be seen 
from any other guard segment or stationary guard position.
Two edges in the left of the figure are only shown partially. They end to the left at corners $c_1\mydef(-CN^2, CN^2)$ and $d_1\mydef(-CN^2, CN^2+1.5)$, respectively.
These corners are connected by an edge $c_1d_1$ which closes the polygon.
We can show the following result.

\begin{figure}[h]
\centering
\includegraphics[width=.9\textwidth]{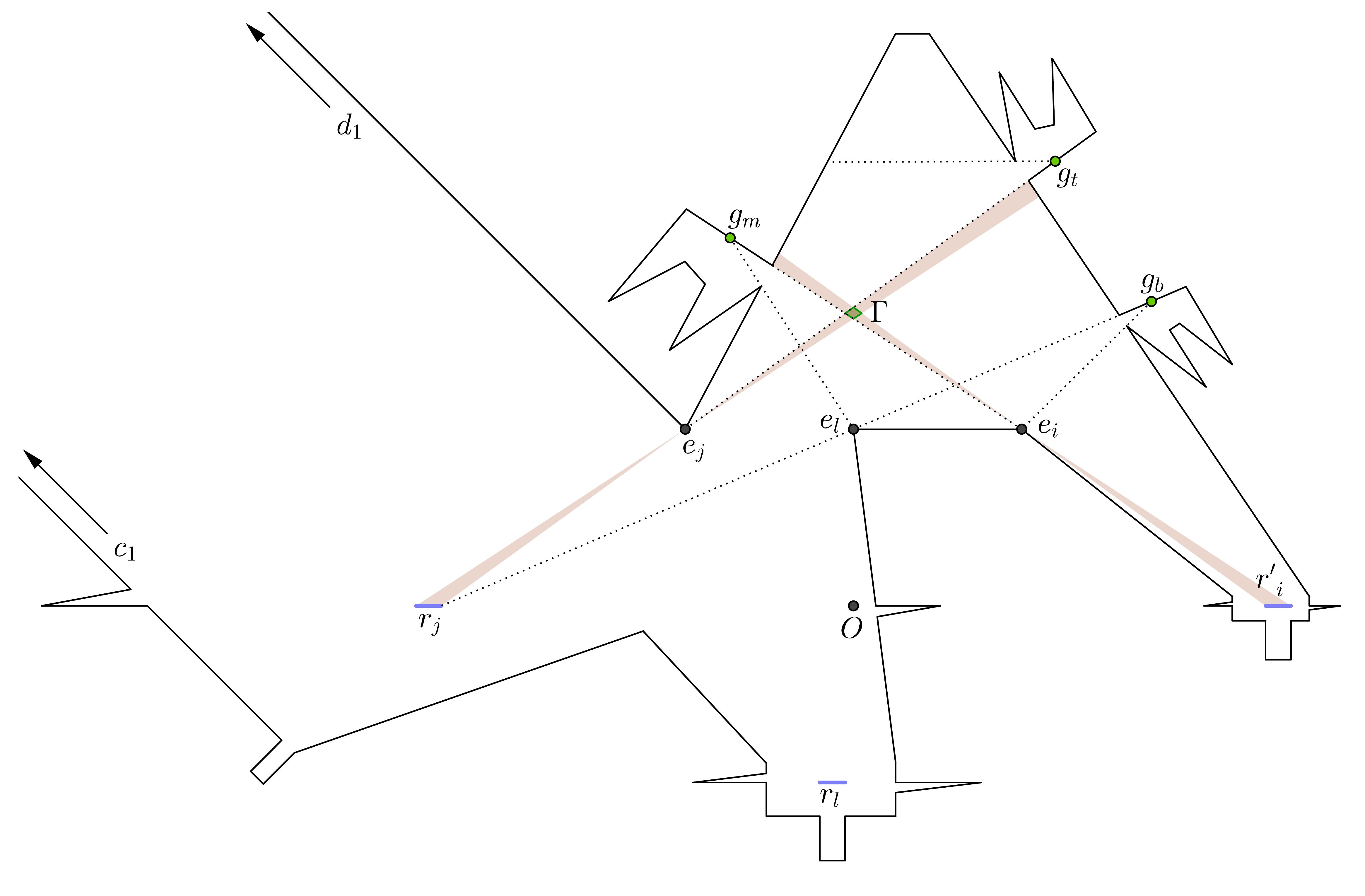}
\caption{The main part of a polygon $\poly_{\textrm{ineq}}$.
The quadrilateral $\Gamma$ consists of all the intersection points $\protect\overrightarrow{p_je_j}\cap\protect\overrightarrow{p_ie_i}$ where $p_j\in r_j$ and $p_i\in r'_i$.
}
\label{fig:additionGadget2Principle}
\end{figure}

\begin{lemma}\label{lem:addition-gadget-7-guards}
Consider the polygon $\poly_{\textrm{ineq}}$ from Figure \ref{fig:additionGadget2Principle}. A set of guards $G \subset \poly_{\textrm{ineq}}$ of cardinality at most $6$ guards $\poly_{\textrm{ineq}}$ if and only if
\begin{itemize}
\item there is exactly one guard placed on each guard segment $r'_i, r_j, r_l$ and at each stationary guard position $g_t,g_m,g_b$, and
\item the variables $x'_i, x_j, x_l$ corresponding to the guard segments $r'_i, r_j, r_l$, respectively, satisfy the inequality $x'_i + x_j \ge x_l$.
\end{itemize}
\end{lemma}

\begin{proof}
Assume first that the two conditions are satisfied.
Observe that the stationary guard positions $g_t,g_m,g_b$ have been chosen so that guards placed at them cannot see any point in the interior of $\Gamma$, but together with the guards $g'_i,g_j,g_l$ placed on $r'_i, r_j, r_l$ can see all of $\poly_{\textrm{ineq}} \setminus \Gamma$.
(For this property to hold, the position of $g'_i,g_j,g_l$ within the corresponding guard segments does not matter.)
Since additionally the inequality $x'_i + x_j \ge x_l$ is satisfied, then Lemma~\ref{lem:inequality} yields that all of $\Gamma$ is seen, and hence $G$ guards $\poly_{\textrm{ineq}}$.

Assume now that a set $G$ of at most $6$ guards sees all of $\poly_{\textrm{ineq}}$.
Using Lemmas~\ref{lem:stationary_guard} and \ref{lem:guard_segment} (where we set $A\mydef\poly_{\textrm{ineq}}$ and choose the points in $M$ among the following: one point $t_1$ defining each stationary guard position, and one point $q_2$ defining each guard segment) we can show that the polygon requires at least $6$ guards, and if there are $6$ guards then there must be one guard at each stationary guard position, and at each guard segment.
Then, as $\Gamma$ is seen by the guards, by Lemma~\ref{lem:inequality} we get that the inequality $x'_i + x_j \ge x_l$ holds.
\end{proof}

In the actual $\geq$-addition gadget we modify $\poly_{\textrm{ineq}}$ in a way described later and scale it down by a factor of $\frac{1}{C N^2}$.
Then we connect it to the main part of the polygon $\poly$. The polygon $\poly_{\textrm{ineq}}$ is attached at the right side of $\poly$, and the connection between $\poly_{\textrm{ineq}}$ and the main area of $\poly$ is via a corridor. The point which corresponds to $O\mydef (0,0)$ in the polygon  $\poly_{\textrm{ineq}}$ is then at the position $m\mydef c_1+(1,-1)$ in $\poly$, where $c_1$ is the bottom-right corner of the corridor.

\subsubsection{Completing the gadget}
Let $r'_i,r_j,r_l$ denote the guard segments of $\poly_{\textrm{ineq}}$. We now show how to enforce dependency between appropriate guard segments from the base line of the polygon $\poly$ and the guard segments $r'_i,r_j,r_l$. 

We ensure that the guard segments $r_j$ and $r_l$ are copies of segments from the base line
by connecting the (scaled) polygon $\poly_{\textrm{ineq}}$ to the polygon $\poly$ via a corridor.
However, the segment  $r'_i$ cannot be copied in this way, as the edges of $\poly_{\textrm{ineq}}$ are blocking visibility.
Instead, we introduce an additional guard segment $r_i$ within the gadget, and by introducing a copy-nook within the construction of $\poly_{\textrm{ineq}}$, we ensure that $r'_i$ is a weak copy of $r_i$.
This is explained in detail in Section~\ref{subsec:copying-into-rj}.
In Section~\ref{subsec:corridor-consistency}, we explain how to copy appropriately chosen segments $s_i,s_j,s_l$ from the base line into $r_i,r_j,r_l$ using a corridor.
In Section~\ref{subsec:summary}, we summarize the properties of the constructed gadget.

The gadget is scaled by a factor of $\frac{1}{CN^2}$ before it is attached to the corridor, and the points $c_1, d_1$ of the gadget are coincident with the points defining the right entrance of the corridor, which have the same names.
After this operation, the middle point of the gadget (i.e., the point corresponding to $O\mydef (0,0)$ in the coordinate system of the gadget) satisfies the equality $m=c_1+(1,-1)$, as stated in Section~\ref{sec:poly}.
For the picture of the complete gadget see Figure~\ref{fig:additionGadgetDetail}.

\begin{figure}[h]
\centering
\includegraphics[width=0.9\textwidth]{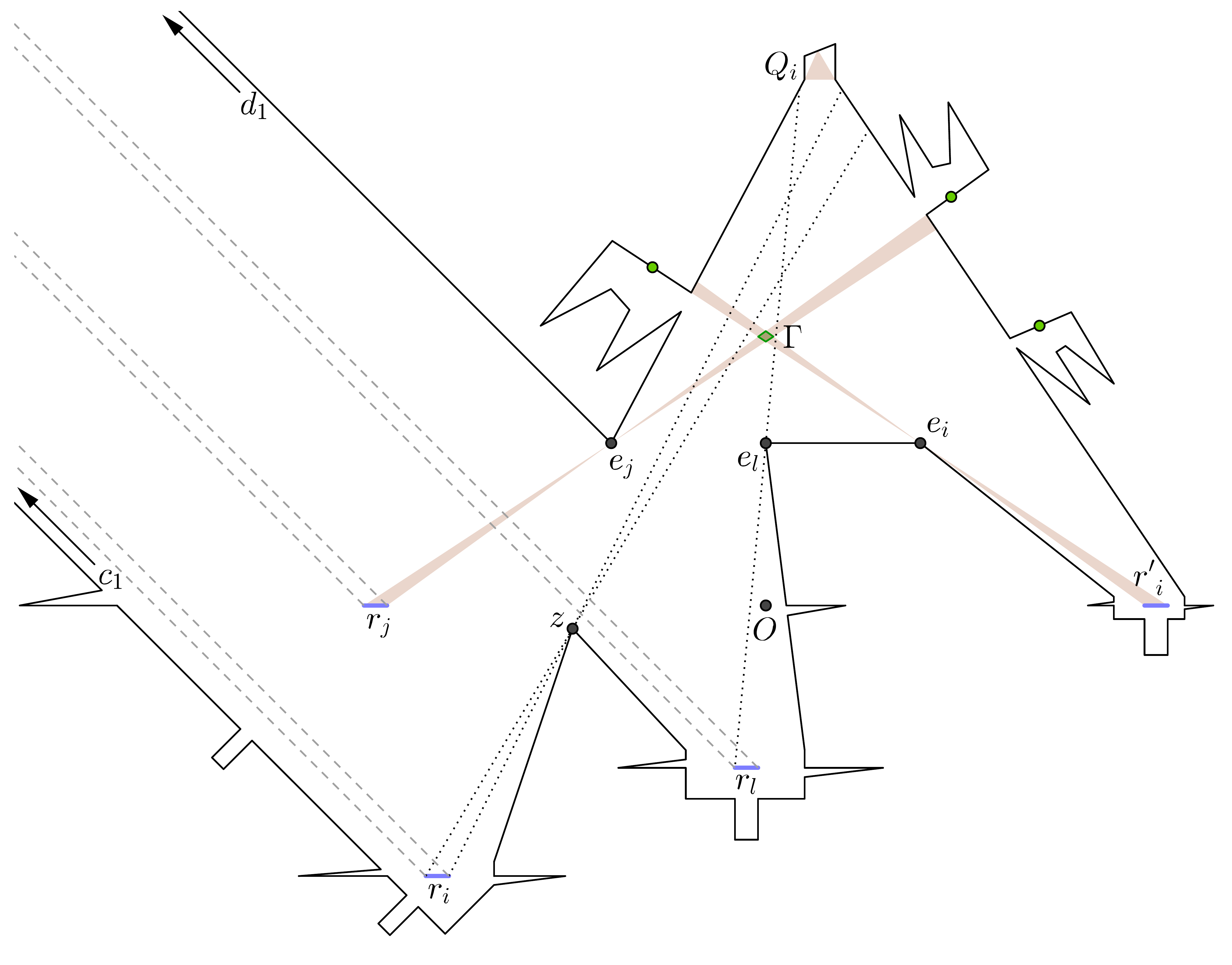}
\caption{Detailed construction of the $\geq$-addition gadget.
Note that the dotted lines show that: no guard on $r_i$ can see any point in $\Gamma$ because of the corner $z$; a guard on $r_i$ can always see both shadow corners of the copy-nook $Q_i$; and no point on $r_l$ sees any point of $Q_i$ because of the corner $e_l$.
For each of the segments $r_\sigma$, $\sigma\in\{i,j,l\}$, the rays from points on $r_\sigma$ through the corridor entrance $c_1d_1$ are between the two grey dashed rays emitting from the endpoints of $r_\sigma$.}
\label{fig:additionGadgetDetail}
\end{figure}

\subsubsection{Introducing a new guard segment $r_i$}\label{subsec:copying-into-rj}
Consider the setting as described in Section~\ref{subsec:ineq-testing-gadget}, and the polygon $\poly_{\textrm{ineq}}$ from Figure \ref{fig:additionGadget2Principle}. 
We explain how to modify $\poly_{\textrm{ineq}}$ into a polygon $\poly'_{\textrm{ineq}}$, which is a 
scaled version of our gadget. The main part of $\poly'_{\textrm{ineq}}$ is shown in Figure~\ref{fig:additionGadgetDetail}, where again the edges to the left with endpoints at $c_1$ and $d_1$ are not fully shown.

The polygon $\poly'_{\textrm{ineq}}$ is obtained from $\poly_{\textrm{ineq}}$ in the following way.
First, we add an additional guard segment $r_i$ of length $3/2$ (i.e., the same as the length of other guard segments in the construction), with its left endpoint at the point $(-20.5,-17)$. This requires introducing a pocket corresponding to the added guard segment. We ensure that $r'_i$ is a \emph{weak copy} of $r_i$ by creating a copy-nook $Q_i$ for the pair of guard segments $r_i, r'_i$, which cannot be seen from any other guard segment or stationary guard position. The shadow corners of $Q_i$ are $(2.5,34)$ and $(4.5,34)$.
We have to ensure that after this modification, the gadget still enforces the desired inequality. In particular, we have to ensure that a guard placed on $r_i$ cannot see any point in the interior of $\Gamma$. We can do that by introducing a new corner $z\mydef (-12.5,-1.5)$ of the polygon that blocks $r_i$ from seeing $\Gamma$. Note that $z$ does not block $r_i$ from seeing the segment $c_1 d_1$ (as shown by the dashed grey lines in Figure~\ref{fig:additionGadgetDetail}). 

\begin{lemma}\label{lem:addition-gadget-8-guards}
A set of guards $G \subset \poly'_{\textrm{ineq}}$ of cardinality at most $7$ guards $\poly'_{\textrm{ineq}}$ if and only if
\begin{itemize}
\item there is exactly one guard placed on each guard segment $r'_i, r_i, r_j, r_l$ and at each stationary guard position,
\item the variables $x_i, x'_i$ corresponding to the guard segments $r_i, r'_i$, respectively, satisfy the inequality $x_i \ge x'_i$, and
\item the variables $x'_i, x_j, x_l$ corresponding to the guard segments $r'_i, r_j, r_l$, respectively, satisfy the inequality $x'_i + x_j \ge x_l$.
\end{itemize}
\end{lemma}

\begin{proof}
Assume that $\poly'_{\textrm{ineq}}$ is guarded by a set $G$ of at most $7$ guards. Similarly as in Lemma~\ref{lem:addition-gadget-7-guards} we can show that there must be exactly one guard at each guard segment and each stationary guard position. As the copy-nook $Q_i$ can be seen only by guards placed on $r_i$ and $r'_i$, it follows from Lemma~\ref{lem:copy-proof} that the polygon is guarded by $G$ only if the variables $x_i, x'_i$ corresponding to $r_i, r'_i$, respectively, satisfy the inequality $x_i \ge x'_i$. 
As the quadrilateral $\Gamma$ cannot be seen by a guard from $r_i$, we get from Lemma~\ref{lem:addition-gadget-7-guards} that the variables $x'_i, x_j, x_l$ corresponding to the guard segments $r'_i, r_j, r_l$, respectively, must satisfy the inequality $x'_i + x_j \ge x_l$.

Now assume that there is exactly one guard placed on each guard segment $r'_i, r_i, r_j, r_l$ and at each stationary guard position,
the variables $x_i, x'_i$ satisfy the inequality $x_i \ge x'_i$, and the variables $x'_i, x_j, x_l$ satisfy the inequality $x'_i + x_j \ge x_l$. Then all of $\Gamma$ is seen by the guards, as is all of the nook $Q_i$. The remaining area is also seen by the guards, which we can show in the same way as in Lemma~\ref{lem:addition-gadget-7-guards}.
\end{proof}

\subsubsection{Attaching the gadget}\label{subsec:corridor-consistency}
We are given a formula from $\Phi$ of the form $x_i+x_j=x_l$, where $i,j,l\in\{1,\ldots,n\}$ and want to construct a gadget imposing an inequality $x_i+x_j\ge x_l$.
We need to show that the values of the required variables can be copied into the guard segments $r_i,r_j,r_l$ in the gadget described above.
First, we will explain how to choose the segments from the base line to be copied into the gadget.
Then we will show that the $\geq$-addition gadget satisfies properties required by Lemma~\ref{lemma:r-slabs}, which will ensure that Lemma~\ref{lemma:copy:works} can be applied, i.e., that the gadget can be connected to the main area by a corridor.

In order to apply a corridor construction as described in Section~\ref{sec:copy} to copy three guard segments from the base line into the gadget, we require the segments in order from left to right on the base line to represent the variables $x_i,x_j,x_l$.
Recall that there are $n$ variables $x_1,\ldots,x_{n}$ in the formula $\Phi$, but that for any $\sigma\in\{1,\ldots,n\}$, we use $x_{\sigma+n}, x_{\sigma+2n}$ and $x_{\sigma+3n}$ as synonyms for $x_\sigma$. Therefore, the inequality $x_{i}+x_{j}\ge x_{l}$ is equivalent to $x_{i}+x_{n+j}\ge x_{2n+l}$.
The guard segments on the base line are $s_1,\ldots,s_{4n}$, where each $s_\sigma$ represents the variable $x_\sigma$.
The segments $s_1,\ldots,s_{3n}$ are right-oriented whereas $s_{3n+1},\ldots,s_{4n}$ are left-oriented.
(In Section~\ref{sec:leftRight} we explain how to obtain these dependencies between the guard segments.)
Therefore, with slight abuse of notation, we redefine $j\mydef j+n$ and $l\mydef l+2n$ so that $i<j<l<3n$, and the guard segments $s_i, s_j, s_l$ satisfy our requirements.

We now show that the gadget satisfies the conditions of Lemma~\ref{lemma:r-slabs}.
Recall that our gadget is the polygon $\poly'_{\textrm{ineq}}$ scaled by a factor of $\frac{1}{CN^2}$.

%

\begin{proof}[Proof of Lemma~\ref{lemma:r-slabs} for the $\geq$-addition gadget.]
Note that in the $\geq$-addition gadget, the segments $r_i, r_j, r_l$ 
have lengths of $\frac{3/2}{CN^2}$ and their right endpoints are placed at positions $m+(\frac{-20.5}{CN^2},\frac{-17.5}{CN^2}), m+(\frac{-24.5}{CN^2},0), m+(\frac{-0.5}{CN^2},\frac{-10.5}{CN^2})$, respectively, where $m\mydef c_1+(1,-1)$. 
As the conditions of Lemma~\ref{lem:ugly} are satisfied, Properties~\ref{rslabs:1} and \ref{rslabs:2} hold.

Property~\ref{rslabs:3} also holds, as the edge $e_j d_1$ blocks all points at stationary guard positions and at the guard segment $r'_i$ from seeing $c_1 d_1$.
\end{proof}

\subsubsection{Summary}\label{subsec:summary}

\begin{lemma}\label{lemma:additionSummery}
Consider the addition gadget together with the corresponding corridor representing an inequality $x_i+x_j\ge x_l$, as described above. The following properties hold.
\begin{itemize}
\item The gadget and the corridor fit into a rectangular box of height $3$.
\item For any guard set of $\poly$, at least $10$ guards have to be placed in the corridor and the gadget.
\item Assume that in the main area $\poly_M$, there is exactly one guard at each guard segment, and there are no guards outside of the guard segments. Then $10$ guards can be placed in the corridor and the gadget so that the whole corridor and gadget is seen if and only if the values $x_i,x_j,x_l$ specified by the guards at the guard segments $s_{i}, s_{j}, s_{l}$ satisfy the inequality $x_{i}+x_{j}\ge x_{l}$.
\end{itemize}
\end{lemma}

\begin{proof}
Recall that by Lemma~\ref{lemma:copy:works}, the distance from $c_0c_1$ to the topmost point in the corridor is at most $1.4$.
The main part of the gadget is centered around the point $c_0+(1,-1)$, and as it is of size $\Theta(\frac{1}{CN^2})$, the vertical space of at most $1.1$ below the line segment $c_0 c_1$ is enough to fit the gadget.

From Lemma~\ref{lemma:copy:works}, there are at least $3$ guards placed within the corridor. From Lemma~\ref{lem:addition-gadget-8-guards}, there are at least $7$ guards placed within the gadget. That gives us at least $10$ guards needed.

Assume that there are exactly $10$ guards within the corridor and gadget and that the corridor is completely seen.
Then, from Lemma~\ref{lemma:copy:works} and \ref{lem:addition-gadget-8-guards}, there is exactly one guard at each guard segment and each stationary guard position. Then, by Lemma~\ref{lemma:copy:works}, the values  $x_i,x_j,x_l$ specified by $s_{i}, s_{j}, s_{l}$, and the values specified by $r_i, r_j, r_l$, are the same.
By Lemma~\ref{lem:addition-gadget-8-guards}, the values $x_i, x'_i$ corresponding to $r'_i,r_i$ satisfy $x_i \ge x'_i$, and we also have $x'_i + x_j \ge x_l$.
That enforces inequality $x_{i}+x_j\ge x_{l}$.

On the other hand, assume that $x_{i}+x_j\ge x_{l}$.
We first place a guard at every of the $6$ stationary guard positions in the corridor and gadget.
By Lemmas \ref{lemma:copy:works} and \ref{lem:addition-gadget-8-guards}, if we set guards at the $4$ guard segments so that the values specified by guards at $r_i, r_j, r_l$ are $x_{i}, x_{j}, x_{l}$, and the value $x'_i$ specified by the guard at $r'_i$ is the same as $x_i$, then all of the gadget is guarded.
\end{proof}


\subsection{The $\le$-addition gadget}\label{sec:additionGadget2}

In this section we present construction of a gadget representing an inequality $x_i+x_j\le x_l$, where $i,j,l\in\{1,\ldots,n\}$.
The idea of the construction of this gadget is analogous to the construction of the $\ge$-addition gadget presented in Section~\ref{sec:additionGadget}, and the basic principle is presented in Section~\ref{sec:AdditionGadget2Idea}.
The principle underlying the $\ge$-addition gadget, as explained in Section~\ref{subsec:addition-idea}, required the polygon to have edges blocking the visibility between the segment $r'_i$ (placed at the right side) and the segments $r_j,r_l$ (placed at the left side).
We managed to get around that by making $r'_i$ a weak copy of a segment $r_i$, which in turn was a copy of a segment $s_i$ on the base line.
In contrast to this, the principle underlying the $\le$-addition gadget presented here requires the polygon to have an edge separating $r'_i$ placed at the left side from the segments $r_j,r_l$ at the right side.
If we were to build a gadget to be placed at the right side of $\poly$, we would have to copy the variables corresponding to both of the segments $r_j,r_l$ weakly, and the gadget would not enforce the desired inequality. To avoid this problem, we will place the gadget at the left side of $\poly$.
Then we introduce an additional guard segment $r_i$, and we make $r'_i$ a weak copy of $r_i$ using a copy-\nook\ $Q_i$.
As $r'_i$ is to the left of $r_i$, the copy-\nook\ $Q_i$ enforces the inequality $x'_i\geq x_i$, where $r'_i,r_i$ represent $x'_i,x_i$, respectively.
The result is that the gadget enforces the desired inequality.
The relative placement of the segments $r_i, r_j, r_l, r'_i$ (in particular, the value of $w$ from Section~\ref{subsec:ineq-testing-gadget}) has to be slightly different than in the construction of $\ge$-addition gadget, as it does not seem to be possible to make the gadgets completely symmetrical.

\begin{figure}[h]
\centering
\includegraphics[clip, trim = 0.9cm 0.65cm 0.5cm 0.7cm,width=0.9\textwidth]{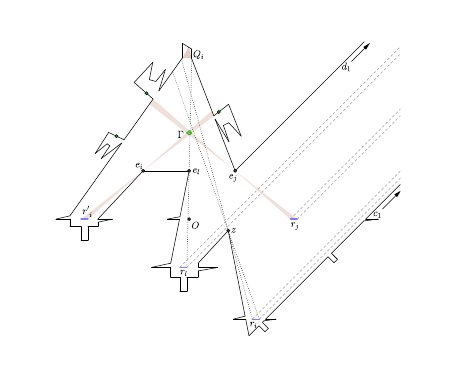}
\caption{Detailed construction of the $\leq$-addition gadget.
As previously, the dotted line shows that no guard on $r_i$ can see any point in $\Gamma$ because of the corner $z$, a guard on $r_i$ can always see both shadow corners of the copy-nook $Q_i$, and no point on $r_l$ sees any point of $Q_i$ because of the corner $e_l$.
For each of the segments $r_\sigma$, $\sigma\in\{i,j,l\}$, the rays from points on $r_\sigma$ through the corridor entrance $c_1d_1$ are between the two grey dashed rays emitting from the endpoints of $r_\sigma$.
}
\label{fig:reflectedGadgetDetail}
\end{figure}

\subsubsection{Idea behind the gadget construction}\label{sec:AdditionGadget2Idea}
The idea behind a gadget imposing an inequality $x'_i+x_j \le x_l$ is similar as for the $\geq$-addition gadget described above.
As before, consider rational values $w,v,h>0$, where $w>v+3/2$, and
let $r'_i,r_j,r_l$ be right-oriented guard segments of length $3/2$ such that $r'_i$ has its left endpoint at the point $(-w,0)$, $r_j$ has its right endpoint at $(w,0)$, and $r_l$ has its left endpoint at $(-2,-h)$. 
Let $g'_i\mydef (-w-1/2+x_i,0)$, $g_j\mydef (w-2+x_j,0)$, and $g_l\mydef (-5/2+x_l,-h)$ be three guards on $r'_i,r_j,r_l$, respectively, representing the values $x'_i,x_j,x_l \in [1/2,2]$.

Suppose that there are corners $e_i\mydef (-v,h),e_j\mydef (v,h),e_l\mydef (0,h)$ of $\poly$.
As before, let $\Gamma$ be a collection of points $\omega$ such that the ray $\overrightarrow{\omega e_i}$ intersects $r'_i$, and the ray $\overrightarrow{\omega e_j}$ intersects $r_j$.
Then $\Gamma$ is a quadrilateral, bounded by the following rays: the rays with origin at the endpoints of $r'_i$ and containing $e_i$, and the rays with origin at the endpoints of $r_j$ and containing $e_j$.
Suppose that
\begin{itemize}
\item
for every point $g'_i$ on $r'_i$ and $\omega$ in $\Gamma$, the points $\omega$ and $g'_i$ can see each other if and only if $\omega$ is on or to the left of the line $\overleftrightarrow{g'_i e_i}$,

\item
for every point $g_j$ on $r_j$ and $\omega$ in $\Gamma$, the points $\omega$ and $g_j$ can see each other if and only if $\omega$ is on or to the left of the line $\overleftrightarrow{g_j e_j}$,

\item
for every point $g_l$ on $r_l$ and $\omega$ in $\Gamma$, the points $\omega$ and $g_l$ can see each other if and only if $\omega$ is on or to the right of the line $\overleftrightarrow{g_l e_l}$.
\end{itemize}
Under these assumptions one can show the following result in an analogous way as we proved Lemma~\ref{lem:inequality}.

\begin{lemma}\label{lem:rev-inequality}
The guards $g'_i, g_j, g_l$ can together see the whole quadrilateral $\Gamma$ if and only if $x'_i + x_j \le x_l$.
\end{lemma}

\subsubsection{Specification of the gadget}
We will present the construction of a gadget with four guard segments $r_i, r_j, r_l,r'_i$, where the segments $r'_i,r_j,r_l$ correspond to the segments in the idea described in Section~\ref{sec:AdditionGadget2Idea}, where this time we set $w\mydef 23.5$, $v\mydef 10$, and $h\mydef 10.5$.
The gadget is shown in Figure~\ref{fig:reflectedGadgetDetail} and should be attached to the left side of the main are $\poly_M$ using a left corridor as described in Section~\ref{subsec:corridor-symmetric}.

As for the case of the $\geq$-addition gadget, there are three stationary guards which do not see any point within $\Gamma$, but which enforce that the whole area except for $\Gamma$ is seen whenever a guard is placed on each of the guard segments $r_i,r_j,r_l,r'_i$.
The guard segment $r_i$ has length $3/2$ (as do $r'_i,r_j,r_l$) and is placed with its left endpoint at the point $(13.75, -21.75)$.
We will ensure that $r'_i$ is a \emph{weak copy} of $r_i$ by creating a copy-\nook\ $Q_i$ the for pair of guard segments $r_i, r'_i$.
The \nook\ $Q_i$ has shadow corners $(-1.5,35)$ and $(0.5,35)$.
To ensure that a guard placed on $r_i$ cannot see any point in the interior of $\Gamma$, we introduce a new corner $z\mydef (8.5, -2.5)$ of the polygon that blocks $r_i$ from seeing $\Gamma$.
Two edges to the right in the figure are not fully shown.
They end at corners $c_1\mydef (CN^2,CN^2)$ and $d_1\mydef (CN^2,CN^2+1.5)$.

In order to attach the gadget to the main are $\poly_M$, we scale down the construction described here by the factor $\frac 1{CN^2}$ and translate it so that the point which corresponds to $O\mydef (0,0)$ in the gadget will be placed at position $m\mydef c_1+(-1,-1)$. (Recall that the left entrance to the corridor, at which we attach the gadget, is the segment $c_1d_1$.)

The following lemma is proved in the same way as Lemma~\ref{lem:addition-gadget-8-guards}.

\begin{lemma}\label{lem:rev-addition-gadget-8-guards}
Let $\poly'_{\textrm{rev-ineq}}$ be the polygon obtained from the gadget described above by closing it by adding the edge $c_1d_1$.
A set of guards $G \subset \poly'_{\textrm{rev-ineq}}$ of cardinality at most $7$ guards $\poly'_{\textrm{rev-ineq}}$ if and only if
\begin{itemize}
\item there is exactly one guard placed on each guard segment $r'_i, r_i, r_j, r_l$ and at each stationary guard position,
\item the variables $x_i, x'_i$ corresponding to the guard segments $r_i, r'_i$, respectively, satisfy the inequality $x_i \le x'_i$, and
\item the variables $x'_i, x_j, x_l$ corresponding to the guard segments $r'_i, r_j, r_l$, respectively, satisfy the inequality $x'_i + x_j \le x_l$.
\end{itemize}
\end{lemma}

\subsubsection{Attaching the gadget}
Here we proceed as in Section~\ref{subsec:corridor-consistency}. We need to show is that the variables $x_i,x_j,x_l$ can be copied into guard segments $r_i,r_j,r_l$ from three guard segments on the base line.
Due to the gadget construction, we now require the segment corresponding to the variable $x_l$ to be the leftmost one, and all guard segments have to be right-oriented. 
With slight abuse of notation, we redefine $i\mydef i+n$ and $j\mydef j+2n$.
Then the segments $s_l,s_i,s_j$ satisfy our requirements.

As before, to prove that the corridor construction enforces the required dependency between the guards on the base line and guards within the gadget, i.e., for Lemma~\ref{lemma:copy:works} to work, we need to show that our gadget construction satisfies the conditions of Lemma~\ref{lemma:rev-r-slabs}.


\begin{proof}[Proof of Lemma~\ref{lemma:rev-r-slabs} for the $\leq$-addition gadget.]
Note that in the $\leq$-addition gadget, the segments $r_i, r_j, r_l$ have lengths of $\frac{3/2}{CN^2}$ and their left endpoints are placed at positions $m+(\frac{13.75}{CN^2},\frac{-21.75}{CN^2}), m+(\frac{22}{CN^2},0), m+(\frac{-2}{CN^2},\frac{-10.5}{CN^2})$, respectively, where $m\mydef c_1+(-1,-1)$.
As the conditions of Lemma~\ref{lem:rev-ugly} are satisfied, Properties~\ref{rev-rslabs:1} and \ref{rev-rslabs:2} hold.

Property~\ref{rev-rslabs:3} also holds, as the edge $e_j d_1$ blocks all points at stationary guard positions and at the guard segment $r'_i$ from seeing $c_1 d_1$.
\end{proof}

\subsubsection{Summary}
In the same way as in Lemma~\ref{lemma:additionSummery}, we get the following result.

\begin{lemma}\label{lemma:rev-additionSummery}
Consider the $\leq$-addition gadget together with the corridor, corresponding to the inequality $x_i+x_j\le x_l$. The following properties hold.
\begin{itemize}
\item The gadget and the corridor fit into a rectangular box of height $3$.
\item For any guard set of $\poly$, at least $10$ guards have to be placed in the corridor and the gadget.
\item Assume that in the main area $\poly_M$, there is exactly one guard at each guard segment, and there are no guards outside of the guard segments. Then $10$ guards can be placed in the corridor and the gadget so that the whole corridor and gadget is seen if and only if the values $x_i,x_j,x_l$ specified by the guards at the guard segments $s_{i}, s_{j}, s_{l}$ satisfy the inequality $x_{i}+x_{j}\le x_{l}$.
\end{itemize}
\end{lemma}


\subsection{The $\geq$- and $\leq$-orientation gadgets}\label{sec:leftRight}

In this section we explain how to enforce consistency between the guard segments on the base line which represent the same variable $x_i$, for $i \in \{1, \ldots, n\}$. Recall that there are four guard segments $s_i, s_{n+i}, s_{2n+i}, s_{3n+i}$ representing the variable $x_i$, and that the first three ones are right-oriented, and the last one is left-oriented.

We will present a gadget enforcing that two guard segments corresponding to the same variable $x_i$ and oriented in different directions specify the variable consistently. We will then use this gadget for the following pairs of guard segments: $(s_i,s_{3n+i})$, $(s_{n+i},s_{3n+i})$, and $(s_{2n+i},s_{3n+i})$.

Consider two guard segments $s_i,s_j$ on the base line, where $s_i$ is right-oriented and $s_j$ is left-oriented, and assume that there is one guard placed on each of these segments. Let $x_i$ and $x_j$ be the values represented by $s_i$ and $s_j$, respectively. Let $x^r_j$ be the value that would be specified by $s_j$ if $s_j$ was right-oriented instead of left-oriented. We have $x_j + x^r_j = 2.5$. Therefore $s_i$ and $s_j$ specify the same value if and only if $x_i + x^r_j = 2.5$. 

Performing a simple modification of the $\geq$- and $\leq$-addition gadgets, we obtain the $\geq$- and $\leq$-orientation gadgets, which together enforce the equality $x_i + x^r_j = 2.5$.
See Figure~\ref{fig:additionGadget3inversion} for a detailed picture of the main part of the $\geq$-orientation gadget, which enforces that $x_i + x^r_j \geq 2.5$, or, equivalently, $x_i\geq x_j$.
In the $\geq$-addition gadget, we copy three values from the base line into the gadget. Here, we copy only the value of the two segments $s_i,s_j$.
Instead of the guard segment $r_l$ inside the gadget, we create a stationary guard position $p$ at the line containing $r_l$ at distance $\frac 1{CN^2}$ to the right of the right endpoint of $r_l$. 
Then $p$ corresponds to the value of $5/2$ on $r_l$ (ignoring that $p$ lies outside $r_l$).

\begin{figure}[h]
\centering
\includegraphics[clip, trim = 0.5cm 0.7cm 1cm 0.7cm,width=0.9\textwidth]{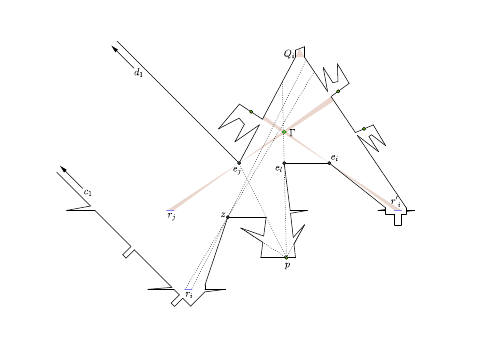}
\caption{Detailed construction of the $\geq$-orientation gadget for $x_i^r+x_j^r\geq 2.5$, which is a modified version of the $\geq$-addition gadget.}
\label{fig:additionGadget3inversion}
\end{figure}

The $\leq$-orientation gadget, which corresponds to the inequality $x_i+x_j^r\leq 5/2$, is obtained by an analogous modification of the $\leq$-addition gadget.
Note that in both of these orientation gadgets, we create $4$ stationary guard positions in the corridor instead of $6$ for the addition gadgets, and the gadget itself contains $4$ stationary guards and $3$ guard segments, instead of $3$ and $4$ in the addition gadgets, respectively.

We summarize the properties of the orientation gadgets by the following Lemma, which can be proven in a way analogous to Lemmas~\ref{lemma:additionSummery} and \ref{lemma:rev-additionSummery}.

\begin{lemma}\label{lemma:orientationSummery}
Consider the $\geq$-orientation gadget (resp.~$\leq$-orientation gadget) together with the corresponding corridor for making $r_i,r_j$ copies of guard segments $s_i,s_j$ on the base line, where $s_i$ is right-oriented and $s_j$ is left-oriented. The following properties hold.
\begin{itemize}
\item The gadget and the corridor fit into a rectangular box of height $3$.
\item For any guard set of $\poly$, at least $9$ guards have to be placed in the corridor and the gadget.
\item Assume that in the main area $\poly_M$, there is exactly one guard at each guard segment, and there are no guards outside of the guard segments. Then $9$ guards can be placed in the corridor and the gadget so that the whole corridor and gadget is seen if and only if the values $x_i,x_j$ specified by the guards at the guard segments $s_{i}, s_{j}$ satisfy the inequality $x_{i}\geq x_{j}$ (resp.~$x_{i}\leq x_{j}$).
\end{itemize}
\end{lemma}


\subsection{The inversion gadget}\label{sec:inversionGadget}

In this section we present the construction of the inversion gadget which represent an inequality $x_i\cdot x_j = 1$, where $i,j\in\{1,\ldots,n\}$.
We made use of Maple~\cite{maple} for the construction and verification of this gadget.

\subsubsection{Idea behind the gadget construction}

We first describe the principle underlying the inversion gadget.
Let $r_i$ and $r_j$ be two guard segments representing variables $x_i$ and $x_j$, respectively.
We want to construct an umbra $Q_u$ such that if the guards at $r_i$ and $r_j$ together see the critical segment $f_0f_1$ of $Q_u$, then one of the inequalities $x_ix_j\leq 1$ or $x_ix_j\geq 1$ follows.
Likewise, we want to construct a nook $Q_n$ such that if the guards at $r_i$ and $r_j$ together see the critical segment $f_2f_3$ of $Q_n$, then the other of the two inequalities follows, so that in effect, $x_ix_j=1$.
This does not seem possible to obtain if both guard segments $r_i,r_j$ are right-oriented, but as we will see, it is possible when $r_i$ is right-oriented and $r_j$ is left-oriented.
Furthermore, in order to get rational coordinates of the shadow corners of $Q_u$ and $Q_n$, it seems necessary to have $r_i$ and $r_j$ at different $y$-coordinates.

Some rather complicated fractions are used as the coordinates of the shadow corners of $Q_u$ and $Q_n$.
In the following, we explain why they need to be somewhat complicated, and how they were found.
Let $a'_i\mydef (1/2,0)$, $b'_i\mydef (2,0)$, $a'_j\mydef (13.9,0.1)$, $b'_j\mydef (15.4,0.1)$, and suppose that $r_i\mydef a'_ib'_i$ is right-oriented and $r_j\mydef a'_jb'_j$ is left-oriented, see Figure~\ref{fig:inversionBothIneqFull}.
Suppose that an umbra $Q_u$ of $r_i$ and $r_j$ has shadow corners $\xi_0\mydef (7,h_l)$ and $\xi_1\mydef (9,h_r)$.
Then the critical segment of $Q_u$ has endpoints
\begin{align*}
f_0&\mydef
\left({\frac {770h_r-45+128h_l}{64h_l+50h_r-5}},
{\frac {134h_lh_r-7h_l}{64h_l+50h_r-5}}\right)\quad
\text{and}\\
f_1&\mydef
\left({\frac {1807h_r-117+49h_l}{98h_l+130h_r-13}},
{\frac {268h_lh_r-17h_l}{98h_l+130h_r-13}}\right).
\end{align*}

Let $\pi_i\colon r_i\longrightarrow f_0f_1$ and $\pi_j\colon r_j\longrightarrow f_0f_1$ be the projections associated with $Q_u$.
We want the function $\pi_j^{-1}\circ\pi_i\colon r_i\longrightarrow r_j$ to map a point on $r_i$ specifying the value $x_i$ to a point on $r_j$ specifying the value $x_j=1/x_i$.
Note that the point $(1,0)$ on $r_i$ specifies the value $1$ of $x_i$.
Therefore, the point
$
\pi_j^{-1}(\pi_i((1,0)))=
(p,0.1)
$,
where
$$
p\mydef \frac {34009h_l^{2}-71042h_lh_r+1147h_l+34153h_r^{2}-859h_r}{2385h_l^{2}-4980h_lh_r+80h_l+2395h_r^{2}-60h_r},
$$
has to specify the value $1$ on $r_j$.
The value specified by $(p,0.1)$ on $r_j$ is $15.9-p$.
Solving the equation $15.9-p=1$ for $h_r$ gives that
$$
h_r=\frac {632h_l+7\pm\sqrt {24881h_l^{2}-2186h_l+49}}{613}.
$$
In order for $h_l$ and $h_r$ both to be rational, it is thus required that $h_l$ is a rational number for which $24881h_l^{2}-2186h_l+49$ is the square of a rational number.
This is in general not -- but luckily sometimes -- the case. 
Indeed, if we define $h_l\mydef \frac{541}{184}$, then $24881h_l^{2}-2186h_l+49=\left(\frac{84061}{184}\right)^2$.
We used Maple to find such a number $h_l$ in a suitable range that made it possible to create the gadget.

These considerations lead us to the following construction.
Let $\xi_0\mydef (7, \frac{541}{184})\approx (7, 2.94)$ and $\xi_1\mydef (9,\frac{259139}{112792})\approx (9, 2.30)$ and suppose that $\xi_0,\xi_1$ are shadow corners of \anumbra\ $Q_u$ with corners $\xi_0\xi_1f_1f_0$ of $r_i,r_j$.
Then $f_0\mydef (\frac{499811}{70923}, \frac{38731813}{13049832})\approx(7.05, 2.97)$ and $f_1\mydef (\frac{112379}{15432}, \frac{4355591}{1419744})\approx(7.28, 3.07)$.

\begin{lemma}
If guards $p_i,p_j$ on $r_i,r_j$, respectively, together see $f_0f_1$, then $x_ix_j\leq 1$.
\end{lemma}

\begin{proof}
Let $\pi_i\colon r_i\longrightarrow f_0f_1$ and $\pi_j\colon r_j\longrightarrow f_0f_1$ be the projections associated with $Q_u$.
Note that since $p_i$ represents the variable $x_i$, we must have $p_i\mydef (x_i,0)$.
Let
$$e\mydef \pi_i(p_i)=\left({\frac {258288\,x_i-16765}{36994\,x_i-3065}},{\frac {20013754\,
x_i-1295695}{6806896\,x_i-563960}}
\right).$$
Now, $\pi^{-1}_j(e)=(15.9-1/x_i,1/10)$, which represents the value $15.9-(15.9-1/x_i)=1/x_i$ on $r_j$.
In order to see $f_0f_1$ together with $p_i$, the guard $p_j$ has to stand on $\pi^{-1}_j(e)$ or to the right.
This corresponds to $x_j$ being at most $1/x_i$.
In other words, if a guard $p_j$ on $r_j$ sees $f_0f_1$ together with $p_i$, then $x_ix_j\leq 1$.
\end{proof}

We now construct a nook that impose the guards to satisfy the opposite inequality $x_ix_j\geq 1$:
Let $\xi_2\mydef (7, \frac{8865}{752})\approx(7, 11.79)$ and $\xi_3\mydef (9,\frac{4214815}{460976})\approx(9,9.14)$ and suppose that $\xi_2,\xi_3$ are shadow corners of \anook\ $Q_n$ with corners $\xi_2\xi_3f_3f_2$ of $r_i,r_j$.
Then $f_2\mydef (\frac{182083}{25835}, \frac{231222249}{19427920})\approx(7.05, 11.90)$ and $f_3\mydef (\frac{205139}{28156}, \frac{130288905}{10586656})\approx(7.29, 12.31)$.

\begin{lemma}
If guards $p_i,p_j$ on $r_i,r_j$, respectively, together see $f_2f_3$, then $x_ix_j\geq 1$.
\end{lemma}

\begin{proof}
Let $\hat\pi_0$ and $\hat\pi_1$ be the associated projections from $r_i$ and $r_j$ to $f_2f_3$, respectively.
Let 
$$\hat e\mydef \hat\pi_0(p_0)=
\left({\frac {470184\,x_i-29953}{67346\,x_i-5517}},{\frac {597022290
\,x_i-37933335}{50644192\,x_i-4148784}}\right).$$
Now, $\hat\pi^{-1}_j(\hat e)=(15.9-1/x_i,1/10)$, which represents the value $15.9-(15.9-1/x_i)=1/x_i$ on $r_j$.
In order to see $f_2f_3$ together with $p_i$, the guard $p_j$ has to stand on $\hat\pi^{-1}_1(\hat e)$ or to the left.
This corresponds to $x_j$ being at least $1/x_i$ so that $x_ix_j\geq 1$.
\end{proof}

We thus have the following lemma:
\begin{lemma}\label{lem:inv}
If guards $p_i$ and $p_j$ placed on guard segments $r_i$ and $r_i$, respectively, see both critical segments $f_0f_1$ and $f_2f_3$, then the corresponding values specified by $p_i$ and $p_j$ satisfy $x_ix_j=1$.
\end{lemma}

\begin{figure}
\centering
\includegraphics[clip, trim = 1cm 1.0cm 1cm 0.7cm,width=0.75\textwidth]{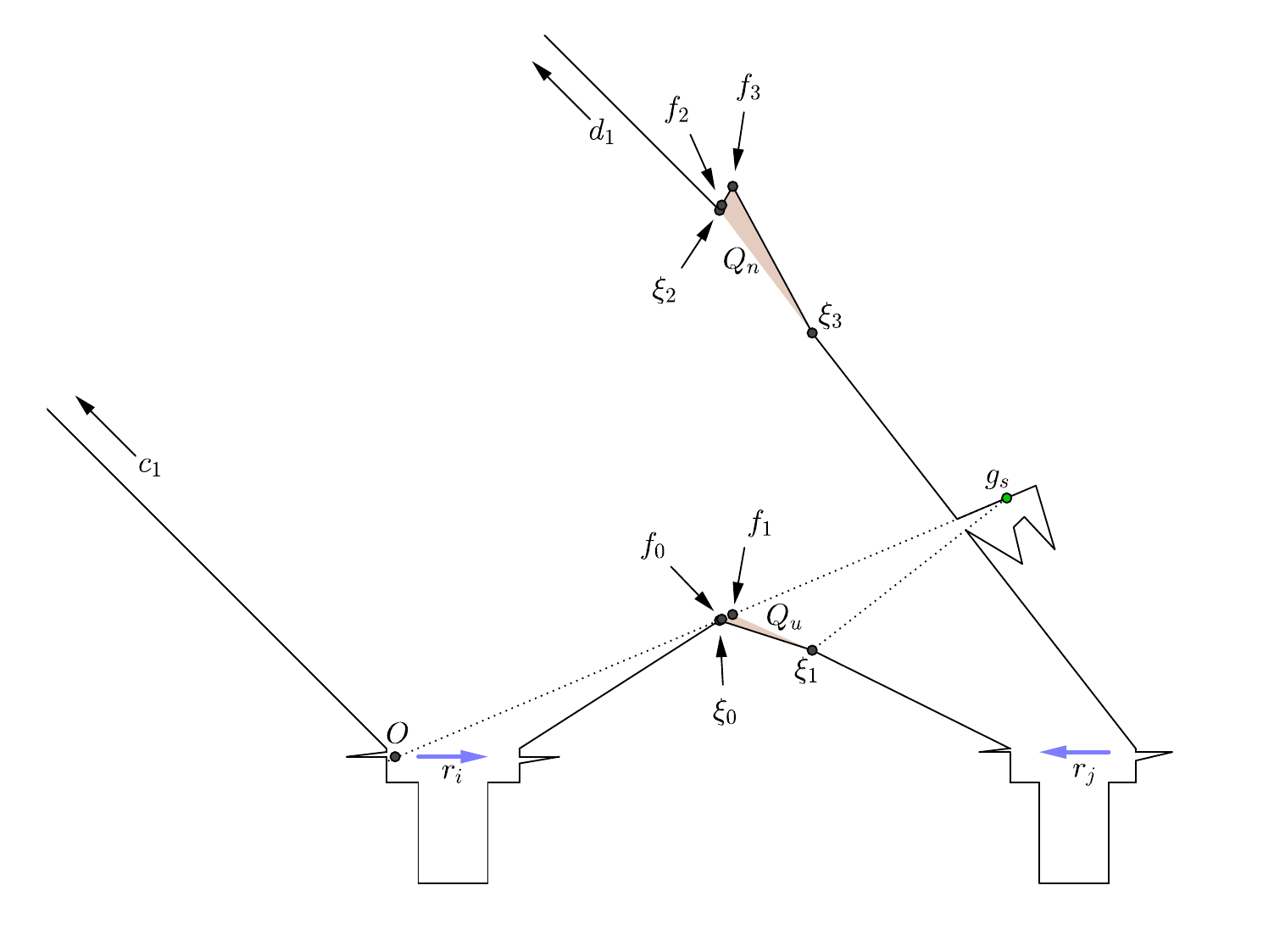}
\caption{The inversion gadget.
The \nook\ and \umbra\ (the brown areas) for the pair of guard segments $r_i, r_j$ impose the inequality $x_ix_j=1$ on the variables $x_i$ and $x_j$ represented by $r_i$ and $r_j$.
The stationary guard position $g_s$ sees the \umbra\ but nothing above the line $\protect\overleftrightarrow{f_0f_1}$.}
\label{fig:inversionBothIneqFull}
\end{figure}

\subsubsection{Specification of the gadget}

We now explain how to make the complete gadget for the equation $x_ix_j=1$, as shown in Figure~\ref{fig:inversionBothIneqFull}.
We make the wall in the gadget so that it creates the \umbra\ $Q_u$ and the \nook\ $Q_n$ as described before.
We also create a stationary guard position at the green point in the figure which sees the umbra $Q_u$, but nothing above the line containing the critical segment of $Q_u$.
Two edges at the left side of the gadget are not fully shown.
They end at corners $c_1\mydef (-CN^2,CN^2)$ and $d_1\mydef (-CN^2,CN^2+1.5)$, respectively.

The gadget contains two guard segments $r_i$ and $r_j$ representing $x_i$ and $x_j$, respectively, and it is required that $r_i$ is right-oriented and $r_j$ is left-oriented.
Therefore, with slight abuse of notation, we redefine $j\mydef j+3n$, so that $s_i,s_j$ are guard segments on the base line representing $x_i,x_j$, respectively, where $s_i$ is right-oriented and $s_j$ is left-oriented.
We then use a corridor as described in Section~\ref{sec:copy} to make $r_i,r_j$ copies of $s_i,s_j$, respectively.
Recall that the endpoints of the right entrance of the corridor are denoted $c_1,d_1$.
In order to attach the gadget to the corridor, we first scale it down by the factor $\frac 1{CN^2}$ and then translate it so that the points $c_1,d_1$ of the gadget coincides with the endpoints of the right entrance of the corridor with the same names.
We thus obtain that the point $O\mydef (0,0)$ in the gadget becomes the point $m\mydef c_1+(1,-1)$ in $\poly$.

\begin{lemma}\label{lem:inversion-gadget-3-guards}
Let $\poly'_{\textrm{inv}}$ be the polygon obtained from the inversion gadget by closing it by adding the edge $c_1d_1$.
A set of guards $G \subset \poly'_{\textrm{inv}}$ of cardinality at most $3$ guards $\poly'_{\textrm{inv}}$ if and only if
\begin{itemize}
\item there is exactly one guard placed on each guard segment $r_i, r_j$ and at the stationary guard position, and
\item the variables $x_i, x_j$ corresponding to the guard segments $r_i, r_j$, respectively, satisfy the equation $x_i \cdot x_j=1$.
\end{itemize}
\end{lemma}

\begin{proof}
Assume that $\poly'_{\textrm{inv}}$ is guarded by a set $G$ of at most $3$ guards. Similarly as in Lemma~\ref{lem:addition-gadget-7-guards} we can show that there must be exactly one guard at each guard segment and at the stationary guard position.
It then follows from Lemma~\ref{lem:inv} that $x_i\cdot x_j=1$.

Now assume that there is exactly one guard placed on each guard segment $r_i, r_j$ and at the stationary guard position, and that the variables $x_i,x_j$ represented by the guards at $r_i,r_j$ satisfy $x_i\cdot x_j=1$.
Then all of $Q_u$ and $Q_n$ is seen by the guards. The remaining area is clearly also seen by the guards.
\end{proof}

\subsubsection{Attaching the gadget}

We now need to show that our gadget construction satisfies the conditions of Lemma~\ref{lemma:r-slabs}.

\begin{proof}[Proof of Lemma~\ref{lemma:r-slabs} for the inversion gadget.]
Note that in the inversion gadget, the segments $r_i, r_j$ 
have lengths of $\frac{3/2}{CN^2}$ and their right endpoints are placed at positions $m+(\frac{2}{CN^2},0)$ and $m+(\frac{15.4}{CN^2},\frac{0.1}{CN^2})$, respectively, where $m\mydef c_1+(1,-1)$. 
As the conditions of Lemma~\ref{lem:ugly} (here used in a simplified version for a gadget with only two guard segments $r_i$ and $r_j$) are satisfied, Properties~\ref{rslabs:1} and \ref{rslabs:2} hold.

Property~\ref{rslabs:3} also holds, as the stationary guard position in the gadget cannot see the edge $c_1d_1$.
\end{proof}

\subsubsection{Summary}

\begin{lemma}\label{lemma:inversionSummery}
Consider the inversion gadget together with the corresponding corridor representing an inequality $x_i\cdot x_j=1$, as described below. The following properties hold.
\begin{itemize}
\item The gadget and the corridor fit into a rectangular box of height $3$.
\item For any guard set of $\poly$, at least $5$ guards have to be placed in the corridor and the gadget in total.
\item Assume that in the main area $\poly_M$, there is exactly one guard at each guard segment, and there are no guards outside of the guard segments. Then all of the corridor and the gadget can be guarded by $5$ guards (together with the guards on the base line) if and only if the values $x_i,x_j$ specified by the guards at the guard segments $s_{i}, s_{j}$ satisfy the inequality $x_{i}\cdot x_{j}=1$.
\end{itemize}
\end{lemma}

\begin{proof}
The proof for the first property is similar to that for the addition gadget in Lemma~\ref{lemma:additionSummery}.

From Lemma~\ref{lemma:copy:works}, there are at least $2$ guards placed within the corridor.
Furthermore, there must be at least $3$ guards placed within the gadget by Lemma~\ref{lem:inversion-gadget-3-guards}.
That gives us at least $5$ guards needed.

Assume that there are exactly $5$ guards within the corridor and gadget and that the corridor is completely seen.
Then, from Lemma~\ref{lemma:copy:works} and \ref{lem:inversion-gadget-3-guards}, there is exactly one guard at each guard segment and each stationary guard position.
By Lemma~\ref{lemma:copy:works}, the values $x_i,x_j$ specified by $s_{i}, s_{j}$, and the values specified by $r_i, r_j$, are the same.
Lemma~\ref{lemma:copy:works} gives that the guards at $r_i,r_j$ must see the critical segments of both $Q_n$ and $Q_u$.
Finally, by Lemma~\ref{lem:inversion-gadget-3-guards}, the values $x_i, x_j$ thus satisfy $x_i\cdot x_j=1$.

On the other hand, assume that $x_{i}\cdot x_j=1$.
We first place a guard at every of the $3$ stationary guard positions in the corridor and gadget.
By Lemmas \ref{lemma:copy:works} and \ref{lem:inversion-gadget-3-guards}, if we set guards at the $2$ guard segments so that the values specified by guards at $r_i, r_j$ are $x_{i}, x_{j}$, then all of the gadget is guarded.
\end{proof}

\subsection{Putting it all together}

Let $\Phi$ be an $\etrinv$ formula with $n$ variables, $k_1$ equations of the form $x_i+x_j=x_l$, and $k_2$ equations of the form $x_i\cdot x_j=1$. We have already explained how to construct the polygon $\poly(\Phi)$, but we shall here give a brief summary of the process.
We start by constructing the main area with $4n$ guard segments. We modify the pockets corresponding to the variables $x_i$ for which $\Phi$ contains equation $x_i=1$, as described in Section~\ref{sec:base line}. To enforce dependency between the base line guard segments corresponding to the same variable, we construct $3n$ $\ge$-orientation gadgets (attached at the right side of the polygon) and $3n$ $\le$-orientation gadgets (attached at the left side), as described in Section~\ref{sec:leftRight}. For each equality of the form $x_i+x_j=x_l$ in $\Phi$, we construct a corresponding $\ge$-addition gadget (attached at the right side), and a $\le$-addition gadget (attached at the left side), as described in Sections~\ref{sec:additionGadget} and \ref{sec:additionGadget2}, respectively. For each equality of the form $x_i\cdot x_j=1$ in $\Phi$ we construct a corresponding inversion gadget (attached at the right side), as described in Section~\ref{sec:inversionGadget}. The total number of gadgets at each side of $\poly$ is therefore at most $3n+k_1+k_2 \leq N$, as stated in Section~\ref{sec:poly}.

Without loss of generality, we assume that the $y$-coordinate of the base line of $\poly(\Phi)$ is $0$.
We set $g(\Phi)\mydef 58n+20k_1+5k_2$. Then $(\poly(\Phi), g(\Phi))$ is an instance of the art gallery problem, and the following theorem holds.

\begin{theorem}\label{thm:final}
Let $\Phi$ be an instance of $\etrinv$.
The polygon $\poly(\Phi)$ has corners at rational coordinates, which can be computed in polynomial time.
Moreover, there exist constants $d_1,\ldots,d_{n} \in \mathbb{Q}$ such that for any $\mathbf x\mydef(x_1,\ldots,x_{n}) \in \RR^n$, $\mathbf x$ is a solution to $\Phi$ if and only if there exists a guard set $G$ of cardinality $g(\Phi)$ containing guards at all the positions $(x_1+d_1,0)\ldots,(x_{n}+d_{n},0)$.
\end{theorem}

\begin{proof}
Consider a guard set $G$ of the polygon $\poly\mydef \poly(\Phi)$.
By Lemma~\ref{lemma:mainAreaGuards}, $G$ has at least $4n$ guards placed in $\poly_M$, and if the number of guards within $\poly_M$ equals $4n$, then there must be exactly one guard at each guard segment.
Lemma~\ref{lemma:orientationSummery} implies that within each of the $6n$ orientation gadgets of $\poly$ together with the corresponding corridors, there are at least $9$ guards, giving at least $54n$ guards in total.
Similarly, from Lemmas \ref{lemma:additionSummery} and \ref{lemma:rev-additionSummery} we obtain that there must be at least $10$ guards placed within each $\ge$-addition gadget and each $\le$-addition gadget plus the corresponding corridors, giving at least $20 k_1$ guards.
By Lemma~\ref{lemma:inversionSummery}, there are at least $5$ guards within each inversion gadget and the corresponding corridor, giving at least $5 k_2$ guards in total. Therefore, $G$ has at least $58n+20k_1+5k_2$ guards, which is equal to $g(\Phi)$.

If a guard set of size $g(\Phi)$ exists, then there are exactly $4n$ guards in $\poly_M$, $9$ guards within each orientation gadget, $10$ guards within each $\ge$-addition gadget and each $\le$-addition gadget, and $5$ guards within each inversion gadget. The same lemmas give us then that there is exactly one guard at each guard segment and each stationary guard position, and no guards away from the guard segments or the stationary guard positions. Also, the variables $x_1,\ldots,x_{n}$ specified by the guard segments $s_1,\ldots,s_{n}$ are a solution to~$\Phi$.

On the other hand, if there exists a solution to $\Phi$, then we get a guard set of size $g(\Phi)$ by placing the guards accordingly.
It is thus clear that the solutions to $\Phi$ correspond to the optimal guard sets of $\poly$, as stated in the theorem.

Due to Lemmas \ref{lem:rational-bottom-wall} and \ref{lem:rational-corridor}, we get that the corners of $\poly_M$ and the corridor corners are all rational, with the nominators and denominators polynomially bounded.
Next, consider all the corners of the gadgets. Each gadget is first described as a polygon with coordinates where nominators and denominators are both of size $\Theta(1)$ (as this construction is fixed and it does not depend on the formula $\Phi$, and we can easily choose the corners so that they are all at rational coordinates), and this polygon is subsequently scaled down by a factor of $\frac 1{CN^2}$ and attached at a corridor entrance, which also has polynomially bounded nominators and denominators.
Thus, the coordinates of the corners in the gadgets have polynomially bounded complexity and can be computed in polynomial time.
\end{proof}


We can now prove the main theorem of the paper.

\ThmMain*

\begin{proof}
By Theorem~\ref{thm:ermembership}, the art gallery problem is in the complexity class $\ER$.
From Theorem~\ref{thm:etrinv} we know that the problem $\etrinv$ is $\ER$-complete.
We presented a polynomial time construction of an instance $(\poly(\Phi),g(\Phi))$ of the art gallery problem from an instance $\Phi$ of $\ETR$.
Theorem~\ref{thm:final} gives that it is $\ER$-hard to solve the art gallery problem when the coordinates of the polygon corners are given by rational numbers. 
Note that the number of corners of $\poly$ is proportional to the input length $|\Phi|$.
By Theorem~\ref{thm:final}, there is a polynomial $|\Phi|^m$ which is a bound on every denominator of a coordinate of a corner in $\poly$.
The product $\Pi$ of denominators of all coordinates of corners of $\poly$ thus has size at most $|\Phi|^{m\; O(|\Phi|)}$.
It follows that we can express $\Pi$ by $O(m|\Phi|\log |\Phi|)$ bits.
By multiplying every coordinate of $\poly$ by $\Pi$, we get a polygon $\poly'$ with integer coordinates and the theorem follows.
\end{proof}

The theorem stated below likewise easily follows from Lemma~\ref{lem:translate-ETRINV} in Section~\ref{sec:etrV} and Theorem~\ref{thm:final}.

\begin{theorem}[Semi-Algebraic Sets]\label{thm:Correspondance}
	Let $\Phi$ be an instance of $\etr$ with $n$ variables.
	Then there is an instance $(\poly,g)$ of the art gallery problem, and constants $c_1,d_1,\ldots,c_n,d_n\in \mathbb{Q}$, such that
	\begin{itemize}
	\item
	if $\Phi$ has a solution, then $\poly$ has a guard set of size $g$, and

	\item
	for any guard set $G$ of $\poly$ of size $g$, there exists a solution $(x_1, \ldots, x_n)\in\RR^n$ to $\Phi$ such that $G$ contains guards at positions $(c_1x_1 +d_1,0), \ldots, (c_n x_n +d_n,0)$.
	\end{itemize}
\end{theorem}

Note that Theorem~\ref{thm:Correspondance} only gives a correspondence between solutions to $\Phi$ and the art gallery problem in one direction, namely from guard sets of $\poly$ to solutions to $\Phi$.
This is inherently unavoidable for two reasons.
First, the solution space for $\Phi$ is in general unbounded, whereas the guards are restricted to $\poly$.
Second, the set of guard sets that guard $\poly$ is closed in the following sense.
Consider a sequence of guard sets $G^1,G^2,\ldots$, each consisting of $g$ guards.
Each guard set $G^i$ is considered as a point in $\RR^{2g}$, where the coordinates $2j$ and $2j+1$ are the coordinates of the $j$'th guard, $j\in\{1,\ldots,g\}$.
Suppose that $G^i$ converges to $G^{*}\in \RR^{2g}$, i.e., $\|G^i - G^{*}\| \longrightarrow 0$ as $i \longrightarrow\infty$.
Then the limit $G^*$ is clearly also also a guard set of $\poly$, so the guard sets with $g$ guards is a closed subset of $\RR^{2g}$.
By restricting ourselves to \emph{compact} semi-algebraic sets $S$, we get a one-to-one correspondence between guard sets of $\poly$ and points in $S$.

\begin{theorem}[Compact Semi-Algebraic Sets]\label{thm:StrongCorrespondance}
	Let $S$ be a compact semi-algebraic set in $\RR^n$.
	Then there is an instance $(\poly,g)$ of the art gallery problem 
	and constants $c_1,d_1,\ldots,c_n,d_n\in \QQ$ such 
	that the following holds.
	For all $\mathbf x \mydef (x_1,\ldots,x_n)\in\RR^n$, the point $\mathbf x$ is in $S$ if and only if $\poly$ has a guard set of size $g$ containing guards at $(c_1x_1 +d_1,0), \ldots, (c_n x_n +d_n,0)$.
\end{theorem}


\begin{proof}
	By Lemma~\ref{lem:compact-translate-ETRINV} there is an instance $\Phi$ of $\etrinv$ with solution set $S'$, such that $S$ can be obtained from $S'$ by removing some coordinates from the points in $S'$ and scale and shift the remaining coordinates.
	The statement now follows from Theorem~\ref{thm:final}.
\end{proof}

We note that both in Theorem~\ref{thm:Correspondance} and~\ref{thm:StrongCorrespondance}, the constants $c_1,d_1,\ldots,c_n,d_n\in \mathbb{Q}$ might be doubly exponentially large in the size of $\Phi$.
Also note that the number $g\mydef g(\Phi)$ of guards will be larger than the number $n$ of variables in $\Phi$.
Finally, we want to point out that if $(p_1,\ldots,p_g)$ is a guard set, so is any permutation of the guards, whereas the set of solutions $(x_1,\ldots,x_n)$ to a polynomial equation $\Phi(x_1,\ldots,x_n)=0$ is clearly not closed under permutation in general.

We can now prove Theorem~\ref{cor:AlgebraicNumbers}, restated below.

\CorAlgebraicNumbers*

\begin{proof}
Let $P(x)$ be a polynomial of degree more than $0$ in one variable $x$ such that the equation $P(x)=0$ has $\alpha$ as a solution.
The equation might have other solutions as well, but we can choose integers $p_1,p_2,q_1,q_2$ such that $\alpha$ is the only solution in the interval $[p_1/q_1,p_2/q_2]$.
Then the formula $P(x)=0\;\land\; p_1\leq q_1x\;\land\; q_2x\leq p_2$ is an instance of the problem $\ETR$ with a unique solution $x=\alpha$.
Now, by Theorem~\ref{thm:Correspondance}, there exists a polygon $\poly$ and rational constants $c,d$ such that in any optimal guard set of $\poly$, one guard has coordinates $(c\alpha+d,0)$.
By subtracting $d$ from the $x$-coordinate of all corners of $\poly$ and then dividing all coordinates by $c$, we get a polygon $\poly'$ such that any optimal guard set of $\poly'$ has a guard at the point $(\alpha,0)$.
\end{proof}

\section{The Picasso theorem}\label{sec:Picasso}
This section is devoted to the proof of Theorem~\ref{thm:Picasso}:

\PicassoTheorem*

In order to construct the polygon $\poly_S$, we make use of the construction described in Section~\ref{sec:hardness} together with two additional gadgets, the Picasso gadget and the scaling gadget.

\subsection{Idea behind the construction}

We start with an informal description of the ideas behind the construction and then prove the theorem.

\begin{figure}[htbp]
	\centering
		\includegraphics{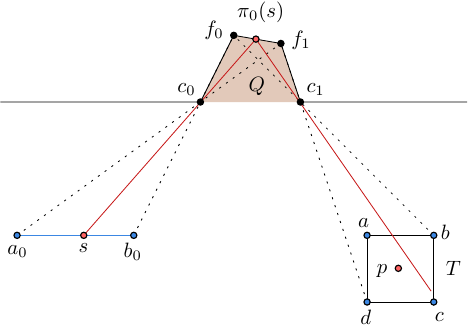}
				\includegraphics{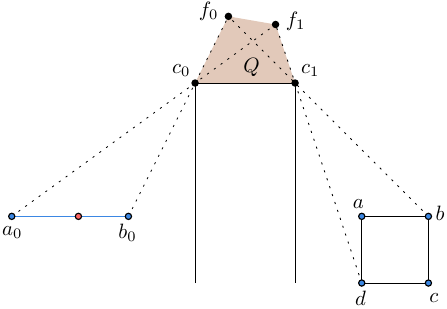}
	\caption{Left: A half-nook $Q$ for $a_0b_0$ and $T$. The guards at $s$ and $p$ together see the critical segment $f_0f_1$. Right: A half-umbra.}
	\label{fig:NookUmbra-Formal}
\end{figure}

It is easy to restrict a guard by a half-plane by a construction similar to that of nooks and umbras.
The idea is, as in the construction of nooks and umbras, that two guards have to see a critical segment, but now one of the guards is not restricted to a guard segment, but instead to a square.
We use the terms half-nooks and half-umbras to denote such constructions, see Figure~\ref{fig:NookUmbra-Formal}.
We give a more formal description a little later.

The next geometric observation is that one can describe a point uniquely as the intersection of three closed half-planes:
\begin{observation}\label{obs:half-planes}
Let $\ell_i,\ell_j,\ell_l$ be three non-vertical lines through a point $p$.
If $h_i$ and $h_l$ are half-planes bounded from above by $\ell_i$ and $\ell_l$, respectively, and $h_j$ is bounded from below by $\ell_j$, then $h_i\cap h_j\cap h_l = p$.
\end{observation}

This motivates us to design a gadget where a guard is restricted to one square $T$ while forming two half-nooks with guards on guard segments $r_i,r_l$, respectively, and one half-umbra with a guard on guard segment $r_j$.
The two half-nooks define two half-planes bounded from above and the half-umbra defines a half-plane bounded from below, see Figure~\ref{fig:PicassoSketch}.
Any position of a point $p$ gives rise to three guard positions $g_i,g_j,g_l$ on $r_i,r_j,r_l$, respectively, such that $p$ is in the boundary of each of the corresponding half-planes, and $p$ is the unique intersection point of the half-planes.
We denote by $x_i(p),x_j(p),x_l(p)$ the values that the three guards $g_i,g_j,g_l$ represent on their respective segments.
Using rational functions, one can express how $x_i(p),x_j(p)$, and $x_l(p)$ depends on the coordinates of $p$.

\begin{figure}[htbp]
	\centering
		\includegraphics{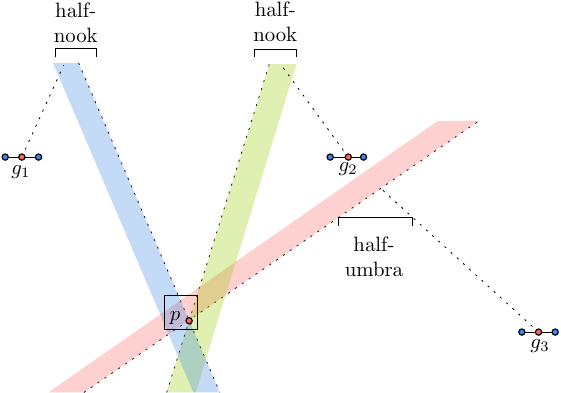}
	\caption{The point $p$ is inside the unit square $[0,1]^2$.
	The intersection of the three half-planes equals $p$.
	The two half-nooks define two half-planes bounded from above and the half-umbra defines a half-plane bounded from below.}
	\label{fig:PicassoSketch}
\end{figure}

We now explain how the Picasso gadget will be used.
Let a quantifier-free formula $\Phi$ of the first-order theory of the reals be given that has exactly two free variables $x,y$, and suppose that the set $S\mydef \{(x,y)\in\RR^2\colon \Phi(x,y)\}$ is a closed subset of $[0,1]^2$.
Let $\Psi$ be a quantifier-free formula of the first-order theory of the reals with five free variables $x,y,x_i,x_j,x_l$ such that when $p=(x,y)\in[0,1]^2$, the tuple $(x,y,x_i,x_j,x_l)$ satisfies $\Psi$ if and only if $x_\sigma = x_\sigma(p)$ for each $\sigma\in\{i,j,l\}$.
Such a formula exists since $p\longmapsto x_\sigma(p)$ is a rational function of $p$.
We consider the formula $\Phi'\mydef \Phi\land \Psi$.
Assume for the ease of presentation that $\Phi'$ is an instance of $\etrinv$.
This is in general not the case, but we will later explain how to get around that issue (this is where the other gadget introduced in this section, the scaling gadget, will be used).
We construct the polygon $\poly\mydef \poly(\Phi')$ as described in Section~\ref{sec:hardness}.
Then a point $p = (x,y)\in[0,1]^2$ is in $S$ if and only if there is an optimal guard set of $\poly$ such that for each $\sigma\in\{i,j,l\}$, the guards representing $x_\sigma$ specify the value $x_\sigma(p)$.
We get the polygon $\poly_S$ from Theorem~\ref{thm:Picasso} by adding the Picasso gadget to $\poly$ and copy the three variables $x_i,x_j,x_l$ to the guard segments $r_i,r_j,r_l$ in the gadget.
By shifting and scaling the polygon, we may assume that the square $T$ in the Picasso gadget is the unit square $[0,1]^2$.
It follows that $p\in S$ if and only if there is an optimal guard set of $\poly_S$ in which a guard is placed at $p$, and this proves the theorem.

\subsection{Half-nooks and half-umbras}

Half-nooks and half-umbras are defined in the same way as nooks and umbras with the only difference that one of the guard segments is replaced by a square $T$.
See Figure~\ref{fig:NookUmbra-Formal}.
Here we define only the case that the square is on the right of the guard segment.
In the case that the square is to the left, half-nooks and half-umbras are defined in an analogous way.

\begin{definition}[half-nook and half-umbra]
Let $\poly$ be a polygon containing a guard segments $r\mydef a_0b_0$ and an axis-parallel
square $T$, where $r$ is to the left of $T$.
Let $abcd$ be the corners of $T$ in clockwise order, where $a$ is the topmost left corner.
Let $c_0,c_1$ be two corners of $\poly$, such that $c_0$ is 
to the left of $c_1$.
Suppose that the rays $\overrightarrow{b_0c_0}$ and $\overrightarrow{bc_1}$ 
intersect at a point $f_0$, the rays $\overrightarrow{a_0c_0}$ 
and $\overrightarrow{dc_1}$ intersect at a point $f_1$, and 
that $Q\mydef c_0c_1f_1f_0$ is a convex quadrilateral 
contained in $\poly$.
We define the function $\pi_0\colon r\longrightarrow f_0f_1$ 
such that $\pi_0(p)$ is the intersection of the 
ray $\overrightarrow{pc_0}$ with the line segment $f_0f_1$.
Analogously, we define $\pi_1 : T \longrightarrow f_0f_1$ 
such that $\pi_1(p)$ is the intersection of 
the ray $\overrightarrow{pc_1}$ with the line segment $f_0f_1$. 
We suppose that $\pi_0$ is bijective and $\pi_1$ is surjective.

We say that $Q$ is a \emph{half-nook} for $r$ and $T$ if for 
every $p\in r$, a guard at $p$ can see all of the 
segment $\pi_0(p)f_{1}$ but nothing else of $f_0f_1$ and 
for every $p \in T$, a guard at $p$ can see all of the segment 
$\pi_1(p)f_{0}$ but nothing else of $f_0f_1$.

We say that $Q$ is a \emph{half-umbra} for $r$ and $T$ if 
for every $p\in r$, a guard at $p$ can see all of the 
segment $\pi_0(p)f_{0}$ but nothing else of $f_0f_1$ and 
for every $p \in T$, a guard at $p$ can see all of the segment 
$\pi_1(p)f_{1}$ but nothing else of $f_0f_1$.

The functions $\pi_0,\pi_1$ are called \emph{projections} of the half-nook or the half-umbra.

We deonte $c_0,c_1$ as the \emph{shadow corners} and $f_0f_1$ is the \emph{critical segment} of the half-nook or half-umbra.
\end{definition}

\begin{observation}\label{obs:half-nookPlane}
	Let $Q$ be a half-nook (half-umbra) for a guard segment $r$ and a square $T$, and let $s\in r$ and $p\in T$.
	Then it holds that $p$ and $s$ together see the half-nook (half-umbra) if and only  if $p$ is on or \emph{below} (\emph{above}) the line $\overleftrightarrow{\pi_0(s)c_1}$.
\end{observation}


\subsection{The Picasso gadget}

We are now ready to give an explicit description of the Picasso-gadget, see Figure~\ref{fig:FullPicasso}.
Two edges at the left side of the gadget are not fully shown.
They end at corners $c_1\mydef (-CN^2,CN^2)$ and $d_1\mydef (-CN^2,CN^2+1.5)$, respectively, where $c_1$ is the right shadow corner of the corridor as in Section~\ref{sec:copy}.
The gadget contains three guard segments $r_i \mydef a_i'b_i'$, $r_j \mydef a_j'b_j'$, and $r_l \mydef a_l'b_l'$, each of width $1.5$, defined by the left endpoints
\[ a_i' \mydef (-31,-7),\ a_j' \mydef (0,-24.5),\ a_l' \mydef (4,-15). \]
The gadget contains an axis-parallel square $T$, in which a guard will ``realize'' the set $S$.
The side length of $T$ is $2$ and the upper left corner is $(-18,-24)$.
Let $Q_i,Q_j,Q_l$ be a half-nook, half-umbra, half-nook of $r_i,r_j,r_l$ and $T$, respectively, with shadow corners:
\[ Q_i: (-23,20)\text{ and }(-21,20),\quad Q_j: (-9,-22)\text{ and }(-6.5,-22),\quad Q_l: (2,-10)\text{ and }(4,-10). \]
For $\sigma\in\{i,j,l\}$, let $\pi_{\sigma 0}$ and $\pi_{\sigma 1}$ be the projections of $Q_\sigma$.
Furthermore, let $x_\sigma\colon T\longrightarrow [1/2,2]$ be defined so that $x_\sigma(p)$ is the value represented by $\pi_{\sigma 0}^{-1}(\pi_{\sigma 1}(p))$.
There is a stationary guard position $g_s$ that sees $Q_j$, but nothing above the critical segment of $Q_j$.

We then use a corridor as described in Section~\ref{sec:copy} to copy in the values of three variables $x_i,x_j,x_l$ at $r_i,r_j,r_l$.
Recall that the endpoints of the right entrance of the corridor are denoted $c_1,d_1$.
In order to attach the gadget to the corridor, we first scale it down by the factor $\frac 1{CN^2}$ and then translate it so that the points $c_1,d_1$ of the gadget coincides with the endpoints of the right entrance of the corridor with the same names.
We thus obtain that the point $O\mydef (0,0)$ in the gadget becomes the point $m\mydef c_1+(1,-1)$ in $\poly_S$.

\begin{lemma}\label{lem:picasso-gadget}
Let $\poly'_{\textrm{Pic}}$ be the polygon obtained from the Picasso gadget by closing it by adding the edge $c_1d_1$.
Consider a set of guards $G \subset \poly'_{\textrm{Pic}}$ of size at most $5$.
If $G$ guards $\poly'_{\textrm{Pic}}$, then there is exactly one guard placed on each guard segment $r_i, r_j,r_l$, a guard at $g_s$, and a guard in $T$.
Consider a point $p\in T$ and suppose that $G$ contains $x_\sigma(p)$ for each $\sigma\in\{i,j,l\}$.
Then $G$ guards $\poly'_{\textrm{Pic}}$ if and only if $G$ also contains $g_s$ and $p$.
\end{lemma}

\begin{proof}
Let $G$ be a set of guards of size at most $5$.
Similarly as in Lemma~\ref{lem:addition-gadget-7-guards} we can show that if $\poly'_{\textrm{Pic}}$ is guarded by $G$, there must be exactly one guard in $T$, at one each guard segment, and one at $g_s$.
It is straightforward to check that if $G$ contains $p\in T$, $g_s$, and $x_\sigma(p)$ for each $\sigma\in\{i,j,l\}$, then $G$ guards $\poly'_{\textrm{Pic}}$.
On the other hand, if there is a point $p\in T$ such that $G$ contains guards $g_s$ and $x_\sigma(p)$ for each $\sigma\in\{i,j,l\}$, then Observations~\ref{obs:half-planes} and \ref{obs:half-nookPlane} imply that $G$ must contain $p$ in order to guard $\poly'_{\textrm{Pic}}$.
\end{proof}

We now need to show that our gadget construction satisfies the conditions of Lemma~\ref{lemma:r-slabs}.

\begin{proof}[Proof of Lemma~\ref{lemma:r-slabs} for the Picasso gadget.]
We make use of Lemma~\ref{lem:ugly}.
First check that the length of every guard segment is $\frac{3/2}{CN^2}$ as required.
Note that all guard segments are contained in the square $m + [-\Delta,\Delta ] \times [-\Delta,\Delta ]$, with $\Delta \mydef \frac{50}{CN^2}$ and $m\mydef c_1 + (1,-1)$.
Furthermore, $a_j'$ is on the line through $a_i' + (\delta,0)$ and direction $(1,-1)$, and $a_l'$ is on the line through $a_i' + (2\delta,0)$ and direction $(1,-1)$.
Thus Property~\ref{rslabs:1} and \ref{rslabs:2} of Lemma~\ref{lemma:r-slabs} are met.
It remains to show that no stationary guard or the guard in the square $T$ can see into the corridor.
That is clear from the construction.
\end{proof}

\begin{figure}[htbp]
	\centering
		\includegraphics[trim= 0.5cm 1cm 0.5cm 1cm, clip,width=0.7\textwidth]
		{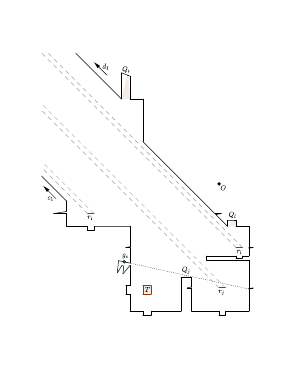}
		\caption{The Picasso-gadget.}
		\label{fig:FullPicasso}
\end{figure}

\subsection{The scaling gadget}

Before we can move to the final step of proving Theorem~\ref{thm:Picasso}, we need another tool, which is a gadget to scale variables.
This will be needed to reverse the scaling that occurs as we apply the transformation described in Section~\ref{sec:etrV}.
The gadget is shown in Figure~\ref{fig:Scaling-gadget}.
Two edges at the left side of the gadget are not fully shown.
They end at corners $c_1\mydef (-CN^2,CN^2)$ and $d_1\mydef (-CN^2,CN^2+1.5)$, respectively, where $c_1$ is the right shadow corner of the corridor as in Section~\ref{sec:copy}.
The gadget contains two guard segments $r_i \mydef a_i'b_i'$ and $r_j \mydef a_j'b_j'$, each of width $1.5$, defined by the left endpoints $a_i' \mydef (0,0)$ and $a'_j\mydef(13.5,0)$.
Let $ab$ be a segment on $r_i$, where $a$ is to the left of $b$.
The scaling gadget contains a copy-nook $Q_n$ of $ab$ and $r_j$ with shadow corners $(5.5,3)$ and $(9.5,3)$.
Likewise, it contains a copy-umbra $Q_u$ of $ab$ and $r_j$ with shadow corners $(5.5,11)$ and $(9.5,11)$.
It is important here to note that in general, $ab$ is not the full segment $r_i$, but just a subset.
This enforces that the value represented by $a$ of $r_i$ will be copied to $1/2$ on $r_j$ and the value represented by $b$ will be copied to $2$.
There is a stationary guard position $g_s$ that sees $Q_u$, but nothing above the critical segment of $Q_u$.

\begin{figure}
	\centering
		\includegraphics[clip, trim=1cm 0.5cm 1cm 0.5cm, width=0.7\textwidth]{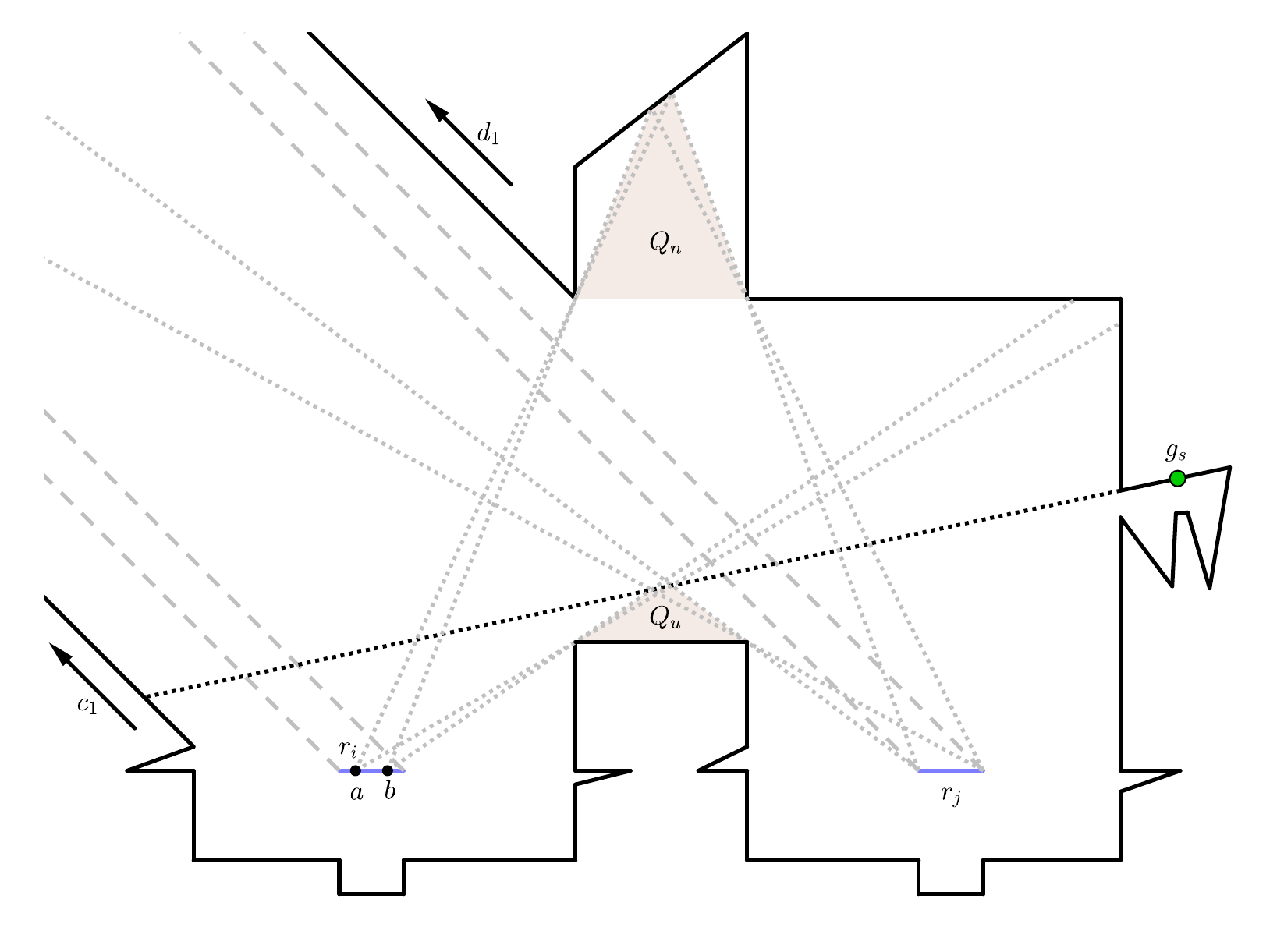}
	\caption{The scaling gadget.}
	\label{fig:Scaling-gadget}
\end{figure}

\begin{lemma}\label{lem:scaling-gadget}
Let $\poly'_{\textrm{sca}}$ be the polygon obtained from a scaling gadget by closing it by adding the edge $c_1d_1$.
A set of guards $G \subset \poly'_{\textrm{sca}}$ of cardinality at most $3$ guards $\poly'_{\textrm{sca}}$ if and only if
\begin{itemize}
\item there is exactly one guard placed on each guard segment $r_i, r_j$ and at $g_s$, and
\item the guards $p_i,p_j$ at $x_i, x_j$, respectively, satisfy that $p_i\in ab$ and $\frac{\|ap_i\|}{\|ab\|}=\frac{\|a'_jp_j\|}{\|a'_jb'_j\|}$.
\end{itemize}
\end{lemma}

\begin{proof}
Assume that $\poly'_{\textrm{sca}}$ is guarded by a set $G$ of at most $3$ guards.
Similarly as in Lemma~\ref{lem:addition-gadget-7-guards} we can show that there must be exactly one guard at each of $r_i,r_j$ and at $g_s$.
Since the guards $p_i,p_j$ must together see the critical segments of $Q_n$ and $Q_u$, it follows that $p_i\in ab$ and $\frac{\|ap_i\|}{\|ab\|}=\frac{\|a'_jp_j\|}{\|a'_jb'_j\|}$.

Now assume that there is exactly one guard placed on each guard segment $r_i, r_j$ and at $g_s$, and that $p_i\in ab$ and $\frac{\|ap_i\|}{\|ab\|}=\frac{\|a'_jp_j\|}{\|a'_jb'_j\|}$.
Then all of $Q_u$ and $Q_n$ is seen by the guards.
The remaining area is clearly also seen by the guards.
\end{proof}

We then use a corridor as described in Section~\ref{sec:copy} to copy in the values of two variables $x_i,x_j$ at $r_i,r_j$.
Recall that the endpoints of the right entrance of the corridor are denoted $c_1,d_1$.
In order to attach the gadget to the corridor, we first scale it down by the factor $\frac 1{CN^2}$ and then translate it so that the points $c_1,d_1$ of the gadget coincides with the endpoints of the right entrance of the corridor with the same names.
We thus obtain that the point $a_i'$ in the gadget becomes the point $m\mydef c_1+(1,-1)$ in $\poly_S$.
We need to prove Lemma~\ref{lemma:r-slabs} in order to guarantee that the corridor copies the guard positions appropriately.

\begin{proof}[Lemma~\ref{lemma:r-slabs} for the Scaling-Gadget]
We make use of Lemma~\ref{lem:ugly}.
First check that the length of $r_i$ and $r_j$ is $\frac{3/2}{CN^2}$ as required.
Both guard segments are contained in the square $m + [-\Delta,\Delta ] \times [-\Delta,\Delta ]$, with $\Delta \mydef \frac{50}{CN^2}$ and $m\mydef c_1 + (1,-1)$.
Furthermore, the left endpoint of $a_j'$ is on the point $a_i' + (\delta,0)$.
Thus Property~\ref{rslabs:1} and \ref{rslabs:2} of Lemma~\ref{lemma:r-slabs} are met.
It remains to observe that $g_s$ cannot see into the corridor.
\end{proof}

\subsection{Proof of the Picasso Theorem}

Using the scaling gadget and the Picasso gadget, we are now ready to describe an art gallery $\poly_S$ as described in Theorem~\ref{thm:Picasso}.

\begin{proof}[Proof of Theorem~\ref{thm:Picasso}]
Let a quantifier-free formula $\Phi$ of the first-order theory of the reals be given that has exactly two free variables $x,y$, and suppose that the set of solutions $S\mydef \{(x,y)\in\RR^2\colon \Phi(x,y)\}$ is a closed subset of $[0,1]^2$.
The theorem is trivially true of $S=\emptyset$, so assume that $S$ is non-empty.
Recall that there is an axis-parallel square $T$ in the Picasso gadget and three guard segments $r_i,r_j,r_l$ and corresponding half-nooks $Q_i,Q_l$ and a half-umbra $Q_j$.
We identify each point in $T$ with the corresponding point in $[0,1]^2$ under the natural linear bijection, so that the upper left corners correspond, etc.
Thus, with slight abuse of notation, consider $S$ as a subset of $T$.
Let $g_\sigma$ be a point on $r_\sigma$, $\sigma\in\{i,j,l\}$, and $p$ a point in $T$.
In order to see the critical segment of $Q_\sigma$ together, the point $g_\sigma$ restricts $p$ to a half-plane.
The point $\pi^{-1}_{\sigma 0}(\pi_{\sigma 1}(p))$ is the point on $r_\sigma$ such that $p$ is on the boundary of that half-plane.
The point $\pi^{-1}_{\sigma 0}(\pi_{\sigma 1}(p))$ specifies the number $x_\sigma(p)\in[1/2,2]$.
The function $x_\sigma\colon T\longrightarrow [1/2,2]$ is a rational function of the $x$- and $y$-coordinate of $p$.
Therefore, there is a quantifier-free formula of the first-order theory of the reals $\Psi$ with five free variables $x,y,x_i,x_j,x_l$ such that when $p=(x,y)\in T$, the tuple $(x,y,x_i,x_j,x_l)$ satisfies $\Psi$ if and only if $x_\sigma = x_\sigma(p)$ for each $\sigma\in\{i,j,l\}$.
We consider the formula $\Phi'\mydef \Phi\land \Psi$, which has a compact set of solutions in $\RR^5$.

We now apply Corollary~\ref{thm:StrongCorrespondance} to obtain an instance $\Phi''$ of $\etrinv$ in which the variables of $\Phi'$ appear scaled down and shiftet.
Consider the variable $x_\sigma$ in $\Phi'$, which has domain $[1/2,2]$.
In $\Phi''$, there is a variable $x'_\sigma$ with corresponding domain $[a,b]\subseteq[1/2,2]$.
Thus, $x_\sigma$ has been scaled down by a factor of $\frac{b-a}{3/2}$ and shifted in order to get $x'_\sigma$.
Let the remaining variables in $\Phi''$ be $\{y'_1,\ldots,y'_k\}$, so that the complete set of variables is $\{x'_i,x'_j,x'_l,y'_1,\ldots,y'_k\}$.
We now make another instance $\Phi'''$ of $\etrinv$, which is identical to $\Phi''$, except that $\Phi'''$ has three extra variables $x_i,x_j,x_l$ appearing in no equations.
Let $(\poly,g)$ be the instance of the art gallery problem as described in Theorem~\ref{thm:final} for the instance $\Phi''$ of $\etrinv$.
Thus, in the main area of $\poly$, there are guard segments representing $x_i,x_j,x_l$, but they are not copied into any gadget.
We might need to use a slightly larger value of $N$ than defined in Section~\ref{sec:poly}, in order to have vertical space for $4$ more gadgets to the right.
For each $\sigma\in\{i,j,l\}$, consider the variable $x'_\sigma$ with domain $[a,b]\subseteq[1/2,2]$, as defined above.
We construct an art gallery $\poly'$ from $\poly$ by adding a scaling gadget into which we copy $x'_\sigma$ and $x_\sigma$, such that the domain $[a,b]$ of $x'_\sigma$ is scaled up to $[1/2,2]$ of $x_\sigma$.
Finally, we also add the Picasso gadget to $\poly'$ and copy in $x_i,x_j,x_l$.
By shifting and scaling $\poly'$, we obtain that the square $T$ in the Picasso gadget coincides with $[0,1]^2$.
It now follows from from Lemma~\ref{lem:scaling-gadget} and Lemma~\ref{lem:picasso-gadget} that for each $p\in[0,1]^2$, there is an optimal guard set of $\poly_S$ containing $p$ if and only if $p\in S$.
\end{proof}


\section{Concluding remarks}\label{sec:concluding}

\subsection{Removing degeneracies}
The polygon $\poly(\Phi)$ described in Section~\ref{sec:hardness} is degenerate in the sense that it contains many triples of collinear corners.
We now show that the construction can be slightly modified in order to avoid all such collinear triples.
A general position assumption makes a remarkable difference in the complexity of some problems, such as recognizing point visibility graphs, which is trivially in $\text{P}$ for points in general position and $\ER$-complete for general sets of points~\cite{cardinal2017recognition}.
We therefore believe it is interesting to note that the $\ER$-completeness of the art gallery problem does not rely on degeneracies.

There are four collinear corners for each guard segment $s$ in $\poly$, two of which are in a spike forcing the guard to be below $s$ and two forcing the guard to be above $s$.
We claim that removing the latter spike, and thus only bounding the guard to be below $s$, does not introduce new optimal guard sets in $\poly$.
By making the angle of the remaining spike sufficiently small, the guard is restricted to a thin quadrilateral $\square_s$ below $s$ such that a guard in $\square_s$ can only see points on the critical segments that a guard on $s$ is intended to see.
We conclude that there must be a guard in any such quadrilateral $\square_s$ and that they must see the same critical segments, and the same regions $\Gamma$ in the addition and orientation gadgets, as the guards on the respective guard segments see in the original construction.

In the following we argue that each guard must be on the respective guard segment and not below.
Consider two guard segments $s_0\mydef a_0b_0$ and $s_1\mydef a_1b_1$ for which there is a nook $Q_1$ and an umbra $Q_2$.
Suppose furthermore that there is a bijection $\pi$ between $s_0$ and $s_1$ of points that can together see the critical segments $f_0f_1$ and $f_2f_3$ of $Q_1$ and $Q_2$, respectively.
This is the case for each pair of guard segments which have a copy-nook and a copy-umbra and for the two guard segments in each inversion gadget.
See Figure~\ref{fig:copyGuardGP}.
If a guard $g$ is placed below $s_0$, it sees strictly less of $f_0f_1$ and $f_2f_3$ than every point in an interval of $s_0$, and due to the existence of the bijection $\pi$, it follows that no point on or below $s_1$ can see the remaining parts of $f_0f_1$ and $f_2f_3$.
Hence, if two guards see $f_0f_1$ and $f_2f_3$ together, they must be on the segments $s_0$ and $s_1$.

\begin{figure}[htbp]
\centering
\includegraphics[clip, trim= 1cm 1cm 1cm 1cm,width=0.8\textwidth]
{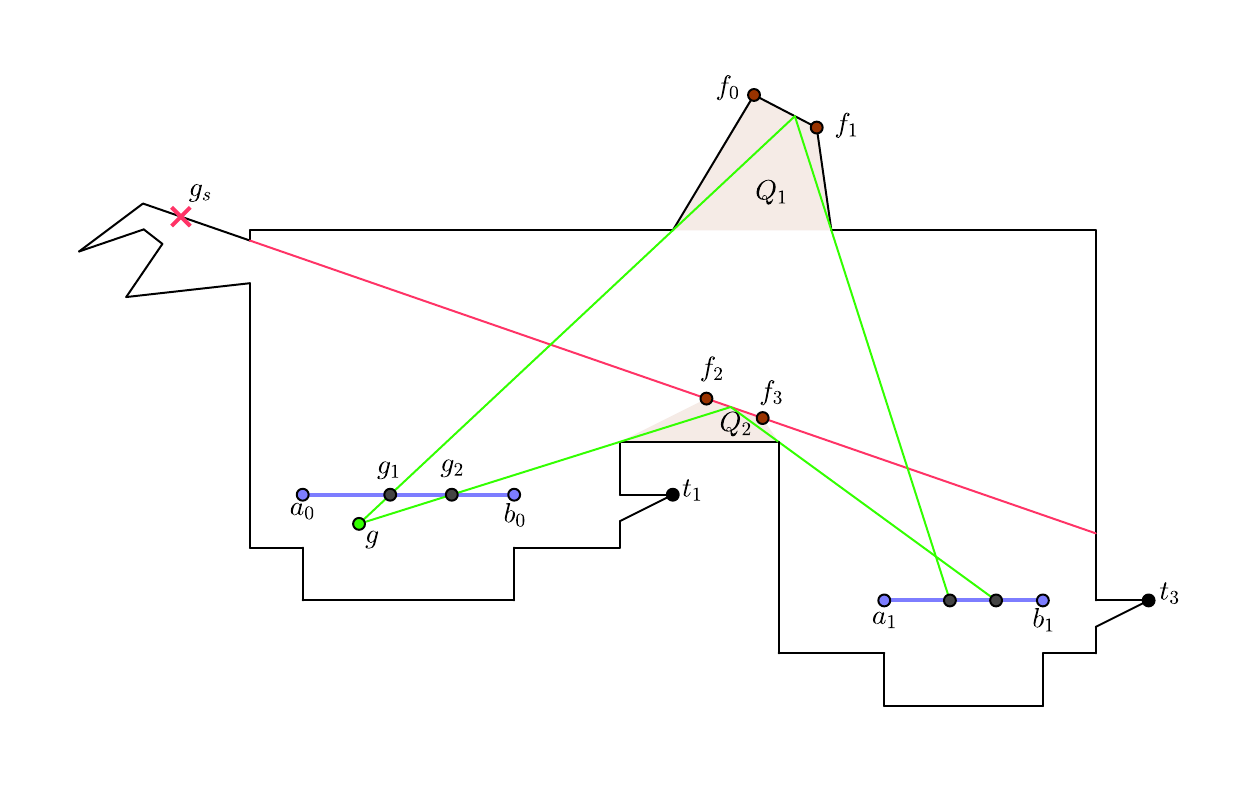}
\caption{A guard at $g$ below the guard segment $a_0b_0$ sees strictly less of the critical segments $f_0f_1$ and $f_2f_3$ than any point on the segment $g_1g_2\subset a_0b_0$.
Hence, no guard on or below $a_1b_1$ can see $f_0f_1$ and $f_2f_3$ together with $g$.}
\label{fig:copyGuardGP}
\end{figure}

A similar argument applies to the addition and orientation gadgets, as follows.
Consider the $\geq$-addition gadget.
As the reader may recall, the guard segments $r_i,r_j,r_l$, representing $x_i,x_j,x_l$, respectively, are copies of guard segments in the main area of $\poly$, and a guard must therefore be placed on each segment by the above remark.
The guard at $r'_i$ is only a weak copy of the guard at $r_i$.
However, if the guard is placed below $r'_i$, it sees less of $\Gamma$, and hence the guard at $r_l$ also has to be further to the left.
Hence, the inequality $x_i+x_j\geq x_l$ still holds.

Finally, we have multiple collinear points on the bottom wall of $\poly$ and on vertical lines passing through the corridors and gadgets.
It is easy to see that small vertical pertubations of the points on the bottom wall and horizontal pertubations of the corridors and gadgets do not affect the overall construction (this, of course, requires adjusting the corridors accordingly).
We conclude that the art gallery problem is $\ER$-complete even for polygons with corners in general position.

The constructed polygon $\poly$ also contains degeneracies in the form of three edges with linear extensions intersecting each other at the same point -- each stationary guard position $g_s$ is such a point.
Let $e$ be the edge of $\poly$ containing $g_s$.
In order to avoid the degeneracy, we expand $\poly$ by a triangle with one edge coincident to $e$ and the third corner in the exterior of $\poly$ and close to $g_s$, thus rerouting the polygon boundary around $g_s$ so that $g_s$ becomes an interior point where only two edge extensions intersect each other.
We leave it for the reader to observe that this does not introduce new optimal guard sets of the polygon and conclude that the problem remains $\ER$-complete even when there is no triple of edges with extensions intersecting at a common point.

\subsection{Open problems}

Note that the guard segments are in the interior of $\poly$.
It is easy to adjust the construction so that each guard segment is contained in an edge of the boundary of $\poly$.
We obtain that the variant of the problem where the guards are restricted to the boundary is also $\ER$-complete.
The critical segments of umbras and the regions $\Gamma$ in addition and orientation gadgets are in the interior of $\poly$.
It is therefore an interesting open question what the complexity is of guarding only the boundary of a given simple polygon.
Guarding the boundary using vertex or point guards is known to be $\text{NP}$-hard~\cite{laurentini1999guarding}.

We proved that we can encode an instance of the art gallery problem as an existential formula using $(n+k)^6$ variables, where $n$ is the number of corners of the polygon and $k$ is the number of guards.
It is natural to investigate how many variables are in fact needed, as fewer variables lead to faster algorithms for the art gallery problem.
For instance, is it necessary to use $\Omega(k)$ variables?
Moitra~\cite{moitra2016almost} worked on a similar problem in the context of computing the nonnegative rank of a matrix.

\section*{Acknowledgments}

We would like to thank S\'andor Fekete and Anna Lubiw for their careful reading of an earlier version of the paper and many helpful suggestions for improvement.
We also wish to thank \'Edouard Bonnet, Udo Hoffmann, Linda Kleist, P\'eter Kutas, and G\"unter Rote for useful discussions, and Christian Knauer for pointers to relevant literature.

Mikkel Abrahamsen was partially supported by Mikkel Thorup's Advanced Grant DFF-0602-02499B from the Danish Council for Independent Research under the Sapere Aude research career programme.
Anna Adamaszek was partially supported by the Danish Council for Independent Research DFF-MOBILEX mobility grant.
Tillmann Miltzow was partially supported by ERC grant no.~280152, PARAMTIGHT: ``Parameterized complexity and the search for tight complexity results''.

\bibliography{lib}{}

\begin{thebibliography}{10}

\bibitem{dictionary}
Umbra --- {D}ictionary.com.
\newblock Accessed 14th of March 2017,
  \url{http://www.dictionary.com/browse/umbra}.

\bibitem{abel}
Zachary Abel, Erik~D. Demaine, Martin~L. Demaine, Sarah Eisenstat, Jayson
  Lynch, and Tao~B. Schardl.
\newblock Who needs crossings? {H}ardness of plane graph rigidity.
\newblock In {\em 32nd International Symposium on Computational Geometry (SoCG
  2016)}, pages 3:1--3:15, 2016.

\bibitem{abrahamsen2017irrational}
Mikkel Abrahamsen, Anna Adamaszek, and Tillmann Miltzow.
\newblock Irrational guards are sometimes needed.
\newblock In {\em 33rd International Symposium on Computational Geometry (SoCG
  2017)}, pages 3:1--3:15, 2017.

\bibitem{aggarwal1984art}
Alok Aggarwal.
\newblock {\em The art gallery theorem: its variations, applications and
  algorithmic aspects}.
\newblock PhD thesis, The Johns Hopkins University, 1984.

\bibitem{e_k1974design}
Alfred~V. Aho, John~E. Hopcroft, and Jeffrey~D. Ullman.
\newblock {\em The design and analysis of computer algorithms}.
\newblock Addison-Wesley, 1975.

\bibitem{basu2006algorithms}
Saugata Basu, Richard Pollack, and Marie-Fran\c{c}oise Roy.
\newblock {\em Algorithms in real algebraic geometry}.
\newblock Springer-Verlag Berlin Heidelberg, 2006.

\bibitem{basu1996combinatorial}
Saugata Basu, Richard Pollack, and Marie-Fran{\c{c}}oise Roy.
\newblock On the combinatorial and algebraic complexity of quantifier
  elimination.
\newblock {\em Journal of the ACM}, 43(6):1002--1045, 1996.

\bibitem{BasuR10}
Saugata Basu and Marie{-}Fran{\c{c}}oise Roy.
\newblock Bounding the radii of balls meeting every connected component of
  semi-algebraic sets.
\newblock {\em Journal of Symbolic Computation}, 45(12):1270--1279, 2010.

\bibitem{belleville1991computing}
Patrice Belleville.
\newblock Computing two-covers of simple polygons.
\newblock Master's thesis, McGill University, 1991.

\bibitem{bienstock1991some}
Daniel Bienstock.
\newblock Some provably hard crossing number problems.
\newblock {\em Discrete \& Computational Geometry}, 6(3):443--459, 1991.

\bibitem{BonnetM16}
{\'{E}}douard Bonnet and Tillmann Miltzow.
\newblock Parameterized hardness of art gallery problems.
\newblock In {\em 24th Annual European Symposium on Algorithms (ESA)}, pages
  19:1--19:17, 2016.

\bibitem{DBLP:conf/compgeom/BorrmannRSFFKNST13}
Dorit Borrmann, Pedro~J. de~Rezende, Cid~C. de~Souza, S{\'{a}}ndor~P. Fekete,
  Stephan Friedrichs, Alexander Kr{\"{o}}ller, Andreas N{\"{u}}chter,
  Christiane Schmidt, and Davi~C. Tozoni.
\newblock Point guards and point clouds: {S}olving general art gallery
  problems.
\newblock In {\em Symposuim on Computational Geometry 2013 (SoCG 2013)}, pages
  347--348, 2013.

\bibitem{broden2001guarding}
Bj{\"o}rn Brod{\'e}n, Mikael Hammar, and Bengt~J. Nilsson.
\newblock Guarding lines and 2-link polygons is {APX}-hard.
\newblock In {\em Proceedings of the 13th Canadian Conference on Computational
  Geometry (CCCG 2001)}, pages 45--48, 2001.

\bibitem{canny1988some}
John Canny.
\newblock Some algebraic and geometric computations in {PSPACE}.
\newblock In {\em Proceedings of the twentieth annual ACM symposium on Theory
  of computing (STOC 1988)}, pages 460--467. ACM, 1988.

\bibitem{Cardinal:2015:CGC:2852040.2852053}
Jean Cardinal.
\newblock Computational geometry column 62.
\newblock {\em SIGACT News}, 46(4):69--78, 2015.

\bibitem{cardinal2017recognition}
Jean Cardinal and Udo Hoffmann.
\newblock Recognition and complexity of point visibility graphs.
\newblock {\em Discrete \& Computational Geometry}, 57(1):164--178, 2017.

\bibitem{de2000computational}
Mark de~Berg, Marc van Kreveld, Mark Overmars, and Otfried Cheong.
\newblock {\em Computational Geometry: Algorithms and Applications (3rd
  edition)}.
\newblock Springer-Verlag, 2008.

\bibitem{engineering}
Pedro~J. de~Rezende, Cid~C. de~Souza, Stephan Friedrichs, Michael Hemmer,
  Alexander Kr{\"{o}}ller, and Davi~C. Tozoni.
\newblock Engineering art galleries.
\newblock In Peter~Sanders Lasse~Kliemann, editor, {\em Algorithm Engineering
  -- Selected Results and Surveys}, pages 379--417. Springer International
  Publishing, 2016.

\bibitem{devadoss2011discrete}
Satyan~L. Devadoss and Joseph O'Rourke.
\newblock {\em Discrete and Computational Geometry}.
\newblock Princeton University Press, 2011.

\bibitem{AreasKleist}
Michael~G. Dobbins, Linda Kleist, Tillmann Miltzow, and Pawe\l{} Rz{\c
  a}{\.z}ewski.
\newblock {$\forall \exists \mathbb{R}$-completeness and area-universality}.
\newblock 2017.
\newblock Preprint, \url{https://arxiv.org/abs/1712.05142}.

\bibitem{DBLP:journals/ipl/EfratH06}
Alon Efrat and Sariel Har{-}Peled.
\newblock Guarding galleries and terrains.
\newblock {\em Information Processing Letters}, 100(6):238--245, 2006.

\bibitem{eidenbenz2001inapproximability}
Stephan Eidenbenz, Christoph Stamm, and Peter Widmayer.
\newblock Inapproximability results for guarding polygons and terrains.
\newblock {\em Algorithmica}, 31(1):79--113, 2001.

\bibitem{DiscretizeTerrain}
Stephan Friedrichs, Michael Hemmer, James King, and Christiane Schmidt.
\newblock The continuous 1.5{D} terrain guarding problem: Discretization,
  optimal solutions, and {PTAS}.
\newblock {\em Journal of Computational Geometry}, 7(1):256--284, 2016.

\bibitem{garg2015etr}
Jugal Garg, Ruta Mehta, Vijay~V. Vazirani, and Sadra Yazdanbod.
\newblock {ETR}-completeness for decision versions of multi-player (symmetric)
  {N}ash equilibria.
\newblock In {\em Proceedings of the 42nd International Colloquium on Automata,
  Languages, and Programming (ICALP 2015), part 1}, volume 9134 of {\em Lecture
  Notes in Computer Science (LNCS)}, pages 554--566, 2015.

\bibitem{TerrainGuardingNP}
James King and Erik Krohn.
\newblock Terrain guarding is {NP}-hard.
\newblock In {\em Proceedings of the Twenty-First Annual {ACM-SIAM} Symposium
  on Discrete Algorithms ({SODA} 2010)}, pages 1580--1593, 2010.

\bibitem{ApproXKirkpatrick15}
David~G. Kirkpatrick.
\newblock An {$O(\lg \lg \mathrm{OPT})$}-approximation algorithm for
  multi-guarding galleries.
\newblock {\em Discrete {\&} Computational Geometry}, 53(2):327--343, 2015.

\bibitem{MonotonHard}
Erik Krohn and Bengt~J. Nilsson.
\newblock Approximate guarding of monotone and rectilinear polygons.
\newblock {\em Algorithmica}, 66(3):564--594, 2013.

\bibitem{laurentini1999guarding}
Aldo Laurentini.
\newblock Guarding the walls of an art gallery.
\newblock {\em The Visual Computer}, 15(6):265--278, 1999.

\bibitem{DBLP:journals/tit/LeeL86}
D.~T. Lee and Arthur~K. Lin.
\newblock Computational complexity of art gallery problems.
\newblock {\em {IEEE} Transactions on Information Theory}, 32(2):276--282,
  1986.

\bibitem{AnnaPreparation}
Anna Lubiw, Tillmann Miltzow, and Debajyoti Mondal.
\newblock The complexity of drawing a graph in a polygonal region.
\newblock 2018.
\newblock Preprint, \url{https://arxiv.org/abs/1802.06699}.

\bibitem{maple}
{Maple 2016.1, Maplesoft, a division of Waterloo Maple Inc., Waterloo, Ontario.
  Maple is a trademark of Waterloo Maple Inc.}

\bibitem{matouvsek2002lectures}
Ji{\v{r}}{\'\i} Matou{\v{s}}ek.
\newblock {\em Lectures on Discrete Geometry}, volume 212 of {\em Graduate
  Texts in Mathematics}.
\newblock Springer-Verlag New York, 2002.

\bibitem{matousek2014intersection}
Ji{\v{r}}{\'{\i}} Matou{\v{s}}ek.
\newblock Intersection graphs of segments and $\exists \mathbb{R}$.
\newblock 2014.
\newblock Preprint, \url{https://arxiv.org/abs/1406.2636}.

\bibitem{mcdiarmid2013integer}
Colin McDiarmid and Tobias M{\"u}ller.
\newblock Integer realizations of disk and segment graphs.
\newblock {\em Journal of Combinatorial Theory, Series B}, 103(1):114--143,
  2013.

\bibitem{mnev1988universality}
Nicolai~E Mn{\"e}v.
\newblock The universality theorems on the classification problem of
  configuration varieties and convex polytopes varieties.
\newblock In Oleg~Y. Viro, editor, {\em Topology and geometry -- Rohlin
  seminar}, pages 527--543. Springer-Verlag Berlin Heidelberg, 1988.

\bibitem{moitra2016almost}
Ankur Moitra.
\newblock An almost optimal algorithm for computing nonnegative rank.
\newblock {\em SIAM Journal on Computing}, 45(1):156--173, 2016.

\bibitem{o1987art}
Joseph O'Rourke.
\newblock {\em Art Gallery Theorems and Algorithms}.
\newblock Oxford University Press, 1987.

\bibitem{o1998computational}
Joseph O'Rourke.
\newblock {\em Computational Geometry in C}.
\newblock Cambridge University Press, 1998.

\bibitem{2004handbook}
Joseph O'Rourke.
\newblock Visibility.
\newblock In Jacob~E. Goodman and Joseph O'Rourke, editors, {\em Handbook of
  Discrete and Computational Geometry}, chapter~28. Chapman \& Hall/CRC, second
  edition, 2004.

\bibitem{o1983some}
Joseph O'Rourke and Kenneth Supowit.
\newblock Some {NP}-hard polygon decomposition problems.
\newblock {\em IEEE Transactions on Information Theory}, 29(2):181--190, 1983.

\bibitem{wiki:quintic}
{Quintic function -- {W}ikipedia{,} The Free Encyclopedia}, 2017.
\newblock [Online; accessed 14-March-2017].

\bibitem{richter1995realization}
J{\"u}rgen Richter-Gebert and G{\"u}nter~M. Ziegler.
\newblock Realization spaces of 4-polytopes are universal.
\newblock {\em Bulletin of the American Mathematical Society}, 32(4):403--412,
  1995.

\bibitem{Schaefer2010}
Marcus Schaefer.
\newblock Complexity of some geometric and topological problems.
\newblock In {\em Proceedings of the 17th International Symposium on Graph
  Drawing (GD 2009)}, volume 5849 of {\em Lecture Notes in Computer Science
  (LNCS)}, pages 334--344. Springer, 2009.

\bibitem{schaefer2013realizability}
Marcus Schaefer.
\newblock Realizability of graphs and linkages.
\newblock In J\'{a}nos Pach, editor, {\em Thirty Essays on Geometric Graph
  Theory}, chapter~23, pages 461--482. Springer-Verlag New York, 2013.

\bibitem{DBLP:journals/mst/SchaeferS17}
Marcus Schaefer and Daniel \v{S}tefankovi\v{c}.
\newblock Fixed points, {N}ash equilibria, and the existential theory of the
  reals.
\newblock {\em Theory of Computing Systems}, 60(2):172--193, 2017.

\bibitem{DBLP:journals/mlq/SchuchardtH95}
Dietmar Schuchardt and Hans{-}Dietrich Hecker.
\newblock Two {NP}-hard art-gallery problems for ortho-polygons.
\newblock {\em Mathematical Logic Quarterly}, 41:261--267, 1995.

\bibitem{shermer1992recent}
Thomas~C. Shermer.
\newblock Recent results in art galleries.
\newblock {\em Proceedings of the IEEE}, 80(9):1384--1399, 1992.

\bibitem{shitov2016universality}
Yaroslav Shitov.
\newblock A universality theorem for nonnegative matrix factorizations.
\newblock {\em Preprint, \url{https://arxiv.org/abs/1606.09068}}, 2016.

\bibitem{Shitov16a}
Yaroslav Shitov.
\newblock The complexity of positive semidefinite matrix factorization.
\newblock {\em {SIAM} Journal on Optimization}, 27(3):1898--1909, 2017.

\bibitem{shor1991stretchability}
Peter~W. Shor.
\newblock Stretchability of pseudolines is {NP}-hard.
\newblock In Peter Gritzmann and Bernd Sturmfels, editors, {\em Applied
  Geometry and Discrete Mathematics: The Victor Klee Festschrift}, volume~4 of
  {\em DIMACS -- Series in Discrete Mathematics and Theoretical Computer
  Science}, pages 531--554. American Mathematical Society and Association for
  Computing Machinery, 1991.

\bibitem{tomas2013guarding}
Ana~Paula Tom{\'a}s.
\newblock Guarding thin orthogonal polygons is hard.
\newblock In {\em Fundamentals of Computation Theory}, pages 305--316.
  Springer, 2013.

\bibitem{urrutia2000art}
Jorge Urrutia.
\newblock Art gallery and illumination problems.
\newblock In J.-R. Sack and J.~Urrutia, editors, {\em Handbook of Computational
  Geometry}, chapter~22, pages 973--1027. Elsevier, 2000.

\end{thebibliography}

\end{document}